\documentclass{article}

\usepackage{arxiv}
\usepackage{enumerate}
\usepackage[utf8]{inputenc} 
\usepackage[T1]{fontenc}    
\usepackage[hidelinks]{hyperref}       
\usepackage{url}            
\usepackage{booktabs}       
\usepackage{amsfonts}       
\usepackage{nicefrac}       
\usepackage{microtype}      
\usepackage{graphicx}
\usepackage{natbib}
\usepackage{doi}
\usepackage{amsmath,amsthm,amssymb}
\usepackage{bm}
\usepackage{xcolor}
\usepackage{algpseudocode}
\usepackage{algorithm}
\usepackage{caption}
\usepackage{subcaption}
\usepackage{comment}
\usepackage{dsfont}
\usepackage{soul}

\newcommand{\R}{\mathbb{R}}

\newcommand{\N}{\mathbb{N}}

\newcommand{\NN}{\mathcal{N}}
\newcommand{\MN}{\mathcal{MN}}
\newcommand{\ME}{\mathcal{ME}}
\newcommand{\norm}[1]{\left\lVert#1\right\rVert}
\newcommand{\abs}[1]{\left|#1\right|}
\newcommand{\xhat}[1]{\hat{\bm{x}}^{#1}}
\newcommand{\xhatj}[2]{\hat{x}_{#2}^{#1}}

\DeclareMathOperator\md{MD}
\DeclareMathOperator\mmd{MMD}

\DeclareMathOperator\diag{diag}
\DeclareMathOperator\tr{tr}
\DeclareMathOperator\vect{vec}
\DeclareMathOperator\cov{cov}
\DeclareMathOperator\old{old}
\DeclareMathOperator\new{new}
\DeclareMathOperator\row{row}
\DeclareMathOperator\col{col}

\DeclareMathOperator\vol{vol}
\DeclareMathOperator\fix{fix}
\DeclareMathOperator\mix{mix}
\DeclareMathOperator\rnd{rnd}
\DeclareMathOperator\mle{MLE}
\DeclareMathOperator\mcd{MCD}
\DeclareMathOperator\mmcd{MMCD}
\DeclareMathOperator\pds{PDS}
\DeclareMathOperator\rank{rank}
\DeclareMathOperator\M{M}

\DeclareMathOperator*{\argmin}{arg\,min}

\newtheorem{definition}{Definition}[subsection]

\newtheorem{remark}{Remark}[subsection]

\theoremstyle{plain}
\newtheorem{theorem}{Theorem}[subsection]
\newtheorem{proposition}{Proposition}[subsection]
\newtheorem{lemma}{Lemma}[subsection]

\title{Robust covariance estimation and explainable outlier detection for matrix-valued data}

\author{\href{https://orcid.org/0000-0002-3430-8308}{\includegraphics[scale=0.06]{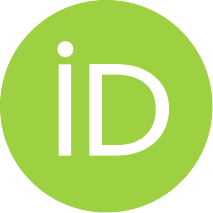}\hspace{1mm}Marcus Mayrhofer}\\
	TU Wien\\
	\texttt{marcus.mayrhofer@tuwien.ac.at} \\
        \And
	\href{https://orcid.org/0000-0003-0329-0595}{\includegraphics[scale=0.06]{orcid.pdf}\hspace{1mm}Una Radojičić} \\
	TU Wien\\
	\texttt{una.radojicic@tuwien.ac.at} \\
	\And
	\href{https://orcid.org/0000-0002-8014-4682}{\includegraphics[scale=0.06]{orcid.pdf}\hspace{1mm}Peter Filzmoser} \\
	TU Wien\\
	\texttt{peter.filzmoser@tuwien.ac.at} \\        
}

\hypersetup{
pdftitle={Robust covariance estimation and explainable outlier detection for matrix-valued data},
pdfsubject={stat.ME},
pdfauthor={Marcus~Mayrhofer, Una~Radojičić, Peter~Filzmoser},
pdfkeywords={Matrix-variate distributions, Covariance with Kronecker structure, Explainable artificial intelligence, Image data, Minimum covariance determinant, Shapley values},
}

\begin{document}
\maketitle

\begin{abstract}
This work introduces the Matrix Minimum Covariance Determinant (MMCD) method, a novel robust location and covariance estimation procedure designed for data that are naturally represented in the form of a matrix. Unlike standard robust multivariate estimators, which would only be applicable after a vectorization of the matrix-variate samples leading to high-dimensional datasets, the MMCD estimators account for the matrix-variate data structure and consistently estimate the mean matrix, as well as the rowwise and columnwise covariance matrices in the class of matrix-variate elliptical distributions. Additionally, we show that the MMCD estimators are matrix affine equivariant and achieve a higher breakdown point than the maximal achievable one by any multivariate, affine equivariant location/covariance estimator when applied to the vectorized data. An efficient algorithm with convergence guarantees is proposed and implemented. As a result, robust Mahalanobis distances based on MMCD estimators offer a reliable tool for outlier detection. Additionally, we extend the concept of Shapley values for outlier explanation to the matrix-variate setting, enabling the decomposition of the squared Mahalanobis distances into contributions of the rows, columns, or individual cells of matrix-valued observations. Notably, both the theoretical guarantees and simulations show that the MMCD estimators outperform robust estimators based on vectorized observations, offering better computational efficiency and improved robustness. Moreover, real-world data examples demonstrate the practical relevance of the MMCD estimators and the resulting robust Shapley values.
\end{abstract}

\keywords{Matrix-variate distributions \and Covariance with Kronecker structure \and Explainable artificial intelligence \and Image data \and Minimum covariance determinant \and Shapley values}

\section{Introduction}\label{sec:Introduction}
Thanks to modern data collection tools, the amount and complexity of available information are increasing rapidly, and matrix-valued data are often observed. Compared to classical multivariate observations, where values for $p$ variables are recorded for one subject, matrix-valued observations are recorded on a grid of $p \times q$ variables. These are then naturally represented as a matrix with $p$ rows and $q$ columns. Some examples include image data, where $p$ and $q$ are given by the resolution of the image, or multivariate data measured on $p$ variables, where the measurements for a subject are available for $q$ replications (e.g., different time points, different spatial locations, different experimental conditions, etc.). Frequently, matrix-valued data are analyzed as classical multivariate data by stacking the matrix columns (or rows) to a vector of length $p\cdot q$. Thus, if $n$ observations are available, the data are arranged in a matrix of dimension $n \times pq$. Depending on the dimensions, this can create high-dimensional data, possibly with a sample size lower than the resulting dimensionality, which constitutes a limitation for multivariate statistical methods. 

As an alternative to vectorizing matrix-valued observations, we model them under the assumption that they originate from a certain matrix-variate distribution. As in the multivariate setting, the class of matrix-elliptical distributions~\citep{gupta2012elliptically}, serves as a natural ground for studying covariance estimation. The matrix-elliptical family is a semi-parametric class of distributions parametrized by the mean $\bm{\M} \in \R^{p \times q}$, row covariance $\bm{\Sigma}^{\row} \in \pds(p)$, column covariance $\bm{\Sigma}^{\col} \in \pds(q)$, and the so-called density generator function $g: [0,\infty) \to \R$.  
Here, $\pds(a)$, with $a \in \N$, denotes the class of all positive definite symmetric $a \times a$ matrices.
More specifically, a random matrix $\bm{X}$ with an absolutely continuous distribution has an elliptical distribution, denoted  $\ME(\bm{\M}, \bm{\Sigma}^{\row}, \bm{\Sigma}^{\col},g)$, if its density can be written as
\begin{align}
    f(\bm{X}) = \det(\bm{\Sigma}^{\row})^{\nicefrac{-q}{2}}\det(\bm{\Sigma}^{\col})^{\nicefrac{-p}{2}} g(\tr(\bm{\Omega}^{\col}(\bm{X}-\bm{\M})'\bm{\Omega}^{\row}(\bm{X}-\bm{\M}))), \label{eq:matrix_elliptical}
\end{align}
with $\bm{\Omega}^{\row} = (\bm{\Sigma}^{\row})^{-1}$ and $\bm{\Omega}^{\col} = (\bm{\Sigma}^{\col})^{-1}$ denoting the precision matrices among the rows and columns, respectively. Matrix elliptical distributions can also be related to their multivariate counterparts. Formally, a random matrix $\bm{X}$ follows a matrix elliptical distribution $\ME(\bm{\M}, \bm{\Sigma}^{\row}, \bm{\Sigma}^{\col},g)$ if and only if its vectorized version $\vect{\bm{X}}$ follows a multivariate elliptical distribution $\mathcal{E}(\vect(\bm{\M}),\bm{\Sigma}^{\col} \otimes \bm{\Sigma}^{\row},g)$ \citep{gupta2012elliptically}. Here, $\vect(\cdot)$ is the vectorization operator, stacking the columns of a matrix on top of each other, $\otimes$ is the Kronecker product.
Probably the most studied matrix elliptical distribution is the matrix normal distribution \citep{dawid1981matrix}, denoted $\MN(\bm{\M}, \bm{\Sigma}^{\row}, \bm{\Sigma}^{\col})$, with density
\begin{align}
    f(\bm{X}|\bm{\M}, \bm{\Sigma}^{\row}, \bm{\Sigma}^{\col}) = \frac{\exp(-\frac{1}{2}\tr(\bm{\Omega}^{\col} (\bm{X}-\bm{\M})' \bm{\Omega}^{\row}(\bm{X}-\bm{\M})))}{(2\pi)^{\nicefrac{pq}{2}} \det(\bm{\Sigma}^{\col})^{\nicefrac{p}{2}} \det(\bm{\Sigma}^{\row})^{\nicefrac{q}{2}}}. \label{eq:matrix_normal_density}
\end{align}

Regarding the estimation of location and covariance for an i.i.d. sample $\bm{\mathfrak{X}} = (\bm{X}_1,\ldots,\bm{X}_n) \in \R^{n \times p \times q}$, with $\bm{X}_i \sim \MN(\bm{\M}, \bm{\Sigma}^{\row}, \bm{\Sigma}^{\col})$, we can either work with the vectorized observations or directly with the matrices. 
In the former setting, the existence and uniqueness of the maximum likelihood estimator (MLE) for the covariance is guaranteed almost surely if $n \geq pq + 1$. However, this approach does not take advantage of the Kronecker structure of the covariance matrix and instead directly estimates the entire $pq$-dimensional matrix $\bm{\Sigma}$. 
In contrast, if we utilize the knowledge of the inherent data structure, we only need to estimate the $p$-dimensional rowwise covariance matrix $\bm{\Sigma}^{\row}$ and the $q$-dimensional columnwise covariance matrix $\bm{\Sigma}^{\col}$. For the matrix-variate sample $\bm{\mathfrak{X}}$, the MLEs for the mean, as well as for the rowwise and columnwise covariance, are given by \citep{Dutilleul1999}:
\begin{align}
    \hat{\bm{\M}} &= \frac{1}{n}\sum_{i = 1}^n \bm{X}_i \label{eq:matrix_mean_mle}\\
    \hat{\bm{\Sigma}}^{\row} &= \frac{1}{qn} \sum_{i = 1}^n (\bm{X}_i - \hat{\bm{\M}})\hat{\bm{\Omega}}^{\col} (\bm{X}_i - \hat{\bm{\M}})'\label{eq:matrix_cov_row_mle}\\
    \hat{\bm{\Sigma}}^{\col} &= \frac{1}{pn} \sum_{i = 1}^n (\bm{X}_i - \hat{\bm{\M}})'\hat{\bm{\Omega}}^{\row} (\bm{X}_i - \hat{\bm{\M}})\label{eq:matrix_cov_col_mle}
\end{align}
\cite{soloveychik2016gaussian} showed that for $n$ i.i.d. samples from a continuous $p \times q$ matrix-variate distribution, there exists no unique maximum of the matrix normal likelihood function if $n < \max(\nicefrac{p}{q}, \nicefrac{q}{p}) + 1$, and that a unique maximum exists almost surely if $n \geq \lfloor \nicefrac{p}{q} + \nicefrac{q}{p} \rfloor + 2$.
Although there are no closed-form solutions for the maximum likelihood estimates (MLEs) of $\bm{\Sigma}^{\row}$ and $\bm{\Sigma}^{\col}$, \cite{Dutilleul1999} proposed an iterative estimation procedure. 
The idea of the so-called \emph{flip-flop} algorithm is to alternate between the computation of $\hat{\bm{\Sigma}}^{\row}$ and $\hat{\bm{\Sigma}}^{\col}$ based on Equations~\eqref{eq:matrix_cov_row_mle} and \eqref{eq:matrix_cov_col_mle}, respectively, until a convergence criterion is met. The algorithm is constructed such that positive definite estimates of subsequent iterations are nondecreasing in likelihood \citep{lu2005likelihood}, and it converges almost surely to the unique maximum from any symmetric positive definite initialization of either $\hat{\bm{\Sigma}}^{\row}$ or $\hat{\bm{\Sigma}}^{\col}$, if $n \geq \lfloor \nicefrac{p}{q} + \nicefrac{q}{p} \rfloor + 2$ \citep{soloveychik2016gaussian}. 

Existing proposals for robust covariance estimation include a generalization of Tyler's M-estimator \citep{tyler1987distribution} introduced by \cite{soloveychik2016gaussian}, a robust estimator for structured covariance matrices with Kronecker structure as a particular case \citep{sun2016robust}, distribution-free robust covariance estimation \citep{zhang2022covariance}, and ML estimation for the matrix t-distribution \citep{thompson2020classification}. 

We propose novel robust estimators for the parameters $\bm{\M}$, $\bm{\Sigma}^{\row}$, and $\bm{\Sigma}^{\col}$, termed the matrix minimum covariance determinant (MMCD) estimators. These estimators generalize the minimum covariance determinant (MCD) approach \citep{Rousseeuw1985}, one of the most widely used approaches for robustly estimating the mean and covariance of multivariate (vector-valued) data. We show that the MMCD estimators are equivariant under matrix affine transformations and surpass the maximal attainable breakdown point of any multivariate, affine equivariant location/covariance estimator when applied to the vectorized data, such as the MCD estimator.
Additionally, we show that the MMCD estimators are consistent for the finite-dimensional parameters $(\bm{\M},\bm{\Sigma}^{\row},\bm{\Sigma}^{\col})$ of the matrix elliptical distribution,  thus bridging a gap between the individual, distribution-specific, estimators in the elliptical family.
Furthermore, a concentration step (C-step) algorithm is developed to efficiently compute the MMCD estimators; see \cite{Rousseeuw1999} for more details on C-step for MCD. Additionally, we introduce a reweighting step that preserves the properties of the MMCD estimators and greatly increases finite-sample efficiency.  

The robust MMCD estimators can then be employed for outlier detection using the Mahalanobis distances~\citep{mahalanobis1936} for matrix-valued observations. Because it is essential to understand the reasons for the outlyingness, we extend the concept of Shapley values introduced in \cite{mayrhofer2022} for outlier explanation in the multivariate case to the matrix-variate setting. Shapley values \citep{shapley1953} are well-known from explainable AI \citep{lundberg2017}, but their computation is usually time-consuming. Our proposal is computationally efficient, and the resulting Shapley values preserve their attractive properties~\citep{shapley1953}. 

The paper is organized as follows. In Section~\ref{section:MMCD}, we introduce the MMCD estimators, then proceed to derive their theoretical properties in Section~\ref{sec:section 3}. Section~\ref{section:MMCD_cstep} is devoted to computational details for the MMCD estimators. 
In Section~\ref{section:shapley}, we propose Shapley values for outlier explanation and present their properties.
In Sections~\ref{section:simulations} and~\ref{section:examples}, we illustrate the performance of the proposed methods on numerical simulations and real-world examples. Section~\ref{section:conclusion} concludes our findings.
The supplementary materials contain more information on the theoretical background in this context, proofs, technical derivations, code, and additional numerical results. 

\section{The MMCD estimators}
\label{section:MMCD}
The MLEs given in Equations~\eqref{eq:matrix_mean_mle}-\eqref{eq:matrix_cov_col_mle}, much like the multivariate normal MLEs, i.e. sample mean and covariance, also serve as valid (consistent) parameter estimators in the class of elliptical distributions; see Remark \ref{rem:MLE_consistency}. However, just like their multivariate counterparts, these are not robust against outlying observations. In order to obtain robust estimators for the finite-dimensional parameters $(\bm{\M},\bm{\Sigma}^{\row}, \bm{\Sigma}^{\col})$ in Equation~\eqref{eq:matrix_elliptical}, we optimize the weighted version of the matrix-normal (log-)likelihood function.
This principle has been similarly used in the context of other robust estimators~\citep[e.g.,][]{neykov2007,garcia2010,kurnaz2018}, and in particular, \cite{Raymaekers2022} show that the MCD estimator can be reformulated in terms of likelihood; the objective of the MCD estimator is to identify the subset of $h$ out of $n$ samples ($\nicefrac{n}{2} \leq h \leq n$) with the smallest determinant of the sample covariance matrix. This is equivalent to determining a subset of size $h$ that maximizes the multivariate normal (log-)likelihood function.

Extending the concept of the multivariate MCD approach, we introduce weights $\bm{w} = (w_1,\dots,w_n) \in \R^{n}$ for a given sample $\bm{\mathfrak{X}} = (\bm{X}_1,\ldots ,\bm{X}_n)$ that is independently drawn from $\MN(\bm{\M}, \bm{\Sigma}^{\row}, \bm{\Sigma}^{\col})$ to formulate the weighted log-likelihood function $l(\bm{w},\bm{\M}, \bm{\Sigma}^{\row}, \bm{\Sigma}^{\col}|\bm{\mathfrak{X}})$ as
\begin{align} 
    \label{eq:loglik}
    -\frac{1}{2} \sum_{i=1}^n w_i \Bigl( 
    p\ln(\det(\bm{\Sigma}^{\col})) + q\ln(\det(\bm{\Sigma}^{\row})) + \mmd^2(\bm{X}) + pq\ln(2\pi) \Bigr),
\end{align}
where $\mmd^2(\bm{X})$ denotes the squared matrix Mahlanobis distance defined as
\begin{align}
    \label{eq:MD=MMD}
    \mmd^2(\bm{X})  := \mmd^2(\bm{X};\bm{\M},\bm{\Sigma}^{\row},\bm{\Sigma}^{\col}) = \tr(\bm{\Omega}^{\col}(\bm{X}-\bm{\M})'\bm{\Omega}^{\row}(\bm{X}-\bm{\M})).
\end{align}

Setting $w_i=1$ for all $i=1,\ldots, n$, yields the traditional log-likelihood function, and its maximization yields the MLEs of Equation~\eqref{eq:matrix_mean_mle}-\eqref{eq:matrix_cov_col_mle}.
However, by taking binary weights, $w_i \in \{0,1\}$, with the constraint that $\sum_{i = 1}^n w_i = h$, we see that $n-h$ contributions are trimmed. 
Since contributions from outliers should be trimmed, the task is to identify the subset of regular observations $H \subset \{1,\ldots,n\}$ with $|H|=h$, where $w_i=1$ for $i\in H$ and $0$ otherwise. 
The resulting constrained optimization problem of finding the weighted MLE can be written as
\begin{align} \label{eq:max_loglik}
    \begin{split}
        &\max_{\bm{w},\bm{\M}, \bm{\Sigma}^{\row}, \bm{\Sigma}^{\col}} l(\bm{w},\bm{\M}, \bm{\Sigma}^{\row}, \bm{\Sigma}^{\col}|\bm{\mathfrak{X}}) \\
        & \text{subject to} \quad w_i \in \{0,1\} \text{ for all } i = 1,\dots,n \quad \text{and} \quad \sum_{i = 1}^n w_i = h.
    \end{split}    
\end{align}

To improve clarity, we will use the following notation for subsamples of $\bm{\mathfrak{X}}$ and estimators based on it: Let $H \subseteq \{1,\dots,n\}$ be a subset of size $h = \abs{H}$, then $\bm{\mathfrak{X}}_{H} := (\bm{X}_i)_{i \in H}$ denotes an $h$-subset of $\bm{\mathfrak{X}}$. An estimator for a parameter $\theta$ based on the sample $\bm{\mathfrak{X}}$ is denoted as $\hat{\theta}_{\bm{\mathfrak{X}}}$ or simply as $\hat{\theta}$ if it is clear on which sample the estimator is computed. Similarly, if an estimator is based on an $h$-subset, it is denoted as $\hat{\theta}_{H}$ or as $\hat{\theta}_{\bm{\mathfrak{X}}_{H}}$.

\begin{proposition}
\label{proposition:MCD_likelihood}
Let $\bm{\mathfrak{X}} = (\bm{X}_1,\dots,\bm{X}_n)$, $\nicefrac{n}{2} \leq h \leq n$ and $h \geq \lfloor \nicefrac{p}{q} + \nicefrac{q}{p} \rfloor + 2$, be an i.i.d. sample from $\MN(\bm{\M}, \bm{\Sigma}^{\row}, \bm{\Sigma}^{\col})$. Maximizing the weighted log-likelihood function~\eqref{eq:max_loglik} is equivalent to minimizing
\begin{align}
    \ln(\det(\hat{\bm{\Sigma}}^{\col}_{H} \otimes \hat{\bm{\Sigma}}^{\row}_{H})) = p\ln(\det(\hat{\bm{\Sigma}}^{\col}_{H})) + q\ln(\det(\hat{\bm{\Sigma}}^{\row}_{H}))
    \label{eq:MMCD_cov_determinant}
\end{align}
across all subsets $H \subset \{1,\dots,n\}$ with $|H| = h$. In Equation~\eqref{eq:MMCD_cov_determinant}, 
\begin{align}
    \hat{\bm{\M}}_{H} &= \frac{1}{h}\sum_{i \in H} \bm{X}_i,  \label{eq:MCD_Mean}\\
    \hat{\bm{\Sigma}}^{\row}_{H} &= \frac{1}{qh} \sum_{i \in H} (\bm{X}_i - \hat{\bm{\M}}_{H})\hat{\bm{\Omega}}^{\col}_{H} (\bm{X}_i - \hat{\bm{\M}}_{H})' \label{eq:MCD_Sigma_r} \text{, and}\\
    \hat{\bm{\Sigma}}^{\col}_{H} &= \frac{1}{ph} \sum_{i \in H} (\bm{X}_i - \hat{\bm{\M}}_{H})'\hat{\bm{\Omega}}^{\row}_{H} (\bm{X}_i - \hat{\bm{\M}}_{H})  \label{eq:MCD_Sigma_c}
\end{align}
denote the MLEs based on the observations in $H$, and $\hat{\bm{\Omega}}^{\row}_{H}=(\hat{\bm{\Sigma}}^{\row}_{H})^{-1}$
and $\hat{\bm{\Omega}}^{\col}_{H}=(\hat{\bm{\Sigma}}^{\col}_{H})^{-1}$ denote the corresponding precision matrices.
\end{proposition}
A proof is given in Supplement~\ref{supplement:mmcd_proofs}. Based on this proposition, we obtain a matrix-variate counterpart to the multivariate MCD estimator's objective, resulting in robust estimators of the parameters $\bm{\M}$, $\bm{\Sigma}^{\row}$, and $\bm{\Sigma}^{\col}$. 
\begin{definition} \label{definition:mmcd}
    Let $\bm{\mathfrak{X}} = (\bm{X}_1,\dots,\bm{X}_n)$, $\nicefrac{n}{2} \leq h \leq n$ and $h \geq \lfloor \nicefrac{p}{q} + \nicefrac{q}{p} \rfloor + 2$, be an i.i.d sample of a continuous $p \times q$ matrix-variate distribution. The \emph{raw} matrix minimum covariance determinant (MMCD) estimators are defined as
    \begin{align}
    \label{eq:MMCD}        
    (\hat{\bm{\M}}_{H^\ast},\hat{\bm{\Sigma}}^{\row}_{H^\ast}, \hat{\bm{\Sigma}}^{\col}_{H^\ast}) := \argmin_{\substack{\hat{\bm{\M}}_{H},\hat{\bm{\Sigma}}^{\row}_{H}, \hat{\bm{\Sigma}}^{\col}_{H}\\ H \subset \{1,\dots,n\}, \abs{H} = h}} p\ln(\det(\hat{\bm{\Sigma}}^{\col}_{H})) + q\ln(\det(\hat{\bm{\Sigma}}^{\row}_{H})),
    \end{align}
    with $\hat{\bm{\M}}_{H}$, $\hat{\bm{\Sigma}}^{\row}_{H}$, and $\hat{\bm{\Sigma}}^{\col}_{H}$ as in Equations~\eqref{eq:MCD_Mean}, \eqref{eq:MCD_Sigma_r}, and \eqref{eq:MCD_Sigma_c}, respectively. 
\end{definition}
The estimators in Definition~\eqref{definition:mmcd}  almost surely exist and are positive definite if $h \geq \lfloor \nicefrac{p}{q} + \nicefrac{q}{p} \rfloor + 2$~\citep{soloveychik2016gaussian}.
If $p = 1$ and/or $q = 1$, optimization problem~\eqref{eq:MMCD} coincides with the optimization problem of the MCD estimator, and one obtains the univariate or multivariate MCD estimator, respectively.

\section{Properties of the MMCD estimators} \label{sec:section 3}

\paragraph{Matrix affine equivariance.}
The concept of affine equivariance in multivariate analysis is rooted in the idea that the estimators used for location and covariance should transform in the same way as the parameters of elliptically symmetrical unimodal distributions (referred to as elliptical distributions hereafter), see \cite{maronna2019robust}.
We can define the matrix-variate analog of affine equivariance based on the properties of matrix-variate elliptical distributions, which are frequently employed to study the robustness properties of normal theory under nonnormal situations \citep{gupta1999}. 

Linear functions of a random matrix $\bm{X} \sim \ME(\bm{\M}, \bm{\Sigma}^{\row}, \bm{\Sigma}^{\col},g)$ also have an elliptical distribution \citep{gupta2012elliptically}. This means that for constant matrices $\bm{A} \in \R^{r \times p}$, $\rank(\bm{A}) = r \leq p$, $\bm{B} \in \R^{q \times s}$, $\rank(\bm{B}) = s \leq q$, and $\bm{C} \in \R^{r \times s}$, the transformed random matrix $\bm{Z} = \bm{A}\bm{X}\bm{B} + \bm{C}$ has density 
\begin{align} 
    \label{eq:elliptical_density_transformation}
    \bm{Z} \sim \ME(\bm{A}\bm{\M}\bm{B} + \bm{C}, \bm{A}\bm{\Sigma}^{\row}\bm{A}', \bm{B}'\bm{\Sigma}^{\col}\bm{B},g).
\end{align}
Let $\hat{\bm{\M}}_{\bm{\mathfrak{X}}}$, $\hat{\bm{\Sigma}}^{\row}_{\bm{\mathfrak{X}}}$, and $\hat{\bm{\Sigma}}^{\col}_{\bm{\mathfrak{X}}}$ denote the estimators based on a sample $\bm{\mathfrak{X}} = (\bm{X}_1,\dots,\bm{X}_n)$ generated by $f(\bm{\M}, \bm{\Sigma}^{\row}, \bm{\Sigma}^{\col})$. Then the estimators of the sample $\bm{\mathfrak{Z}} = (\bm{A}\bm{X}_1\bm{B} + \bm{C},\dots,\bm{A}\bm{X}_n\bm{B} + \bm{C})$ should transform in the same way as the parameters in Equation~\eqref{eq:elliptical_density_transformation}, i.e., 
\begin{align}\label{eq:matrix_affine_equivariance}
    \hat{\bm{\M}}_{\bm{\mathfrak{Z}}} = \bm{A}\hat{\bm{\M}}_{\bm{\mathfrak{X}}}\bm{B} + \bm{C}, \quad 
    \hat{\bm{\Sigma}}^{\row}_{\bm{\mathfrak{Z}}} = \bm{A}\hat{\bm{\Sigma}}^{\row}_{\bm{\mathfrak{X}}}\bm{A}', \quad 
    \hat{\bm{\Sigma}}^{\col}_{\bm{\mathfrak{Z}}} = \bm{B}'\hat{\bm{\Sigma}}^{\col}_{\bm{\mathfrak{X}}}\bm{B}.
\end{align}
Properties~\eqref{eq:matrix_affine_equivariance} provide a suitable generalization of affine equivariance to the matrix-variate setting, and it is easy to verify that they hold for the estimators given in \eqref{eq:matrix_mean_mle}-\eqref{eq:matrix_cov_col_mle}.
However, they do not imply affine equivariance of the location and covariance estimators for the vectorized observations. This would only hold for transformations with the Kronecker structure $\vect(\bm{A}\bm{X}\bm{B} + \bm{C}) = (\bm{B}' \otimes \bm{A})\vect(\bm{X}) + \vect(\bm{C})$. 
We refer to Properties~\eqref{eq:matrix_affine_equivariance} as \emph{matrix affine equivariance} to avoid confounding definitions. 

\begin{lemma} \label{lemma:kronecker_affine_transformation}
    Let $\bm{\mathfrak{X}} = (\bm{X}_1,\dots,\bm{X}_n)$ be a sample of $p \times q$ matrices, where $\bm{X}_i \sim \ME(\bm{\M}_{\bm{\mathfrak{X}}}, \bm{\Sigma}^{\row}_{\bm{\mathfrak{X}}}, \bm{\Sigma}^{\col}_{\bm{\mathfrak{X}}},g)$, and let $\bm{\mathfrak{Z}} = (\bm{Z}_1,\dots,\bm{Z}_n)$ be the affine transformation of $\bm{\mathfrak{X}}$, i.e.,  $\bm{Z}_i = \bm{A}\bm{X}_i\bm{B} + \bm{C}$, $\bm{A} \in \R^{p \times p}$, $\bm{B} \in \R^{q \times q}$, $\bm{A}, \bm{B}$ invertible, and $\bm{C} \in \R^{p \times q}$. The following then holds:
    \begin{enumerate}\renewcommand{\labelenumi}{(\alph{enumi})}
        \item The MMCD estimators as in Definition~\ref{definition:mmcd}
        are matrix affine equivariant. 
        \item $\mmd^2(\bm{Z}_i;\hat{\bm{\M}}_{\bm{\mathfrak{Z}}}, \hat{\bm{\Sigma}}^{\row}_{\bm{\mathfrak{Z}}}, \hat{\bm{\Sigma}}^{\col}_{\bm{\mathfrak{Z}}}) = \mmd^2(\bm{X}_i;\hat{\bm{\M}}_{\bm{\mathfrak{X}}}, \hat{\bm{\Sigma}}^{\row}_{\bm{\mathfrak{X}}}, \hat{\bm{\Sigma}}^{\col}_{\bm{\mathfrak{X}}})$, where $(\hat{\bm{\M}}_{\bm{\mathfrak{Z}}}, \hat{\bm{\Sigma}}^{\row}_{\bm{\mathfrak{Z}}}, \hat{\bm{\Sigma}}^{\col}_{\bm{\mathfrak{Z}}})$ are matrix affine equivariant location and covariance estimators of the transformed sample $\bm{\mathfrak{Z}}$. 
    \end{enumerate}
\end{lemma}
Lemma~\ref{lemma:kronecker_affine_transformation} shows that the MMCD estimators are equivariant under matrix affine transformations, and a proof is given in Supplement~\ref{supplement:mmcd_properties_proofs}. 

\paragraph{Breakdown point.}
The finite sample breakdown point of an estimator evaluates its resilience to contamination. It refers to the largest proportion of observations that may be arbitrarily replaced by outliers such that the estimator still contains some information about the true parameter \citep{maronna2019robust}. Let $\bm{\mathfrak{X}}$ be a sample of $n$ matrix-variate observations in $\R^{p \times q}$ and suppose $\bm{\mathfrak{Y}}$ is a corrupted version, obtained by replacing $m$ samples of $\bm{\mathfrak{X}}$ by arbitrary matrices. The finite sample breakdown point of a location estimator $\hat{\bm{\M}}$ is given by
\begin{align}\label{eq:breakdown_point_location}
    \varepsilon^{\ast}(\hat{\bm{\M}}, \bm{\mathfrak{X}}) = 
    \min_{1 \leq m \leq n} \left\{ \frac{m}{n} : \sup_{m} \norm{\hat{\bm{\M}}_{\bm{\mathfrak{X}}} - \hat{\bm{\M}}_{\bm{\mathfrak{Y}}}} = \infty \right\}
\end{align}
and the (joint) finite sample breakdown point of row and columnwise covariance estimators $\hat{\bm{\Sigma}}^{\row}$ and $\hat{\bm{\Sigma}}^{\col}$ is given by
\begin{align}\label{eq:breakdown_point_covariance}
    \varepsilon^{\ast}(\hat{\bm{\Sigma}}^{\row}, \hat{\bm{\Sigma}}^{\col}, \bm{\mathfrak{X}}) = 
    \min_{1 \leq m \leq n} \left\{ \frac{m}{n} : \sup_{m} \max_{i,j} \abs{ \log(\lambda_i(\hat{\bm{\Sigma}}^{\row}_{\bm{\mathfrak{Y}}})\lambda_j(\hat{\bm{\Sigma}}^{\col}_{\bm{\mathfrak{Y}}})) - \log(\lambda_i(\hat{\bm{\Sigma}}^{\row}_{\bm{\mathfrak{X}}})\lambda_j(\hat{\bm{\Sigma}}^{\col}_{\bm{\mathfrak{X}}}))} = \infty \right\}.
\end{align}

While the MCD and the MMCD estimators coincide for the case that $p = 1$ and/or $q = 1$, the following theorem shows that the MMCD estimators achieve a higher breakdown point than the MCD estimators applied to the vectorized samples if $p \geq 2$ and $q \geq 2$.

\begin{theorem} \label{theorem:MMCD_breakdown_point}
Let $\bm{\mathfrak{X}}$ be a collection of $n$ i.i.d. samples from a continuous $p \times q$ matrix-variate distribution, where $d = \lfloor \nicefrac{p}{q} + \nicefrac{q}{p} \rfloor$, $p,q \in \N, p \geq 2, q \geq 2$, and let $\hat{\bm{\M}}$, $\hat{\bm{\Sigma}}^{\row}$, and $\hat{\bm{\Sigma}}^{\col}$ denote the MMCD estimators, then
\begin{align*}
    \varepsilon^{\ast}(\hat{\bm{\M}}, \bm{\mathfrak{X}}) = \varepsilon^{\ast}(\hat{\bm{\Sigma}}^{\row}, \hat{\bm{\Sigma}}^{\col}, \bm{\mathfrak{X}}) = \frac{1}{n}\lfloor\min(n - h + 1,h - (d+1))\rfloor =: \frac{m}{n},
\end{align*}
with $\nicefrac{n}{2} \leq h \leq n$ and $h \geq d + 2$. 
\end{theorem}
The proof extends established methodologies from \cite{Rousseeuw1985} and \cite{lopuhaa1991breakdown} to address the matrix variate setting, leveraging additional insights and techniques outlined in Supplement~\ref{supplement:mmcd_properties_proofs}.
Theorem~\ref{theorem:MMCD_breakdown_point} implies that the maximum breakdown point of the MMCD estimators is $\nicefrac{1}{n}\lfloor \nicefrac{(n-d)}{2} \rfloor$ and is attained if $h = \lfloor\nicefrac{(n+d+2)}{2}\rfloor$. This means that the maximum breakdown point of the MMCD covariance estimators for $p \geq 2, q \geq 2$ is higher than the upper bound for the breakdown point of affine equivariant covariance estimators applied to vectorized samples, which is given by $\nicefrac{1}{n}\lfloor \nicefrac{(n-pq+1)}{2}\rfloor$ \citep{davies1987asymptotic,lopuhaa1991breakdown}. However, as mentioned earlier, affine equivariance in the matrix-variate setting does not imply affine equivariance in the multivariate setting. Thus, the mentioned upper bound for the vectorized observations does not apply. In other words, since affine equivariance (in the vectorized case) is not a requirement for matrix-variate affine equivariance, it is possible to achieve a higher breakdown point for the MMCD estimators than for any affine equivariant multivariate estimator applied to the vectorized data. 

To illustrate the advantage of respecting the inherent data structure of matrix-variate data for the breakdown properties, we compare the maximum breakdown points of the MCD and MMCD estimators in Figure~\ref{fig:breakdown_point} for different combinations of $p$ and $q$, and for different sample sizes $n$. Here, the MCD estimator is applied to the vectorized data, and the dimensionality of the samples is $pq$, which can get large. 
This affects the computability of the MCD estimator since it requires a subset size larger than the dimension. 

\begin{figure}[!ht]
    \centering
    \includegraphics[width=\linewidth]{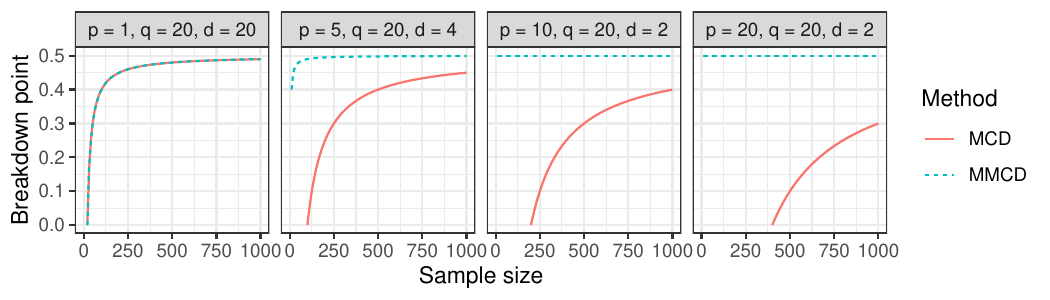}
    \caption{Comparison of the maximum breakdown point of the MMCD estimators for matrix-variate data with $p = 1, 5, 10, 20$ rows and $q = 20$ columns, and the MCD estimator applied to the vectorized data. When $p = 1$, both estimators and their breakdown points coincide. However, increasing the number of rows yields better breakdown properties for the MMCD estimators, as the proportion between the number of rows and columns $d = \lfloor \nicefrac{p}{q} + \nicefrac{q}{p} \rfloor$ is approaching 2. 
    }
    \label{fig:breakdown_point}
\end{figure}

\paragraph{Consistency for elliptical distributions.}
Let us now consider the asymptotic behavior of the MMCD estimators. By scaling the rowwise or columnwise MMCD covariance estimator by a distribution-specific consistency factor, we can achieve consistency for elliptical distributions. 
\begin{theorem}\label{theorem:consistency}
Let $\bm{X}_1,\dots,\bm{X}_n$ be a random sample from an elliptical matrix-variate distribution $\mathcal{ME}(\bm{\M},\bm{\Sigma}^{\row},\bm{\Sigma}^{\row},g)$ with positive definite covariances $\bm{\Sigma}^{\row},\bm{\Sigma}^{\col}$, and let $(\hat{\bm{\M}},\hat{\bm{\Sigma}}^{\row},\hat{\bm{\Sigma}}^{\col})$ be the corresponding MMCD estimators. Then, it holds that 
$$
\left\|\hat{\bm{\M}}-\bm{\M}\right\| \xrightarrow{a.s.} 0, \quad \left\|{c(\alpha)\hat{\bm{\Sigma}}^{\col}\otimes\hat{\bm{\Sigma}}^{\row}-\bm{\Sigma}^{\col}\otimes\bm{\Sigma}^{\row}}\right\| \xrightarrow{a.s.} 0,
$$
where $c(\alpha), \alpha = \nicefrac{h}{n} \in [0.5,1],$ is a distribution-specific consistency factor as in \cite{croux1999influence}. 
\end{theorem}
The proof of the consistency of the \emph{raw} MMCD estimators relies on the strong consistency of the MCD estimator given in \cite{butler1993asymptotics,cator2012central} and is provided in Supplement~\ref{supplement:mmcd_properties_proofs}. It shows that the consistency factor of the MCD estimator and the MMCD estimator must coincide, and therefore, we use the consistency factor
\begin{align} \label{eq:consistency_factor}
    c(\alpha) = \frac{\alpha}{F_{\chi^2_{pq+2}}(\chi^2_{\alpha; pq})}
\end{align}
proposed by \cite{croux1999influence} to obtain consistency at the normal model, where $F_{\chi^2_{pq+2}}$ denotes the CDF of the chi-square distribution with $pq+2$ degrees of freedom, and $\chi^2_{pq;\alpha}$ denotes the $\alpha$ quantile of the chi-square distribution with $pq$ degrees of freedom.
\begin{remark}\label{rem:MLE_consistency}
Note first that for $h=n$, the corresponding MMCD estimators coincide with the ones defined in \eqref{eq:matrix_mean_mle}-\eqref{eq:matrix_cov_col_mle}. Therefore, a simple, yet not discussed in the literature, consequence of Theorem \ref{theorem:consistency} is that the estimators obtained maximizing the likelihood under the matrix-normal model, are consistent estimators of the corresponding finite-dimensional parameters $(\bm{\M},\bm{\Sigma}^{\row},\bm{\Sigma}^{\row})$ in the semi-parametric, matrix elliptical family.    
\end{remark}
\paragraph{Reweighted MMCD - improving efficiency.}
The \emph{raw} MMCD estimators are most robust when about half of the observations are trimmed, i.e., $h = \lfloor\nicefrac{(n+d+2)}{2}\rfloor$. However, this leads to a low efficiency at the normal model. While efficiency could be increased by trimming fewer samples, this would lead to lower robustness.
To enhance a robust estimator's efficiency without compromising robustness, \cite{lopuhaa1991breakdown, maronna2019robust} proposed a one-step reweighing procedure. We can apply this technique for the MMCD estimators by defining weighted ML estimators with weights depending on the Mahalanobis distances given the raw MMCD estimators.
 
\begin{definition} \label{definition:mmcd_reweighted}
    Let $\bm{\mathfrak{X}}$ be a collection of $n$ i.i.d. samples from a continuous $p \times q$ matrix-variate distribution, where $d = \lfloor \nicefrac{p}{q} + \nicefrac{q}{p} \rfloor$, $p,q \in \N, p \geq 2, q \geq 2$, and let $\hat{\bm{\M}}$, $\hat{\bm{\Sigma}}^{\row}$, and $\hat{\bm{\Sigma}}^{\col}$ denote the \emph{raw} MMCD estimators as in Definition~\ref{definition:mmcd}. The \emph{reweighted} MMCD estimators are given by
    \begin{align}
        \tilde{\bm{\M}} &= \frac{1}{\sum_{i = 1}^n w(\mmd(\bm{X}_i))} \sum_{i = 1}^n w(\mmd(\bm{X}_i)) \bm{X}_i,  \label{eq:MCD_reweighted_Mean}\\
        \tilde{\bm{\Sigma}}^{\row} &= \frac{1}{q \sum_{i = 1}^n w(\mmd(\bm{X}_i))} \sum_{i = 1}^n w(\mmd(\bm{X}_i))  (\bm{X}_i - \tilde{\bm{\M}})\tilde{\bm{\Omega}}^{\col} (\bm{X}_i - \tilde{\bm{\M}})' \label{eq:MCD_reweighted_Sigma_r} \text{, and}\\
        \tilde{\bm{\Sigma}}^{\col} &= \frac{1}{p \sum_{i = 1}^n w(\mmd(\bm{X}_i))} \sum_{i = 1}^n w(\mmd(\bm{X}_i))  (\bm{X}_i - \tilde{\bm{\M}})'\tilde{\bm{\Omega}}^{\row} (\bm{X}_i - \tilde{\bm{\M}})  \label{eq:MCD_reweighted_Sigma_c},
    \end{align}
    where $w : [0,\infty) \rightarrow [0,\infty)$ is a non-increasing and bounded weight function such that $w(\mmd(\bm{X}_i)) > 0$ for at least $\lfloor\nicefrac{(n+d+2)}{2}\rfloor$ observations that vanishes for large distances, i.e., $w(\mmd(\bm{X}_i)) = 0$ if $\mmd(\bm{X}_i) > c_1 > 0$. 
\end{definition}

The following theorem shows that the \emph{reweighted} MMCD estimator preserves the breakdown point of the original estimator. The simulations presented in Section~\ref{section:simulations} illustrate substantial improvements in the efficiency of the reweighed MMCD estimators. With increasing sample size, the finite sample efficiency exceeds $90\%$ across various selections of $p$ and $q$.

\begin{theorem} \label{theorem:MMCD_reweighted_breakdown_point}
    Let $\bm{\mathfrak{X}}$ be a collection of $n$ i.i.d. samples from a continuous $p \times q$ matrix-variate distribution, where $d = \lfloor \nicefrac{p}{q} + \nicefrac{q}{p} \rfloor$, $p,q \in \N, p \geq 2, q \geq 2$, and let $\hat{\bm{\M}}_{\bm{\mathfrak{X}}}$, $\hat{\bm{\Sigma}}^{\row}_{\bm{\mathfrak{X}}}$, and $\hat{\bm{\Sigma}}^{\col}_{\bm{\mathfrak{X}}}$ denote the \emph{raw} MMCD estimators as in Definition~\ref{definition:mmcd} with breakdown points 
    \begin{align*}
        \varepsilon^{\ast}(\hat{\bm{\M}}_{\bm{\mathfrak{X}}}, \bm{\mathfrak{X}}) = \varepsilon^{\ast}(\hat{\bm{\Sigma}}^{\row}_{\bm{\mathfrak{X}}}, \hat{\bm{\Sigma}}^{\col}_{\bm{\mathfrak{X}}}, \bm{\mathfrak{X}}) = \frac{1}{n}\lfloor\min(n - h + 1,h - (d+1))\rfloor  =: \frac{m}{n},
    \end{align*}
    and let $\tilde{\bm{\M}}_{\bm{\mathfrak{X}}}$, $\tilde{\bm{\Sigma}}^{\row}_{\bm{\mathfrak{X}}}$, and $\tilde{\bm{\Sigma}}^{\col}_{\bm{\mathfrak{X}}}$ denote the \emph{reweighted} estimators as in Definition~\ref{definition:mmcd_reweighted}. Then, 
    \begin{align*}
        \varepsilon^{\ast}(\tilde{\bm{\M}}_{\bm{\mathfrak{X}}}, \bm{\mathfrak{X}}) \geq \frac{m}{n} \quad \text{and} \quad 
        \varepsilon^{\ast}(\tilde{\bm{\Sigma}}^{\row}_{\bm{\mathfrak{X}}}, \tilde{\bm{\Sigma}}^{\col}_{\bm{\mathfrak{X}}}, \bm{\mathfrak{X}}) \geq \frac{m}{n}.
    \end{align*}
\end{theorem}

A proof is given in Supplement~\ref{supplement:mmcd_properties_proofs}. For the algorithm used to compute the reweighted MMCD estimators introduced in the following section, we use the weight function $w: [0,\infty) \mapsto \{0,1\}$ with 
\begin{align} \label{eq:mmcd_weights}
    w(\mmd^2(\bm{X}_i)) := \begin{cases}
        1 & \text{if} \quad i \in H \vee \mmd^2(\bm{X}_i) < \chi^2_{pq;0.975}\\
        0 & \text{otherwise}
    \end{cases}.
\end{align}
Note that the $h$ observations in the $h$-subset of the raw MMCD estimator have the lowest MMDs, and the condition that all observations $i \in H$ get a positive weight ensures that the reweighting step does not lead to an estimator that uses fewer than $h$ samples.

\section{Algorithm}
\label{section:MMCD_cstep}
\cite{Rousseeuw1999} proposed the Fast-MCD algorithm
to efficiently compute the MCD estimator. The key idea to find the $h$-subset with the lowest covariance determinant is based on the concentration step (C-step): after each C-step, the objective function is smaller or equal as before, and by repeatedly applying C-steps convergence is reached within finitely many iterations.

\subsection{Adapting the C-step}

Adapting the structure of the Fast-MCD algorithm to the matrix-variate setting leads to the development of the MMCD algorithm. This adaptation necessitates a modification in the covariance estimation during the C-step to derive suitable counterparts for computing the MMCD estimators. However, this process encounters a challenge due to the involvement of two covariance matrices, as depicted in Equations~\eqref{eq:MCD_Sigma_r} and \eqref{eq:MCD_Sigma_c}, both lacking closed-form solutions for their estimation.
To address this issue, we incorporate the flip-flop algorithm introduced by \cite{Dutilleul1999} within the C-step. 
Consider a matrix-variate random sample $\bm{\mathfrak{X}} = (\bm{X}_1,\dots,\bm{X}_n)$, with $\bm{X}_i \in \R^{p \times q}$, and any $h$-subset $H_{\old} \subset \{1,\ldots,n\}$, with $\abs{H_{\old}} = h > \lfloor \nicefrac{p}{q} + \nicefrac{q}{p} \rfloor + 2$. First, the MLEs $(\hat{\bm{\M}}_{H_{\old}}, \hat{\bm{\Sigma}}^{\row}_{H_{\old}}, \hat{\bm{\Sigma}}^{\col}_{H_{\old}})$ are computed based on the observations in the subset $H_{\old}$ using the flip-flop algorithm, which is non-decreasing in likelihood. 
Next, compute the squared Mahalanobis distances $d^2_i(H_{\old}) := \mmd^2(\bm{X}_i,\hat{\bm{\M}}_{H_{\old}}, \hat{\bm{\Sigma}}^{\row}_{H_{\old}}, \hat{\bm{\Sigma}}^{\col}_{H_{\old}})$ for all $i=1,\ldots ,n$.
In Proposition~\ref{proposition:MCD_likelihood}, we showed that $\sum_{i \in H_{\old}} d^2_i(H_{\old}) = hpq$, hence only the terms of the log-likelihood function involving the determinants change in this step.
To construct the new subset $H_{\new}$, sort the squared MMDs in ascending order, resulting in a permutation $\pi$ of $\{1,\ldots,n\}$ such that $d^2_{\pi(1)}(H_{\old}) \leq \ldots \leq d^2_{\pi(n)}(H_{\old})$, and define a new $h$-subset $H_{\new} = \{\pi(1),\ldots,\pi(h)\}$.
Since the estimators do not change in this step, the terms involving the determinant are constant, and by construction, the sum of the Mahalanobis distances either decreases or stays constant. 
Hence, the reordering is non-decreasing in likelihood. Finally, the estimators are updated using the flip-flop algorithm based on the observations in the subset $H_{\new}$, resulting in estimators $(\hat{\bm{\M}}_{H_{\new}}, \hat{\bm{\Sigma}}^{\row}_{H_{\new}}, \hat{\bm{\Sigma}}^{\col}_{H_{\new}})$, 
increasing the likelihood once more, and it follows that 
\begin{equation}
    \label{eq:MMCD_C-step}
    p\ln(\det(\hat{\bm{\Sigma}}^{\col}_{H_{\new}})) + q\ln(\det(\hat{\bm{\Sigma}}^{\row}_{H_{\new}})) \leq
    p\ln(\det(\hat{\bm{\Sigma}}^{\col}_{H_{\old}})) + q\ln(\det(\hat{\bm{\Sigma}}^{\row}_{H_{\old}})).
\end{equation}
By repeatedly applying such C-steps, we can decrease the covariance determinant in subsequent iterations as in Equation~\eqref{eq:MMCD_C-step}. This results in a decreasing and non-negative sequence of determinants that must converge after exploring finitely many $h$-subsets. Similar to the multivariate case, we obtain equality of the determinants from one $h$-subset to the next if and only if the estimators do not change from one to the next iteration. However, this does not necessarily imply that we have found a global optimum. 
A pseudo-code for this matrix-variate version of the C-step is given in Algorithm~\ref{algorithm:MMCD_C_step} in Supplement~\ref{supplement:mmcd_algorithm}.

\subsection{The MMCD algorithm}

The MMCD algorithm is a matrix-variate extension of the Fast-MCD procedure of~\cite{Rousseeuw1999}, aiming to alleviate the C-steps dependence on the initial subset by using multiple initial subsets, iteratively conducting C-steps on each until convergence, and ultimately selecting the solution with the lowest determinant. While this explains the idea of the algorithm, there are more computational considerations and adjustments in the full MMCD algorithm. A pseudo-code of the MMCD Algorithm~\ref{algorithm:MMCD} is given in Supplement~\ref{supplement:mmcd_algorithm}.

As in its multivariate counterpart, the MMCD procedure uses so-called \emph{elemental} subsets to initialize the procedure. This means that we use $m$ subsets of size $d + 2, d = \lfloor \nicefrac{p}{q} + \nicefrac{q}{p} \rfloor$, instead of size $h$, to increase the probability of obtaining at least one clean initial subset.
Using $m = 500$ elemental subsets by default allows for a reasonable tradeoff between a wide variety of settings where we likely obtain at least one clean elemental subset and the computational demands of computing initial estimators. If either $p \ll q$ or $q \gg p$, $d$ will be large, and using more initial subsets is recommended. Using elemental subsets increases not only the robustness of the initial estimators but also the computational efficiency. 

Moreover, the MMCD procedure only uses 2 C-step and MLE iterations for the initial elemental subsets to ensure even faster computation of the initial estimators. In the MLE procedure, \cite{werner2008estimation} demonstrated that the same asymptotic efficiency can be attained using only two iterations instead of iterating until convergence. 
As for the C-step, \cite{Rousseeuw1999} outlined that after two iterations, subsets with the lowest covariance determinant during the procedure can already be identified, even before reaching convergence. Moreover, simulations show that we can identify those initial subsets that yield robust solutions after 2 C-steps, whether we use 2 MLE iterations or iterate the flip-flop algorithm until convergence. This is described in detail in Supplement~\ref{supplement:mmcd_subsets}, where we also show that elemental subsets indeed yield more robust solutions than their larger counterparts in case of high contamination.

The initialization step of the MMCD procedure yields $m$ initial estimators, and we keep the $10$ estimators with the lowest covariance determinant. Using those as initial estimators, we iterate C-steps until convergence on the complete data set $\bm{\mathfrak{X}}$. The solution with the lowest covariance determinant then yields the raw MMCD estimators. 

The raw MMCD estimators are scaled using the consistency factor $c(\alpha)$ given in Equation~\eqref{eq:consistency_factor} to achieve consistency at the normal model as outlined in Theorem~\ref{theorem:consistency}. Based on those rescaled raw MMCD estimators, the reweighted estimators described in Definition~\ref{definition:mmcd_reweighted} are computed using the weights given in Equation~\eqref{eq:mmcd_weights}.
The reweighted MMCD estimators are then scaled using $c(\tilde{\alpha}) = c(\nicefrac{\tilde{h}}{n})$, where $\tilde{h}$ denotes the number of observations with weights one. 

The MMCD algorithm repeatedly computes Mahalanbois distances for all $n$ samples, which is computationally expensive when $n$ gets large. To improve the computational efficiency for settings where $n$ is large, we implemented the subsampling approach proposed by \cite{Rousseeuw1999}. The idea is to split the sample of $n$ observations into several smaller subsamples and compute initial estimators on those subsamples before working on the large set with $n$ observations.

\section{Outlier detection and explainability}
\label{section:shapley}

Given a sample $\bm{\mathfrak{X}} = (\bm{X}_1,\dots,\bm{X}_n)$ of matrix-variate observations, the task for outlier detection is to identify those observations which are ``far away'' from the center of the data cloud with respect to its shape. In robust statistics, it is common to consider the Mahalanobis distance for this purpose, assume an underlying normal distribution of the observations, and use a quantile of the chi-square distribution as an outlier cutoff value~\citep{maronna2019robust}. Here, we follow the same idea: an observation $\bm{X}_i$ is flagged as an outlier if
$$
\mmd^2(\bm{X}_i;\hat{\bm{\M}},\hat{\bm{\Sigma}}^{\row},\hat{\bm{\Sigma}}^{\col})>\chi^2_{pq;0.975} \ ,
$$
for $i \in \{1,\ldots ,n\}$ and the MMCD estimators
$\hat{\bm{\M}}$, $\hat{\bm{\Sigma}}^{\row}$, and $\hat{\bm{\Sigma}}^{\col}$.

Even though this information is valuable in practice, it is not very useful for understanding the reasons for outlyingness. This is the goal of outlier explainability, where the contributions of the cells/rows/columns of the matrix-valued observations are investigated in more detail. We will use the concept of Shapley values for this purpose and first briefly review how this is applied to multivariate data before extending it to the matrix-variate case. For details, we refer to~\cite{mayrhofer2022}.

\subsection{Shapley values for multivariate data}
Let $\bm{x} = (x_1,\ldots ,x_p)'$ denote an observation vector from a population with expectation vector $\bm{\mu} = (\mu_1,\ldots ,\mu_p)'$ and covariance matrix $\bm{\Sigma}$, and $P=\{1,\ldots ,p\}$ the index set of the variables. Then the outlyingness contributions $\bm{\phi}(\bm{x},\bm{\mu},\bm{\Sigma}) = \bm{\phi}(\bm{x}) = (\phi_1(\bm{x}),\ldots, \phi_p(\bm{x}))$ based on the Shapley value assign each variable its average marginal contribution to the squared Mahalanobis distance, i.e.,
\begin{align}
    \phi_k(\bm{x},\bm{\mu},\bm{\Sigma}) 
    = \sum_{S \subseteq P\setminus\{k\}} \frac{\abs{S}!(p-\abs{S}-1)!}{p!} \Delta_k \md^2(\xhat{S})
    = (x_k-\mu_k) \sum_{j =1}^p (x_j-\mu_j) \omega_{jk},
    \label{eq:shapley_md}
\end{align}
with marginal contributions
\begin{equation}
    \Delta_k \md^2(\xhat{S}) := \md^2(\xhat{S\cup\{k\}})-\md^2(\xhat{S}) \quad \text{and} \quad
    \xhatj{S}{j}:= \begin{cases}
        x_j & \text{if } j \in S\\
        \mu_j & \text{if } j \notin S
    \end{cases}
    \label{eq:definexj}
\end{equation}
as the components of $\xhat{S}$. Here, $\omega_{jk}$ is the element $(j,k)$ of $\bm{\Omega} = \bm{\Sigma}^{-1}$.
Since $\bm{\phi}(\bm{x})$ is based on the Shapley value, it is the only decomposition of the squared Mahalanobis distance based on Equation~\eqref{eq:definexj} that fulfills the following properties:
\begin{itemize}
\item{\emph{Efficiency:}}
The contributions $\phi_j(\bm{x})$, for $j = 1,\ldots ,p$, sum up to the squared Mahalanobis distance of $\bm{x}$, hence $\sum_{j = 1}^p \phi_j(\bm{x}) = \md^2(\bm{x})$.
\item{\emph{Symmetry:}}
If $\md^2(\xhat{S \cup \{j\}}) = \md^2(\xhat{S \cup \{k\}})$ holds for all subsets $S \subseteq P\setminus \{j,k\}$
for two coordinates $j$ and $k$, then $\phi_j(\bm{x}) = \phi_k(\bm{x})$.
\item{\emph{Monotonicity:}}
Let $\bm{\mu}, \tilde{\bm{\mu}} \in \R^p$ be two vectors and $\bm{\Sigma}, \tilde{\bm{\Sigma}} \in \pds(p)$ be two matrices. If 
\begin{align*}
        \md_{\bm{\mu},\bm{\Sigma}}^2(\xhat{S \cup \{j\}}) - \md_{\bm{\mu},\bm{\Sigma}}^2(\xhat{S}) \geq \md_{\tilde{\bm{\mu}},\tilde{\bm{\Sigma}}}^2(\xhat{S \cup \{j\}}) - \md_{\tilde{\bm{\mu}},\tilde{\bm{\Sigma}}}^2(\xhat{S})
\end{align*}
holds for all subsets $S \subseteq P$, then $\phi_j(\bm{x},\bm{\mu},\bm{\Sigma}) \geq \phi_j(\bm{x},\tilde{\bm{\mu}},\tilde{\bm{\Sigma}})$.
\end{itemize}

In words, the coordinate $\phi_k(\bm{x})$ of the Shapley value is the average marginal contribution of the $k$-th variable to the squared Mahalanobis distance and is obtained by averaging over all marginal outlyingness contributions $\Delta_k \md^2(\xhat{S})$ across all possible subsets $S \subseteq P \setminus \{k\}$. Although this suggests an exponential computational complexity, which becomes costly, especially if $p$ is large, the second equality in Equation~\eqref{eq:shapley_md} reveals just linear complexity; for a proof we refer to \cite{mayrhofer2022}. Equation~\eqref{eq:shapley_md} allows for another insight into the Shapley value by comparing it to the squared Mahalanobis distance, which can be written as $\sum_{j,k = 1}^p (x_j-\mu_j)(x_k-\mu_k)\omega_{jk}$. While the latter calculates an outlyingness measure by aggregating the contributions $(x_j-\mu_j)(x_k-\mu_k)\omega_{jk}$ of all variables for the entire observation, Equation~\eqref{eq:shapley_md} shows that a coordinate $\phi_k(\bm{x})$ of the Shapley value only considers the contributions that are associated with the $k$-th variable.

\subsection{Shapley value for matrix-valued data}

To define Shapley values for matrix-variate data, we can use the connection between the matrix and multivariate Mahalanobis distance; see Equation~\eqref{eq:MD=MMD}. 
Let $\bm{X} \in \R^{p \times q}$ be a matrix-variate sample with mean $\bm{\M} \in \R^{p \times q}$ and covariance matrices $\bm{\Sigma}^{\row} \in \pds(p)$ and $\bm{\Sigma}^{\col} \in \pds(q)$. The $pq$-dimensional vectorized observation is denoted as $\bm{x} = \vect(\bm{X})$, with mean $\bm{\mu} = \vect(\bm{\M})$ and covariance matrix $\bm{\Sigma} = \bm{\Sigma}^{\col} \otimes \bm{\Sigma}^{\row}$.
Based on Equation~\eqref{eq:shapley_md}, we can obtain outlyingness contributions for every coordinate of $\bm{x}$ and hence for every cell of the matrix $\bm{X}$ by
\begin{align*}
    \phi_{a}(\bm{x}) &= (x_{a} - \mu_{a})\sum_{b = 1}^{pq}(x_{b}-\mu_{b})\omega_{ab}
    = (x_{jk} - m_{jk})\sum_{i = 1}^{p}\sum_{l=1}^{q}(x_{il}-m_{il})\omega_{ij}^{\row}\omega_{kl}^{\col} = \phi_{jk}(\bm{X}),
\end{align*}
with $a = i+(l-1)p$ and $b = j+(k-1)p$, and for 
$j=1,\ldots ,p$ and $k=1,\ldots,q$. Using matrix operations, we can efficiently compute the $p \times q$ matrix containing the cellwise Shapley values $\phi_{jk}(\bm{X})$ as
\begin{align}
    \bm{\Phi}(\bm{X}) &= (\bm{X}-\bm{\M}) \circ \bm{\Omega}^{\row} (\bm{X}-\bm{\M}) \bm{\Omega}^{\col} \in \R^{p \times q}, \label{eq:shapley_cellwise}
\end{align}
where $\circ$ refers to element-wise multiplication. 

Next, we discuss how matrix affine transformations as in Equation~\eqref{eq:matrix_affine_equivariance} affect the cellwise Shapley values for matrix-variate data. 
\begin{proposition} \label{proposition:shapley_properties}
    Let $\bm{X} \in \R^{p \times q}$ be a sample from $\ME(\bm{M},\bm{\Sigma}^{\row},\bm{\Sigma}^{\col},g)$, $\bm{A} \in \R^{p \times p}$, $\bm{B} \in \R^{q \times q}$, $\bm{A}, \bm{B}$ invertible, and $\bm{C} \in \R^{p \times q}$. 
    Then, the cellwise Shapley values are
\emph{not} matrix affine equivariant, i.e., $\bm{\Phi}(\bm{A}\bm{X}\bm{B}) \neq \bm{A}\bm{\Phi}\bm{X})\bm{B}$ for general positive definite $\bm{A}$ and $\bm{B}$. However, they are
    \begin{enumerate}\renewcommand{\labelenumi}{(\alph{enumi})}
        \item shift invariant, i.e., $\bm{\Phi}(\bm{X} + \bm{C}) = \bm{\Phi}(\bm{X})$,
        \item scale invariant, i.e., if $\bm{A}$ and $\bm{B}$ are scaling matrices, 
        thus diagonal matrices with non-zero entries,
        then $\bm{\Phi}(\bm{A}\bm{X}\bm{B}) = \bm{\Phi}(\bm{X})$, 
        \item permutation equivariant, i.e., if $\bm{A}$ and $\bm{B}$ are permutation matrices, then $\bm{\Phi}(\bm{A}\bm{X}\bm{B}) = \bm{A}\bm{\Phi}(\bm{X})\bm{B}$, and
    \end{enumerate}
\end{proposition}
The proofs are given in Supplement~\ref{supplement:shapley_proofs}.
When considering gray-scale image data, shifting or rescaling the gray-scale information would not change the cellwise Shapley values. Further, exchanging rows and columns of the image; in particular mirroring or rotating the image by $90^{\circ}$, would equivalently transform the Shapley values.
Similarly to the setting of cellwise outliers \citep{alqallaf2009propagation}, cellwise Shapley values are tied to the original coordinate system and are not matrix affine equivariant.

It can be preferable in some applications to obtain outlyingness explanations for a complete row or column of the matrix-valued observations, especially when we want to compare multiple observations.
In the following, we show how Shapley values for rows can be obtained; Shapley values for columns can be computed based on the transposed matrix or by adapting the following notation accordingly for columns. 

Consider again the set $P=\{1,\ldots ,p\}$, and $S \subseteq P \setminus \{j\}$.
The rowwise marginal contributions to the matrix Mahalanobis distance are defined as as 
\begin{align*}
    \Delta_j\mmd(\hat{\bm{X}}^{S}) &:= \mmd(\hat{\bm{X}}^{S \cup \{j\}}) - \mmd(\hat{\bm{X}}^{S}) ,
\end{align*}
where the $i$-th row of $\hat{\bm{X}}^{S}$ is given as $(x_{i1},\dots,x_{iq})$ if $i \in S$ and $(m_{i1},\dots,m_{iq})$ if $i \notin S$.
\begin{proposition}\label{proposition:matrix_shapley}
    The $j$-th coordinate of the rowwise Shapley value is given by
    \begin{align}
            \phi_{j.}(\bm{X}) := & \sum_{S \subseteq P \setminus \{j\}} 
            \frac{\abs{S}!(p-\abs{S}-1)!}{p!} \Delta_j\mmd(\hat{\bm{X}}^{S}) \\
            = &\sum_{i = 1}^p\sum_{k = 1}^q\sum_{l = 1}^q (x_{jl}-m_{jl}) (x_{ik}-m_{ik}) \omega_{ij}^{\row}\omega_{kl}^{\col}
            =\sum_{k=1}^q\phi_{jk}(\bm{X}),
        \label{eq:rowwise_shapley}
    \end{align}
\end{proposition}
A proof for Equation~\eqref{eq:rowwise_shapley} can be found in Supplement~\ref{supplement:shapley_proofs}.
Thus, a rowwise Shapley value is obtained by summing up the cellwise Shapley values for the corresponding row, which is equivalent to adapting the marginal contributions to a rowwise replacement.
The vectors containing the rowwise or columnwise Shapley values can also be computed by 
\begin{align}
    \bm{\phi}_{\row}(\bm{X}) &=  \diag(\bm{\Omega}^{\row} (\bm{X}-\bm{\M}) \bm{\Omega}^{\col} (\bm{X}-\bm{\M})') \in \R^{p} \label{eq:shapley_rowwise}  \quad \text{and}\\
    \bm{\phi}_{\col}(\bm{X}) &= \diag((\bm{X}-\bm{\M})'\bm{\Omega}^{\row} (\bm{X}-\bm{\M}) \bm{\Omega}^{\col}) \in \R^{q} \label{eq:shapley_columnwise},
\end{align}
respectively. The properties listed in Proposition~\ref{proposition:shapley_properties} also apply in this setting.

\section{Simulations}
\label{section:simulations}

In the simulation studies outlined in this section, our primary focus is to rigorously assess the performance of the MMCD estimators.
We aim to validate their demonstrated theoretical properties and compare efficiency against ML estimators.
Despite our initial intention to include various robust estimators as mentioned in Section  \ref{sec:Introduction}, practical constraints arose as the relevant routines were exclusively accessible in \texttt{Matlab}, while our current framework operates within \texttt{R}.
For the sake of consistency and practical implementation, we concentrate on comparing the efficiency of the raw and reweighted MMCD alongside an in-depth analysis of the MLEs, (reweighted) MMCD estimators, and MCD estimator based on the vectorized samples on contaminated data.
To ensure the highest possible breakdown point across all simulations and examples discussed in this paper, we set $h = \lfloor\nicefrac{(n+d+2)}{2}\rfloor$ for the MMCD estimators and $h = \lfloor\nicefrac{(n+pq+1)}{2}\rfloor$ for the MCD estimator. We conduct 100 repetitions for each simulation setting and visualize the results through either line plots or boxplots. In the line plots, the solid lines represent average scores, while the shaded areas depict the one standard error regions.

\paragraph{Finite-sample efficiency.} To analyze the finite-sample efficiency, we generate samples from a centered matrix normal distribution with dimensions $(p,q) \in \{(5,20), (50,20)$, $(100,50)\}$ for various sample sizes $n \in \{20,50,100,300,1000\}$. For the rowwise covariance matrix we adopt the covariance matrices proposed by \cite{agostinelli2015robust}, denoted $\bm{\Sigma}^{\row} = \bm{\Sigma}^{\rnd} \in \pds(p)$, which have random entries and generally yield low correlations.
For the columnwise covariance, we use 
$\bm{\Sigma}^{\col} = \bm{\Sigma}^{\mix}(0.7) \in \pds(q)$, with entries $\sigma_{jk}^{\mix}(0.7) = 0.7^{\abs{j-k}}$. We assess the normal finite-sample efficiency by comparing the ratio 
$$
\frac{D(\hat{\bm{\Sigma}}^{\row}_{\mle}, \hat{\bm{\Sigma}}^{\col}_{\mle})}{D(\hat{\bm{\Sigma}}^{\row}_{\mmcd}, \hat{\bm{\Sigma}}^{\col}_{\mmcd})},
$$
where $D(\hat{\bm{\Sigma}}^{\row}, \hat{\bm{\Sigma}}^{\col})$ denotes the Kullback-Leiber (KL) divergence of the estimators $\hat{\bm{\Sigma}}^{\row}$ and $\hat{\bm{\Sigma}}^{\col}$ in the matrix normal setting $\MN(\bm{\M}, \bm{\Sigma}^{\row}, \bm{\Sigma}^{\col})$, which is given by 
\begin{align}\label{eq:normal_KL_divergence}
    \begin{split}
        D(\hat{\bm{\Sigma}}^{\row}, \hat{\bm{\Sigma}}^{\col}) 
        =&\tr(\bm{\Omega}^{\row}\hat{\bm{\Sigma}}^{\row}) \tr(\bm{\Omega}^{\col}\hat{\bm{\Sigma}}^{\col})\\
        &- q\log(\det(\bm{\Omega}^{\row}\hat{\bm{\Sigma}}^{\row})) -
        p\log(\det(\bm{\Omega}^{\col}\hat{\bm{\Sigma}}^{\col})) -  pq,
    \end{split}
\end{align}
with $\bm{\Omega}^{\row} = (\bm{\Sigma}^{\row})^{-1}$ and $\bm{\Omega}^{\col} = (\bm{\Sigma}^{\col})^{-1}$.
As shown in Figure~\ref{fig:efficiency}, the efficiency of the raw MMCD estimators is below $0.5$ on average. In contrast, the reweighted estimators' efficiency is above $0.5$ for $n=100$ and it rises to over $0.9$ as the sample size increases.
\begin{figure}[!ht]
    \centering
    \includegraphics[width = \linewidth]{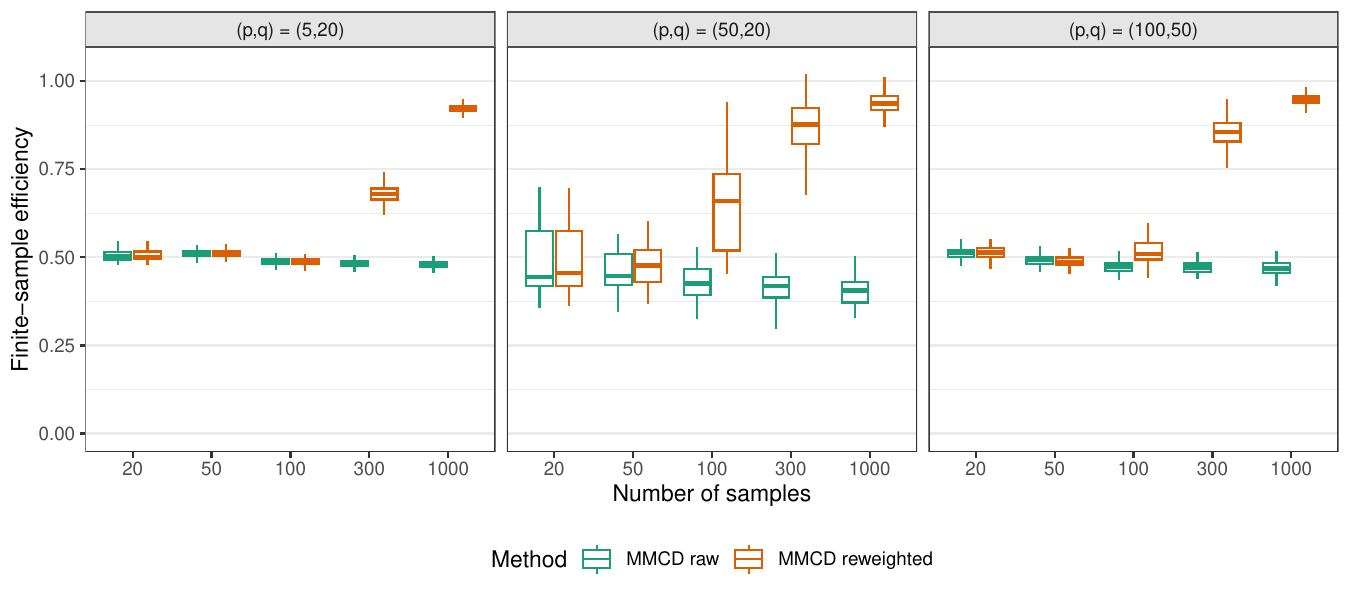}
    \caption{Comparison of the finite-sample efficiency of raw and reweighted MMCD.}
    \label{fig:efficiency}
\end{figure}

\paragraph{Robustness and matrix size.} For the setting with contamination, we consider matrix-variate samples with $p \in \{2,\dots,30\}$ rows and $q = \{10,20,30\}$ columns for sample sizes $n \in \{100,1000\}$. The clean data are generated from a centered matrix normal distribution with $\bm{\Sigma}^{\row} = \bm{\Sigma}^{\rnd}$ and $\bm{\Sigma}^{\col} = \bm{\Sigma}^{\mix}(0.7)$. A fraction, $\varepsilon = 0.1$, of the clean data is replaced by outliers, sampled from a matrix normal distribution with a mean matrix where all entries are equal to $\gamma = 1$. The covariance matrices of the outliers are the same as for the regular observations. 

We use KL divergence \eqref{eq:normal_KL_divergence} to analyze the quality of the covariance estimation. Additionally, we analyze outlier detection capabilities of the squared Mahalanobis distance based on the estimators, with the $\chi^2_{pq,0.99}$ quantile as a detection threshold. We also include the Mahalanobis distances based on true parameters used to generate the data as a benchmark and measure performance by precision and recall. Due to the excessively long computation times of the Fast-MCD procedure in higher-dimensional scenarios, we used the deterministic MCD \citep{hubert2012deterministic} when $pq > 300$. Since the MCD estimator requires $n > pq$, it is only computed for those settings. 
\begin{figure}[!ht]
    \centering
    \includegraphics[width = \linewidth]{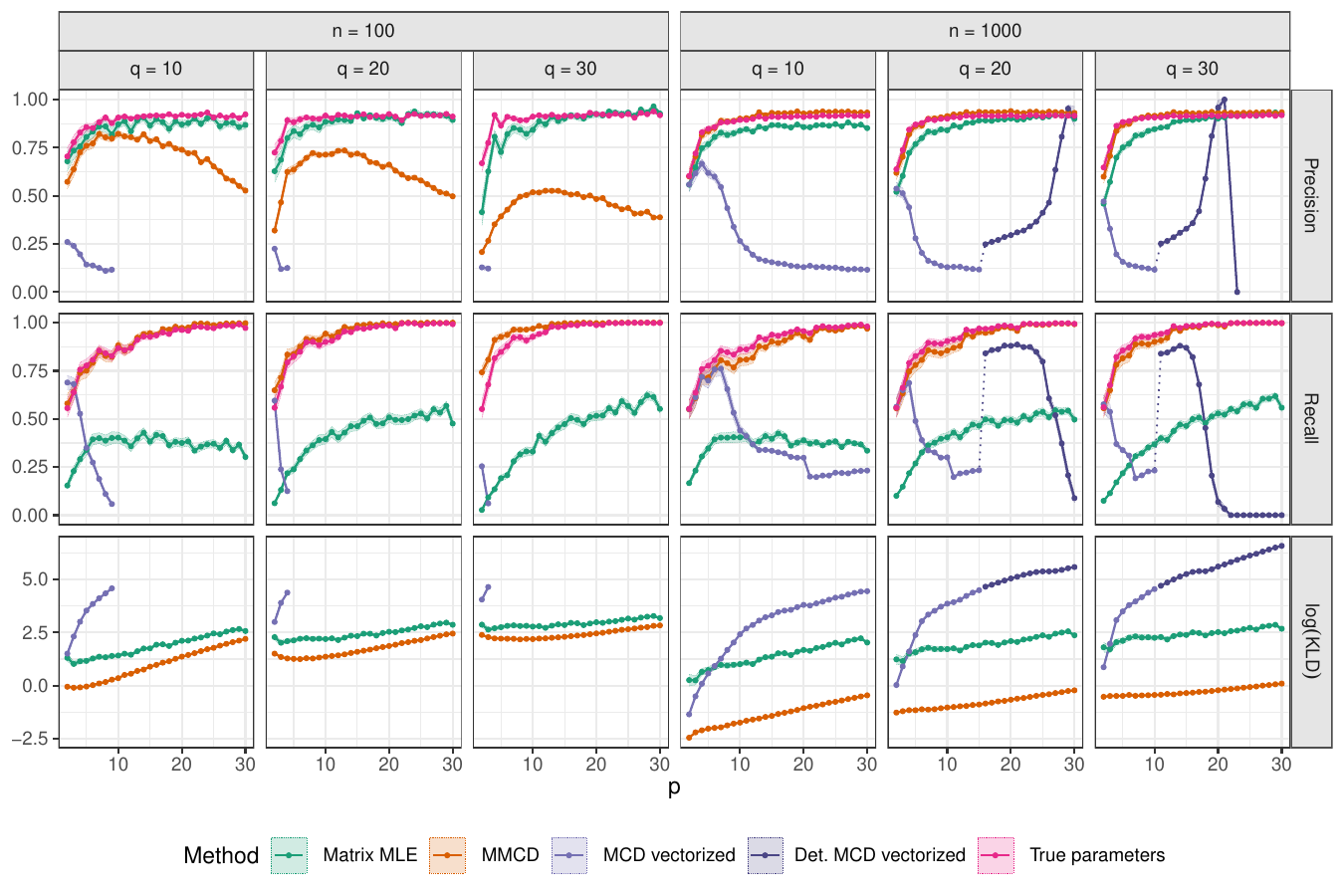}
    \caption{Comparison of precision, recall, and KL divergence for ML and MMCD estimators, (deterministic) MCD estimators with vectorized data, and true parameters as a benchmark for outlier detection for simulated data from a matrix normal distribution with $10\%$ contamination.}
    \label{fig:simulation_pq_ratio_line1}
\end{figure}

Figure~\ref{fig:simulation_pq_ratio_line1} shows that the MMCD estimators have lower KL divergence than the competing methods and attain a recall similar to the benchmark approach based on the true parameters used to generate the data across all settings. The precision of the MMCD estimators depends on the dimensionality of the matrix-variate samples as well as on the sample size. For $n = 100$, the precision decreases with increasing dimensionality $pq$, but the effect is mitigated by an improving performance when $\max\{\nicefrac{p}{q},\nicefrac{q}{p}\}$ is small. 
For $n = 1000$, the precision is close to the precision based on the true parameters. This suggests that for small sample sizes, a correction similar to the one proposed by \cite{pison2002small} for the MCD could lead to a better performance. In the matrix-variate setting, such a correction would not only be dependent on $pq$ and $n$ but also on $\nicefrac{p}{q}$ and $\nicefrac{q}{p}$. 

For small $p$ and $q$, the comparison between the MMCD estimators and the MCD for the vectorized observations is of special interest. For $n = 1000$ and $q = 10$ they have a similar recall when $p \leq 6$, and for $q \in \{20,30\}$ the MCD estimators show substantial improvements when the deterministic MCD approach is used instead of the Fast-MCD. This can be explained by the dependence of the Fast-MCD on the robustness of the initial solutions, and with an increasing $pq$, the probability of obtaining a clean subset becomes very small. The MCD estimator shows a steep drop in precision as $p$ increases when Fast-MCD is used. For the deterministic MCD, we see a trade-off between precision and recall with increasing dimensionality, but the KL divergence remains high. With increasing dimensionality, even the nonrobust matrix MLEs outperform the MCD estimator which highlights the importance of respecting the inherent data structure of matrix-variate observations. 

\paragraph{Robustness and contamination type.} In addition to the shift outliers we also consider block and cell contamination for matrix normal samples of size $(p,q) = (5,20)$. In all three settings, we consider a fraction of $\varepsilon = 0.1$ contaminated samples. Let $\bm{X} = (x_{jk}), j = 1,\dots,5, k = 1,\dots,20$, denote a sample from a centered matrix normal distribution with rowwise covariance $\bm{\Sigma}^{\row} = \bm{\Sigma}^{\rnd}$ and columnwise covariance $\bm{\Sigma}^{\col} = \bm{\Sigma}^{\mix}(0.7)$. For block contamination, we replaced the top left $2 \times 5$ block, corresponding to the entries $x_{jk}, j = 1,2, k = 1,\dots,5$, with entries from a shifted matrix normal distribution with a mean matrix where all entries are equal to $\gamma = 1$ and covariance matrices corresponding to the top left block of $\bm{\Sigma}^{\row}$ and $\bm{\Sigma}^{\col}$. For cell contamination, a fraction of $0.1$ of the cells of the outlying observations are randomly permuted. The shift outliers are generated with a mean shift $\gamma = 1$ as before. 
\begin{figure}[!ht]
    \centering
    \includegraphics[width = 1\linewidth]{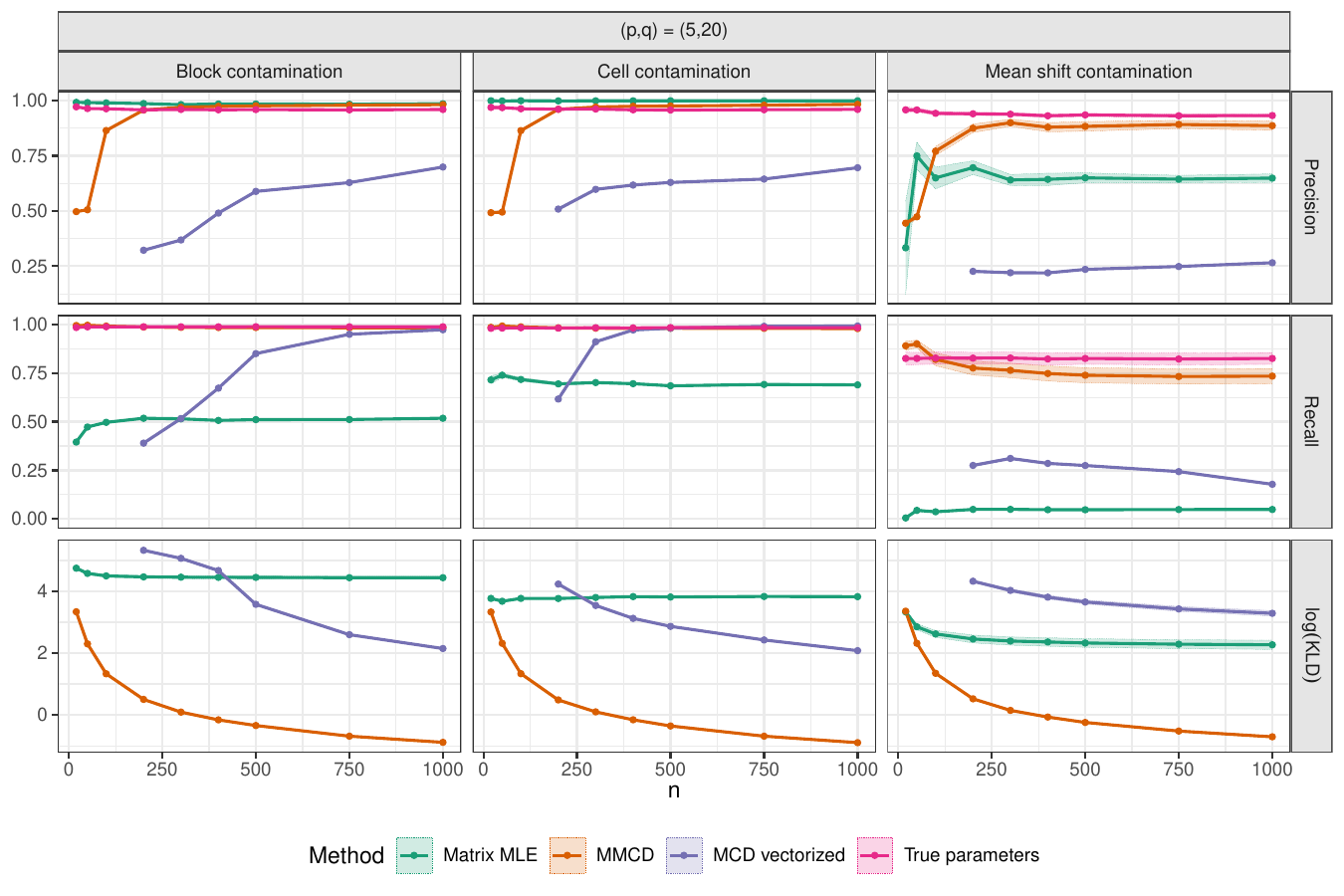}
    \caption{Precicion, recall, and logarithm of KL divergence comparing block, cell, and sample contamination.}
    \label{fig:simulation_contamination_type1}
\end{figure}   

Figure~\ref{fig:simulation_contamination_type1} shows that the MMCD estimators are better suited for outlier detection and yield more robust covariance estimates than the matrix MLEs as well as the MCD estimator on the vectorized observations. Overall, the results are similar across all three simulation scenarios, only for mean shift contamination we see higher variation than in the other two settings. This is likely because the block and cell contamination interfere with covariance estimation more profoundly, i.e., the KL divergence of the matrix MLEs is highest for block contamination followed by cell and shift contamination. 
    
The supplementary materials~\ref{supplement:simulations} provide in-depth simulation studies that expand upon the scenarios discussed in this section. These simulations analyze the effects of the level of contamination and mean shifts for multiple types of covariance matrices. Additionally, we extend our analysis beyond the normal model to include samples generated from a matrix t-distribution, examining performance across a range of degrees of freedom. For this scenario, we also compute the ML estimators for the matrix t-distribution \citep{thompson2020classification}. We include a summary of computation time and consider additional performance metrics, such as the F-score (harmonic mean of precision and recall), Frobenius error, and the angle between eigenvalues of covariance matrices.

\section{Examples}
\label{section:examples}

\subsection{Glacier weather data -- Sonnblick observatory}

We analyze the publicly available weather data from Austria's highest weather station, located in the Austrian Central Alps at an elevation of 3106 m above sea level on top of the glaciated mountain ``Hoher Sonnblick'' (datasource: GeoSphere Austria - \url{https://data.hub.geosphere.at}). 
The observed parameters are monthly averages of temperature (T), precipitation (P), proportion of solid precipitation (SP), air pressure (AP), and sunshine hours (SH). 
We consider the monthly values between 1891 and 2022 and exclude five years with missing values, yielding $n = 127$ observations of $p = 5$ times $q = 12$ dimensional matrices.
Our goal is to identify observations that show a different weather pattern than the majority of the data and explain why the corresponding years deviate from the majority. 
We did not adjust for a possible yearly trend in this exploratory analysis as we wish to understand long-term patterns and shifts in climate without the influence of adjustments. 

\begin{figure}[!ht]
    \centering
    \includegraphics[width=\linewidth]{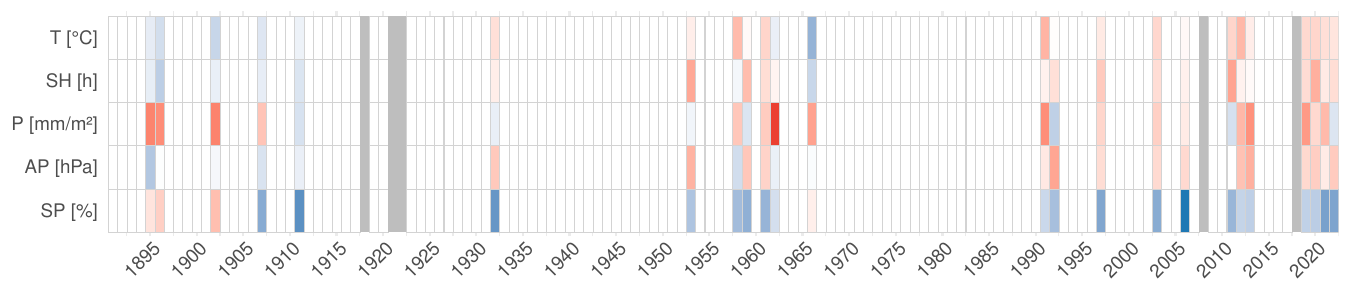}
    \caption{Yearly outlyingness contributions for the glacier weather data. Regular years are white, and years that contain missing data are gray. Outliers are colored as follows: blue for ``above average'', red for ``below average'', and color intensity proportional to rowwise Shapley value.}
    \label{fig:weather_years}
\end{figure}

In total, outlier detection based on the MMCD estimators flags 23 outlying matrices, which are indicated in Figure~\ref{fig:weather_years} as colored years: 
If the aggregated monthly measurements are above their average, the cells are colored red; otherwise, they are colored blue. The rowwise Shapley value is then used to determine color brightness, i.e., the larger the outlyingness contribution, the darker the color. Years with missing observations are grey; years with only white cells refer to regular observations.
It is visible that the outlier frequency increases in the last period. Moreover, more recent outliers are characterized by increased temperature, precipitation, air pressure, and a lack of solid precipitation (e.g. snow) -- a clear signal of a climate change. 
\begin{figure}[!ht]
    \centering
    \includegraphics[width = 1\linewidth]{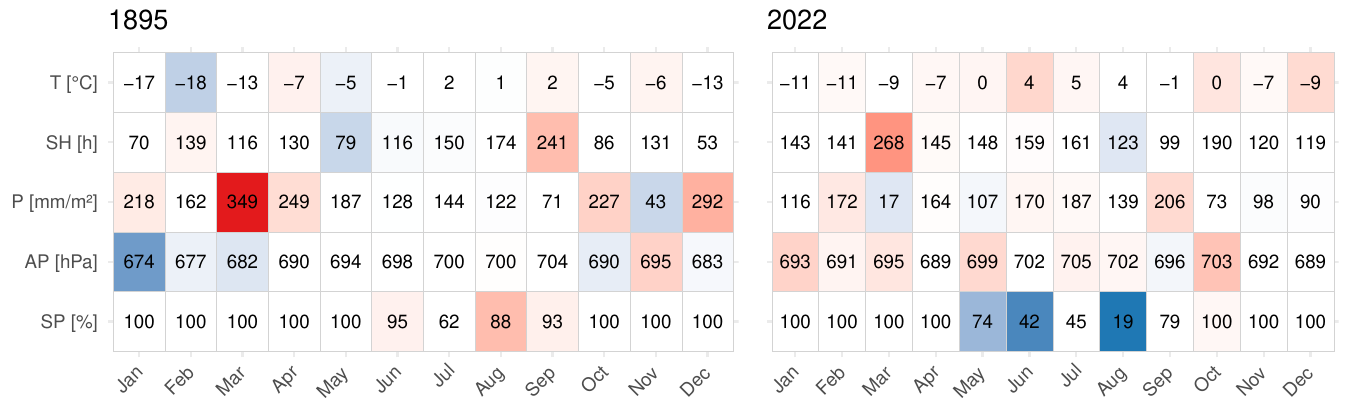}
    \caption{Outlyingess contributions based on cellwise Shapley values for the years 1895 and 2022 of the glacier weather data using the same color scheme as in Figure~\ref{fig:weather_years}.}
    \label{fig:weather_months}
\end{figure}

In Figure~\ref{fig:weather_months}, we use cellwise Shapley values to understand which parameters in which months contributed most to the outlyingness of 1895 and 2022, corresponding to the first and last outlying observation in the data set, respectively, where the color scheme is inherited from Figure \ref{fig:weather_years}. The largest outlyingness contribution is due to an unusually large amount of precipitation in March 1895. Overall, high amounts of precipitation were observed that year, with a high percentage of snow even in the summer months. 
In contrast, the largest outlyingness contributions in 2022 are due to a very sunny March and low percentages of snowfall in May, June, and August.

\subsection{Darwin data}

We consider the DARWIN (Diagnosis AlzheimeR WIth haNdwriting) \citep{cilia2022diagnosing} data set containing handwriting samples of 174 subjects, 89 diagnosed with Alzheimer’s disease (AD), and 85 healthy subjects (H). Each individual completed 25 handwriting tasks on paper, and the pen movements were recorded using a graphic tablet. The tasks are ordered in difficulty.
From the raw handwriting data, 18 features were extracted: Total Time, Air Time, Paper Time, Mean Speed on paper, Mean Speed in air, Mean Acceleration on paper, Mean Acceleration in air, Mean Jerk on paper, Mean Jerk in air, Pressure Mean, Pressure Variance, Generalization of the Mean Relative Tremor (GMRT) on paper, GMTR in air, Mean GMRT, Pendowns Number, Max X Extension, Max Y Extension, and Dispersion Index. For a more detailed description of the data, we refer to \cite{cilia2018experimental}. 
In \cite{cilia2022diagnosing}, each task was considered separately to train a classifier, and the combination of the classifiers led to an improvement in the classification of subjects. Our focus here lies not in the classification task but rather in explaining the differences between AD and H groups. We treat the observations as matrices, with the rows representing the extracted features and the columns representing the tasks.
Because of linear dependencies, the variables Total Time and Mean GMRT were excluded. Further, the variable Air Time had several extreme and unreliable measurements and was thus also excluded. This yields observation matrices with $p = 15$ features and $q=25$ tasks.

We applied the MMCD procedure only on the healthy subjects and used the robust estimators to compute MMDs for all observations. Thus, the MMDs presented in Figure~\ref{fig:darwin_mmd_shv} left are generally smaller for the H group, whereas all observations from the AD group exceed the outlier cutoff value. The fact that healthy subjects also exceed the cutoff value shows the heterogeneity in this group.
\begin{figure}
    \centering
    \includegraphics[width=\linewidth]{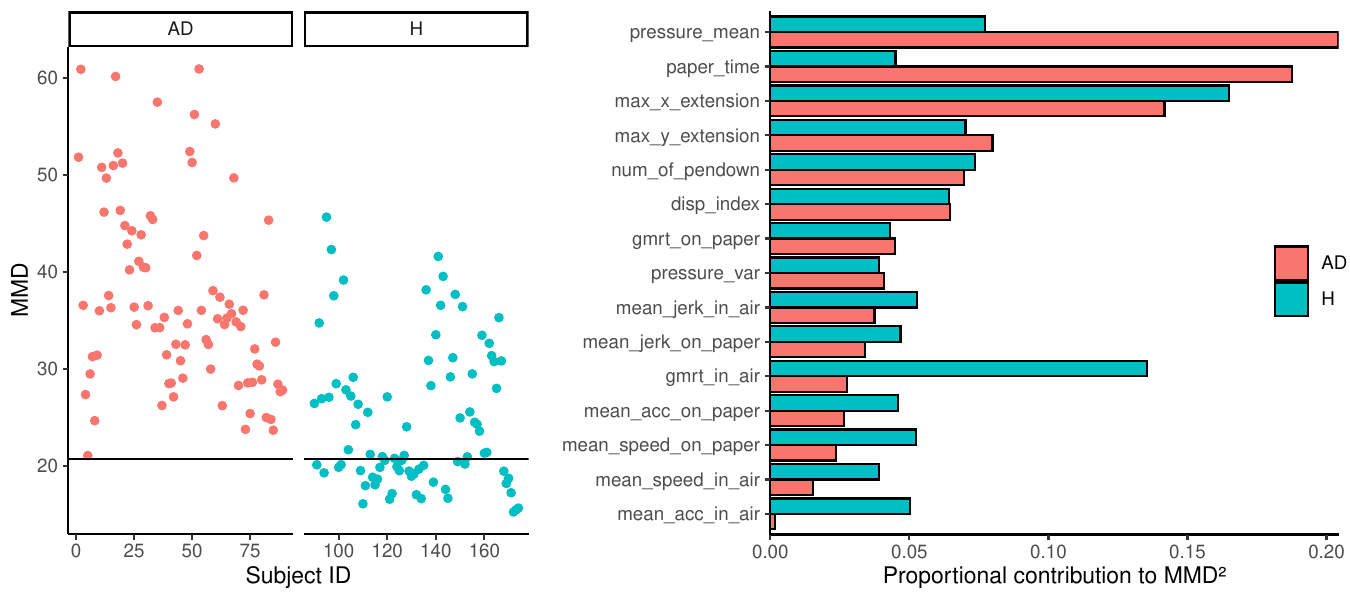}
    \caption{Plot of robust MMD based on MMCD estimators for the Darwin data on the left, and average proportional rowwise Shapley values for the H and AD subjects on the right.}
    \label{fig:darwin_mmd_shv}
\end{figure}
In the right panel of Figure~\ref{fig:darwin_mmd_shv}, we consider the \emph{average proportional contributions} of the variables to the MMDs for the H and AD groups. The outlyingness contributions are based on the rowwise Shapley values, resulting in 15 scores for each individual. Since those scores sum up to the squared MMD, we can divide them by the squared MMD to get proportional contributions, and by averaging over all individuals in the H and AD groups, respectively, we obtain the values shown in this plot. Large differences between the AD and H groups indicate variables that are important to distinguish between healthy individuals and those who have Alzheimer's disease. For example, Pressure Mean and Paper Time are evidently higher in the AD group.

\subsection{Video data}
In this example, we examine a surveillance video of a beach sourced from \cite{li2004statistical}. The video comprises 633 frames, each sized at $128 \times 160$ pixels; five selected frames are shown in Figure~\ref{fig:video_overview}. The majority of the frames depict the beach scene. 
Around frame 500, a man walks into the scene from the left and partly disappears behind the tree. As he continues walking, he reappears on the right side of the tree and remains in the video until the end.

\begin{figure}[!ht]
    \centering
    \includegraphics[width=\linewidth]{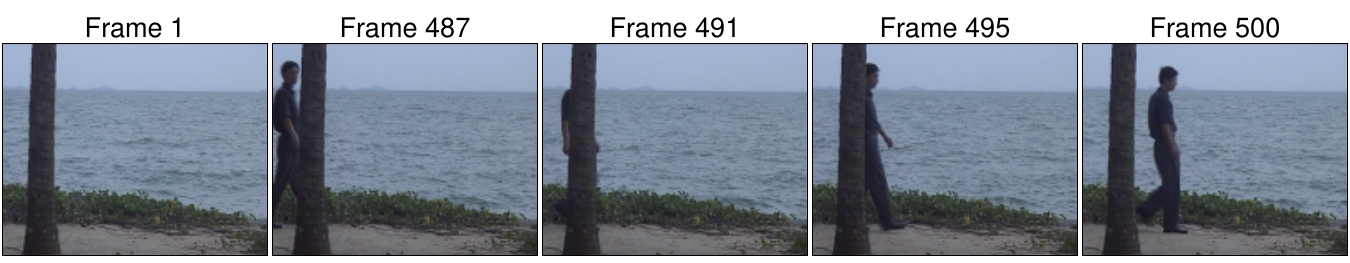}
    \caption{Selected frames of the video data.}
    \label{fig:video_overview}
\end{figure}

For our analyis, we converted the original RGB video to a grayscale video, applied the MMCD procedure, and obtained MMDs for all 633 frames, which are visualized in Figure~\ref{fig:video_md}. The plot on the left shows the robust MMDs for all 633 frames, and the one on the right for frames 471 to 633 to better highlight the increase in MMD when the man enters the scenery, with a short drop in MMD when he disappears behind the tree. We indicate frames 487, 491, and 495, also presented in Figure~\ref{fig:video_shapley} in terms of their cellwise Shapley values. We see that the pixels that form the contours of the man and most of the pixels of the man's head contribute most to the outlyingness. When the man disappears behind the tree, there are fewer pixels with high outlyingness contributions. Since the sum of the contributions amounts to the squared MMD of an observation, this explains the behavior of the MMDs of the frames shown in Figure~\ref{fig:video_md1}. 
\begin{figure}[!ht]
     \centering
     \begin{subfigure}[b]{0.49\textwidth}
        \centering
        \includegraphics[width=\linewidth]{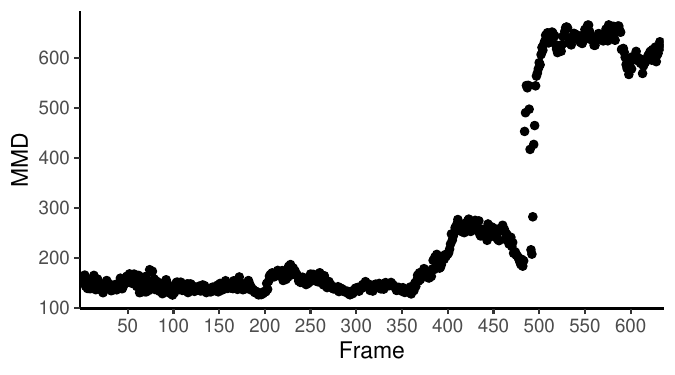}
        \caption{Frames 1 to 633.}
        \label{fig:video_md0}
     \end{subfigure}
     \hfill
    \begin{subfigure}[b]{0.49\textwidth}
        \centering
        \includegraphics[width=\linewidth]{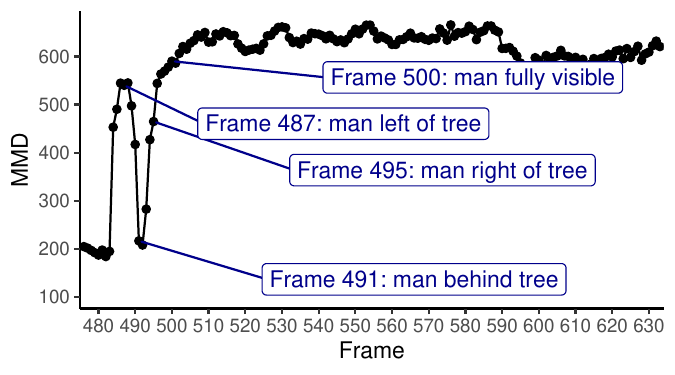}
        \caption{Frames 471 to 633.}
        \label{fig:video_md1}
    \end{subfigure}
    \caption{Plot of robust MMD based on MMCD estimators for the video data.}
    \label{fig:video_md}
\end{figure}
It is interesting to see a certain increase in the MMD in Figure~\ref{fig:video_md0} between frames 400 and 450. Here, the Shapley values on the contour of the palm tree contribute the most to the outlyingness. This could be caused by a slight shifting of the camera or a small movement of the palm tree due to wind.

\begin{figure}[!ht]
    \centering
    \includegraphics[width=\linewidth]{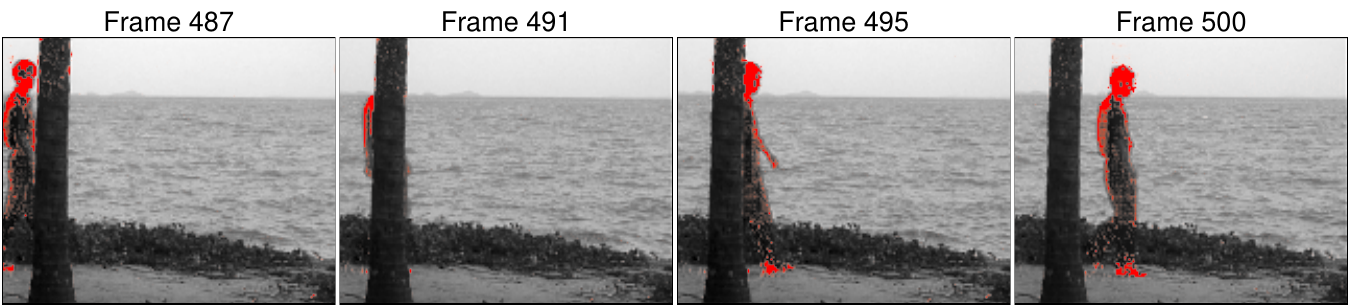}
    \caption{Outlyingness scores based on cellwise Shapley values are shown in red, where darker colors indicate higher outlyingness contributions, and the grayscale video frames are displayed in the background.}
    \label{fig:video_shapley}
\end{figure}

\section{Summary and conclusions}
\label{section:conclusion}

Matrix-valued observations, like images or dual-factor data tables, are common in various fields. 
To apply multivariate methods on matrix-valued data, the matrices are typically converted to vectors by stacking either the rows or columns. This disrupts the inherent data structure and increases dimensionality, thereby complicating parameter estimation. 
Thus, it is often preferable to model matrix-valued data directly with matrix-variate distributions. In this setting, Maximum Likelihood (ML) estimation methods exist for estimating the mean, as well as the row and column covariances, respectively. However, these estimators are sensitive to deviations caused by outliers among matrix-valued observations.

This work introduced the MMCD (matrix minimum covariance determinant) estimators as a robust counterpart to the ML estimators in the matrix-variate normal model. Several desirable properties are achieved: equivariance under matrix affine transformations, high breakdown point, and consistency under elliptical matrix-variate distributions. The proposed reweighted versions lead to higher efficiency but not to any loss in terms of breakdown point. An algorithm along the lines of the Fast-MCD procedure \citep{Rousseeuw1999} allows for efficient computation
of the estimators. Simulation experiments validate the theoretical properties and advantages.
Depending on the ratio of the number of rows and columns of the matrix-valued observations, the MMCD estimators show a big advantage over robust estimation for vectorized observations regarding breakdown and computational efficiency.

We further extended the outlier explanation concept based on Shapley values \citep{mayrhofer2022} to the matrix-variate setting.
This allows for an additive decomposition of the matrix-variate Mahalanobis distance of an observation into Shapley contributions of either the rows, the columns, or the matrix cells. The resulting Shapley values greatly aid with diagnostics, particularly in revealing those cells (rows, columns) of the matrix with the most substantial contributions to the outlyingness of the observation.

The efficiency of MMCD estimators in outlier detection for large sample sizes is evident from the simulations. However, our future research aims to improve and extend these estimators. For instance, smaller sample sizes might benefit from integrating finite sample corrections proposed by \cite{pison2002small} to enhance the results. 
Furthermore, the iterative computation of MMCD covariance estimators, which involves inverse covariance matrices, requires data that ensures full-rank estimates at each iteration. This requirement may be impeded for example in image data, in case certain rows or columns maintain constant pixel values across all observations.
To solve this, regularization involving a linear combination of the covariance matrix with a full-rank target matrix can be used \citep{ledoit2004well}, similarly to the multivariate setting \citep{boudt2020minimum}.

The MMCD objective can be expressed as a trimmed maximum likelihood problem, and thus, can be extended to tensor-valued data using ML estimation for the tensor normal distribution~\cite{manceur2013maximum}. The framework of \cite{Raymaekers2022} can be used to develop a cellwise robust version of the MMCD. 
Our ongoing research focuses on extending the MMCD estimators and outlier explanations based on Shapley values to the field of functional data analysis. Our goal is to introduce robust estimators and enhance interpretability for multivariate functional data. 
In the future, we also plan to incorporate these robust estimators as plug-in estimators to robustify established multivariate methodologies in the matrix-variate domain, like principal component analysis and discriminant analysis.

\noindent\textbf{Disclosure statement}\\
The authors report there are no competing interests to declare.

\noindent\textbf{Software and data availability}\\
The \texttt{R} package \texttt{robustmatrix} includes a parallelized \texttt{C++} implementation of the MMCD algorithm and a vignette to reproduce the examples presented in this paper. 

\noindent\textbf{Acknowledgements}\\
This work was supported by the AI4CSM project and has received funding from the ECSEL Joint Undertaking (JU) under grant agreement No 101007326; and the Austrian IKT der Zukunft programme via the Austrian Research Promotion Agency (FFG) and the Austrian Federal Ministry for Climate Action, Environment, Energy, Mobility, Innovation and Technology (BMK) under project No 884070. This work was supported by the Austrian Science Fund (FWF), project number I 5799-N.

\bibliographystyle{apalike}
\bibliography{references} 

\begin{thebibliography}{}

\bibitem[Agostinelli et~al., 2015]{agostinelli2015robust}
Agostinelli, C., Leung, A., Yohai, V.~J., and Zamar, R.~H. (2015).
\newblock Robust estimation of multivariate location and scatter in the
  presence of cellwise and casewise contamination.
\newblock {\em Test}, 24:441--461.

\bibitem[Alqallaf et~al., 2009]{alqallaf2009propagation}
Alqallaf, F., Van~Aelst, S., Yohai, V.~J., and Zamar, R.~H. (2009).
\newblock Propagation of outliers in multivariate data.
\newblock {\em The Annals of Statistics}, 37(1):311--331.

\bibitem[Boudt et~al., 2020]{boudt2020minimum}
Boudt, K., Rousseeuw, P.~J., Vanduffel, S., and Verdonck, T. (2020).
\newblock The minimum regularized covariance determinant estimator.
\newblock {\em Statistics and Computing}, 30(1):113--128.

\bibitem[Butler et~al., 1993]{butler1993asymptotics}
Butler, R., Davies, P., and Jhun, M. (1993).
\newblock Asymptotics for the minimum covariance determinant estimator.
\newblock {\em The Annals of Statistics}, pages 1385--1400.

\bibitem[Cator and Lopuha{\"a}, 2012]{cator2012central}
Cator, E.~A. and Lopuha{\"a}, H.~P. (2012).
\newblock {Central limit theorem and influence function for the MCD estimators
  at general multivariate distributions}.
\newblock {\em Bernoulli}, 18(2):520 -- 551.

\bibitem[Cilia et~al., 2022]{cilia2022diagnosing}
Cilia, N.~D., De~Gregorio, G., De~Stefano, C., Fontanella, F., Marcelli, A.,
  and Parziale, A. (2022).
\newblock Diagnosing alzheimer’s disease from on-line handwriting: a novel
  dataset and performance benchmarking.
\newblock {\em Engineering Applications of Artificial Intelligence},
  111:104822.

\bibitem[Cilia et~al., 2018]{cilia2018experimental}
Cilia, N.~D., De~Stefano, C., Fontanella, F., and Di~Freca, A.~S. (2018).
\newblock An experimental protocol to support cognitive impairment diagnosis by
  using handwriting analysis.
\newblock {\em Procedia Computer Science}, 141:466--471.

\bibitem[Croux and Haesbroeck, 1999]{croux1999influence}
Croux, C. and Haesbroeck, G. (1999).
\newblock Influence function and efficiency of the minimum covariance
  determinant scatter matrix estimator.
\newblock {\em Journal of Multivariate Analysis}, 71(2):161--190.

\bibitem[Davies, 1987]{davies1987asymptotic}
Davies, P.~L. (1987).
\newblock Asymptotic behaviour of s-estimates of multivariate location
  parameters and dispersion matrices.
\newblock {\em The Annals of Statistics}, pages 1269--1292.

\bibitem[Dawid, 1981]{dawid1981matrix}
Dawid, A.~P. (1981).
\newblock Some matrix-variate distribution theory: notational considerations
  and a bayesian application.
\newblock {\em Biometrika}, 68(1):265--274.

\bibitem[Dutilleul, 1999]{Dutilleul1999}
Dutilleul, P. (1999).
\newblock The mle algorithm for the matrix normal distribution.
\newblock {\em Journal of Statistical Computation and Simulation},
  64(2):105--123.

\bibitem[Garc{\'\i}a-Escudero et~al., 2010]{garcia2010}
Garc{\'\i}a-Escudero, L.~A., Gordaliza, A., Matr{\'a}n, C., and Mayo-Iscar, A.
  (2010).
\newblock A review of robust clustering methods.
\newblock {\em Advances in Data Analysis and Classification}, 4:89--109.

\bibitem[Gupta and Nagar, 1999]{gupta1999}
Gupta, A. and Nagar, D. (1999).
\newblock {\em Matrix Variate Distributions}.
\newblock Monographs and Surveys in Pure and Applied Mathematics. Taylor \&
  Francis.

\bibitem[Gupta and Varga, 2012]{gupta2012elliptically}
Gupta, A.~K. and Varga, T. (2012).
\newblock {\em Elliptically contoured models in statistics}, volume 240.
\newblock Springer Science \& Business Media.

\bibitem[Hubert et~al., 2012]{hubert2012deterministic}
Hubert, M., Rousseeuw, P.~J., and Verdonck, T. (2012).
\newblock A deterministic algorithm for robust location and scatter.
\newblock {\em Journal of Computational and Graphical Statistics},
  21(3):618--637.

\bibitem[Kurnaz et~al., 2018]{kurnaz2018}
Kurnaz, F.~S., Hoffmann, I., and Filzmoser, P. (2018).
\newblock Robust and sparse estimation methods for high-dimensional linear and
  logistic regression.
\newblock {\em Chemometrics and Intelligent Laboratory Systems}, 172:211--222.

\bibitem[Ledoit and Wolf, 2004]{ledoit2004well}
Ledoit, O. and Wolf, M. (2004).
\newblock A well-conditioned estimator for large-dimensional covariance
  matrices.
\newblock {\em Journal of multivariate analysis}, 88(2):365--411.

\bibitem[Li et~al., 2004]{li2004statistical}
Li, L., Huang, W., Gu, I. Y.-H., and Tian, Q. (2004).
\newblock Statistical modeling of complex backgrounds for foreground object
  detection.
\newblock {\em IEEE Transactions on image processing}, 13(11):1459--1472.

\bibitem[Lopuhaa and Rousseeuw, 1991]{lopuhaa1991breakdown}
Lopuhaa, H.~P. and Rousseeuw, P.~J. (1991).
\newblock Breakdown points of affine equivariant estimators of multivariate
  location and covariance matrices.
\newblock {\em The Annals of Statistics}, pages 229--248.

\bibitem[Lu and Zimmerman, 2005]{lu2005likelihood}
Lu, N. and Zimmerman, D.~L. (2005).
\newblock The likelihood ratio test for a separable covariance matrix.
\newblock {\em Statistics \& Probability Letters}, 73(4):449--457.

\bibitem[Lundberg and Lee, 2017]{lundberg2017}
Lundberg, S.~M. and Lee, S.-I. (2017).
\newblock A unified approach to interpreting model predictions.
\newblock In Guyon, I., Luxburg, U.~V., Bengio, S., Wallach, H., Fergus, R.,
  Vishwanathan, S., and Garnett, R., editors, {\em Advances in Neural
  Information Processing Systems 30}, pages 4765--4774. Curran Associates, Inc.

\bibitem[Mahalanobis, 1936]{mahalanobis1936}
Mahalanobis, P.~C. (1936).
\newblock On the generalized distance in statistics.
\newblock {\em Proceedings of the National Institute of Sciences (Calcutta)},
  2:49--55.

\bibitem[Manceur and Dutilleul, 2013]{manceur2013maximum}
Manceur, A.~M. and Dutilleul, P. (2013).
\newblock Maximum likelihood estimation for the tensor normal distribution:
  Algorithm, minimum sample size, and empirical bias and dispersion.
\newblock {\em Journal of Computational and Applied Mathematics}, 239:37--49.

\bibitem[Maronna et~al., 2019]{maronna2019robust}
Maronna, R.~A., Martin, R.~D., Yohai, V.~J., and Salibi{\'a}n-Barrera, M.
  (2019).
\newblock {\em Robust Statistics: Theory and Methods (with R)}.
\newblock John Wiley \& Sons.

\bibitem[Mayrhofer and Filzmoser, 2023]{mayrhofer2022}
Mayrhofer, M. and Filzmoser, P. (2023).
\newblock Multivariate outlier explanations using shapley values and
  mahalanobis distances.
\newblock {\em Econometrics and Statistics}.

\bibitem[Neykov et~al., 2007]{neykov2007}
Neykov, N., Filzmoser, P., Dimova, R., and Neytchev, P. (2007).
\newblock Robust fitting of mixtures using the trimmed likelihood estimator.
\newblock {\em Computational Statistics \& Data Analysis}, 52(1):299--308.

\bibitem[Pison et~al., 2002]{pison2002small}
Pison, G., Van~Aelst, S., and Willems, G. (2002).
\newblock Small sample corrections for lts and mcd.
\newblock {\em Metrika}, 55:111--123.

\bibitem[Raymaekers and Rousseeuw, 2023]{Raymaekers2022}
Raymaekers, J. and Rousseeuw, P.~J. (2023).
\newblock The cellwise minimum covariance determinant estimator.
\newblock {\em Journal of the American Statistical Association}, 0(0):1--12.

\bibitem[Ro{\'s} et~al., 2016]{ros2016existence}
Ro{\'s}, B., Bijma, F., de~Munck, J.~C., and de~Gunst, M.~C. (2016).
\newblock Existence and uniqueness of the maximum likelihood estimator for
  models with a kronecker product covariance structure.
\newblock {\em Journal of Multivariate Analysis}, 143:345--361.

\bibitem[Rousseeuw, 1985]{Rousseeuw1985}
Rousseeuw, P. (1985).
\newblock Multivariate estimation with high breakdown point.
\newblock {\em Mathematical Statistics and Applications Vol. B}, pages
  283--297.

\bibitem[Rousseeuw and Driessen, 1999]{Rousseeuw1999}
Rousseeuw, P.~J. and Driessen, K.~V. (1999).
\newblock A fast algorithm for the minimum covariance determinant estimator.
\newblock {\em Technometrics}, 41(3):212--223.

\bibitem[Shapley, 1953]{shapley1953}
Shapley, L.~S. (1953).
\newblock A value for n-person games.
\newblock {\em Contributions to the Theory of Games}, 2(28):307--317.

\bibitem[Soloveychik and Trushin, 2016]{soloveychik2016gaussian}
Soloveychik, I. and Trushin, D. (2016).
\newblock Gaussian and robust kronecker product covariance estimation:
  Existence and uniqueness.
\newblock {\em Journal of Multivariate Analysis}, 149:92--113.

\bibitem[Srivastava et~al., 2008]{srivastava2008models}
Srivastava, M.~S., von Rosen, T., and Von~Rosen, D. (2008).
\newblock Models with a kronecker product covariance structure: estimation and
  testing.
\newblock {\em Mathematical methods of statistics}, 17:357--370.

\bibitem[Sun et~al., 2016]{sun2016robust}
Sun, Y., Babu, P., and Palomar, D.~P. (2016).
\newblock Robust estimation of structured covariance matrix for heavy-tailed
  elliptical distributions.
\newblock {\em IEEE Transactions on Signal Processing}, 64(14):3576--3590.

\bibitem[Thompson et~al., 2020]{thompson2020classification}
Thompson, G.~Z., Maitra, R., Meeker, W.~Q., and Bastawros, A.~F. (2020).
\newblock Classification with the matrix-variate-t distribution.
\newblock {\em Journal of Computational and Graphical Statistics},
  29(3):668--674.

\bibitem[Tyler, 1987]{tyler1987distribution}
Tyler, D.~E. (1987).
\newblock A distribution-free m-estimator of multivariate scatter.
\newblock {\em The Annals of Statistics}, pages 234--251.

\bibitem[Werner et~al., 2008]{werner2008estimation}
Werner, K., Jansson, M., and Stoica, P. (2008).
\newblock On estimation of covariance matrices with {K}ronecker product
  structure.
\newblock {\em IEEE Transactions on Signal Processing}, 56(2):478--491.

\bibitem[Zhang et~al., 2022]{zhang2022covariance}
Zhang, Y., Shen, W., and Kong, D. (2022).
\newblock Covariance estimation for matrix-valued data.
\newblock {\em Journal of the American Statistical Association}, pages 1--12.

\end{thebibliography}

\bigskip
\appendix
\section{Preliminaries} 
\label{supplement:preliminaries}

Consider an i.i.d. sample $\bm{\mathfrak{X}} = (\bm{X}_1,\ldots,\bm{X}_n) \in \R^{n \times p \times q}$, with $\bm{X}_i \sim \MN(\bm{\M}, \bm{\Sigma}^{\row}, \bm{\Sigma}^{\col})$. Due to the factored covariance structure of matrix normal data, the rowwise and columnwise covariance matrices $\bm{\Sigma}^{\row}$ and $\bm{\Sigma}^{\col}$ are only identified up to a multiplicative constant $\kappa \neq 0$, since replacing $\bm{\Sigma}^{\row}$ by $\kappa\bm{\Sigma}^{\row}$ and $\bm{\Sigma}^{\col}$ by $\nicefrac{1}{\kappa}\bm{\Sigma}^{\col}$ does not change the pdf of $\bm{X}$. 
While the Kronecker product $\bm{\Sigma}^{\col} \otimes \bm{\Sigma}^{\row}$ can be uniquely identified, the issue of trivial non-uniqueness of $\bm{\Sigma}^{\row}$ and $\bm{\Sigma}^{\col}$ is commonly solved by either fixing a diagonal entry, the determinant, or the norm of either matrix \citep{ros2016existence, soloveychik2016gaussian}. For simplicity, we assume that the first diagonal entry of $\bm{\Sigma}^{\col}$ is set to one.
This implies that the uniqueness of $\bm{\Sigma}^{\col} \otimes \bm{\Sigma}^{\row}$ is equivalent to the uniqueness of $\bm{\Sigma}^{\col}$ and $\bm{\Sigma}^{\row}$ with the identifiability constraint $\sigma^{\col}_{11} = 1$. The multiplicative constant for their estimators is also chosen such that $\hat{\sigma}^{\col}_{11} = 1$.

Instead of using Equations~\eqref{eq:matrix_mean_mle}-\eqref{eq:matrix_cov_col_mle} for mean and covariance estimation, it is also possible to consider the vectorized samples $\bm{x}_i = \vect(\bm{X}_i) \sim \NN(\bm{\mu},\bm{\Sigma}), i = 1, \dots, n$,  where $\bm{\mu} = \vect(\bm{M})$ and $\bm{\Sigma} = \bm{\Sigma}^{\col} \otimes \bm{\Sigma}^{\row}$ denote the mean and covariance matrix, respectively.
Then the maximum likelihood estimators for mean and covariance are given by 
\begin{align}
    \hat{\bm{\mu}} = \frac{1}{n} \sum_{i = 1}^n \bm{x}_i \quad \text{and} \quad \hat{\bm{\Sigma}} = \frac{1}{n} \sum_{i = 1}^n (\bm{x}_i -\hat{\bm{\mu}})(\bm{x}_i -\hat{\bm{\mu}})', \label{eq:multivariate_normal_mle}
\end{align}
respectively. The computation of the MLEs for matrix-variate samples based on Equations~\eqref{eq:matrix_mean_mle}-\eqref{eq:matrix_cov_col_mle} involves estimating $p(p+1)/2 + q(q+1)/2 + pq$ parameters instead of $pq(pq + 1)/2 + pq$ parameters for the vectorized observations according to Equation~\eqref{eq:multivariate_normal_mle}. This raises the question of whether fewer than $pq+1$ observations are sufficient for guaranteeing the existence and uniqueness of MLEs for i.i.d. samples from a matrix normal distribution.
This question was investigated in several papers, such as \cite{Dutilleul1999, lu2005likelihood, srivastava2008models, ros2016existence, soloveychik2016gaussian}. We rely on the latter for the most recent proof of those conditions. 
Note that it is not necessary to assume that the sample consists of i.i.d. observations. In fact, the i.i.d. assumption can be relaxed to allow for statistically dependent samples and it is not even necessary to require identical distribution \citep[Remarks 2 and 6]{soloveychik2016gaussian}. The critical condition for existence and uniqueness is that the sample contains at least $n \geq \lfloor \nicefrac{p}{q} + \nicefrac{q}{p} \rfloor + 2$ observations that are not collinear. 
The same holds for the existence and uniqueness of the MMCD estimators, where $n$ is replaced by $h$, and for properties like the breakdown point the assumptions could be relaxed only requiring that the sample is in general position, i.e., no subset of $r, 2 \leq r \leq \lfloor \nicefrac{p}{q} + \nicefrac{q}{p} \rfloor + 2$ samples lies on an $r-2$ dimensional subspace. However, the i.i.d. assumption is still necessary when we consider properties like consistency. 

The idea of the multivariate MCD estimator is as follows: Let $\bm{x}_i = (x_{i1},\ldots,x_{ip})' \in \R^p$ denote the $i$-th observation of a data set in the multivariate setting, where $i=1, \ldots, n$. The objective of the MCD estimator is to find the subset of $h$ out of $n$ observations whose sample covariance matrix has the lowest determinant, with $\nicefrac{n}{2} \leq h \leq n$ and $h > p$.
In total, there are $\binom{h}{n}$ possible $h$-subsets, and thus, a strategy needs to be used to tackle the optimization problem efficiently. This has been done with the so-called Fast-MCD algorithm~\citep{Rousseeuw1999}, which internally sorts the observations based on their Mahalanobis distances. For an observation $\bm{x}_i$ from a population with mean $\bm{\mu} \in \R^{p}$ and covariance $\bm{\Sigma} \in \pds(p)$ it is given by
\begin{align*}
    \md(\bm{x}_i,\bm{\mu},\bm{\Sigma}) = \sqrt{(\bm{x}_i - \bm{\mu})' \bm{\Sigma}^{-1} (\bm{x}_i - \bm{\mu})}.
\end{align*}

Since the Mahalanobis distance is vital for the computation of the MCD estimator, it will also be crucial in a matrix-variate extension, where it can be directly derived from the Mahalanobis distance of a vectorized matrix-variate observation $\bm{X}$ as
\begin{align*}
    \mmd^2(\bm{X}) = \mmd^2(\bm{X};\bm{\M},\bm{\Sigma}^{\row},\bm{\Sigma}^{\col})
    & = \md^2(\vect(\bm{X}))\\ 
    & =  \vect(\bm{X}-\bm{\M})'(\bm{\Omega}^{\col} \otimes \bm{\Omega}^{\row})\vect(\bm{X}-\bm{\M}) \\
    &= \sum_{i = 1}^p \sum_{j = 1}^p \sum_{k = 1}^q \sum_{l = 1}^q  (x_{ik}-m_{ik}) (x_{jl}-m_{jl}) \omega_{ij}^{\row} \omega_{kl}^{\col} \\ 
    &= \tr(\bm{\Omega}^{\col}(\bm{X}-\bm{\M})'\bm{\Omega}^{\row}(\bm{X}-\bm{\M})),
\end{align*}
where $m_{ij}$, $\omega_{ij}^{\row}$ and $\omega_{ij}^{\col}$ denote the elements $(i,j)$ of the matrices $\bm{\M}$, $\bm{\Omega}^{\row}$ and $\bm{\Omega}^{\col}$, respectively. 
If $\bm{X}$ has a matrix normal distribution, then the squared matrix Mahalanbois distance has a $\chi^2$ distribution with $pq$ degrees of freedom, $\mmd^2(\bm{X}) \sim \chi^2_{pq}$ \citep{gupta1999}. 

\section{Proofs of Section~\ref{section:MMCD}}
\label{supplement:mmcd_proofs}

\begin{proof}[Proof of Proposition~\ref{proposition:MCD_likelihood}]
In optimization problem~\eqref{eq:max_loglik} we want to maximize 
\begin{align} 
        \begin{split}
            l(\bm{w},\bm{\M}, \bm{\Sigma}^{\row}, \bm{\Sigma}^{\col}|\bm{\mathfrak{X}}) =&
            -\frac{1}{2} \sum_{i=1}^n w_i \Bigl( 
            p\ln(\det(\bm{\Sigma}^{\col})) + q\ln(\det(\bm{\Sigma}^{\row})) \Bigr) \\ 
            &- \frac{1}{2}\sum_{i=1}^n w_i \mmd^2(\bm{X}_i) - hpq\ln(2\pi) 
        \end{split}  \label{eq:MMCD_likelihood}
\end{align}
subject to $w_i \in \{0,1\}$ for all $i = 1,\dots,n$ and $\sum_{i = 1}^n w_i = h$. 
In Equation~\eqref{eq:MMCD_likelihood}, $\mmd^2(\bm{X}_i)$ is defined as in Equation~\eqref{eq:MD=MMD}. 

For any random $h$-subset $H$ (or equivalently the corresponding set of weights $\bm{w}$) the constrained MLEs for $\bm{\M}$, $\bm{\Sigma}^{\row}$, and $\bm{\Sigma}^{\col}$ of Equation~\eqref{eq:MMCD_likelihood} can be written as:
\begin{align*}
    \hat{\bm{\M}}_{H} &= \frac{1}{h} \sum_{i = i}^n w_i \bm{X}_i = \frac{1}{h}\sum_{i \in H} \bm{X}_i\\
    \hat{\bm{\Sigma}}^{\row}_{H} &= \frac{1}{qh} \sum_{i = i}^n w_i (\bm{X}_i - \hat{\bm{\M}}_{H})\hat{\bm{\Omega}}^{\col}_{H} (\bm{X}_i - \hat{\bm{\M}}_{H})' = \frac{1}{qh} \sum_{i \in H} (\bm{X}_i - \hat{\bm{\M}}_{H})\hat{\bm{\Omega}}^{\col}_{H} (\bm{X}_i - \hat{\bm{\M}}_{H})'\\
    \hat{\bm{\Sigma}}^{\col}_{H} &= \frac{1}{ph} \sum_{i = i}^n w_i  (\bm{X}_i - \hat{\bm{\M}}_{H})'\hat{\bm{\Omega}}^{\row}_{H} (\bm{X}_i - \hat{\bm{\M}}_{H}) = \frac{1}{ph} \sum_{i \in H} (\bm{X}_i - \hat{\bm{\M}}_{H})'\hat{\bm{\Omega}}^{\row}_{H} (\bm{X}_i - \hat{\bm{\M}}_{H})
\end{align*}
Using those estimators to compute the sum of the Mahalanobis distances $\mmd^2(\bm{X}_i)$  in Equation~\eqref{eq:MMCD_likelihood} we obtain
\begin{align*}
    \sum_{i=1}^n w_i \mmd^2(\bm{X}_i) =& 
    \sum_{i \in H} \tr\left(\hat{\bm{\Omega}}^{\col}_{H}(\bm{X}_i - \hat{\bm{\M}}_{H})' \hat{\bm{\Omega}}^{\row}_{H} (\bm{X}_i - \hat{\bm{\M}}_{H})\right) \\
    =& \sum_{i \in H} \tr\left((\bm{X}_i -\hat{\bm{\M}}_{H})\hat{\bm{\Omega}}^{\col}_{H} (\bm{X}_i -\hat{\bm{\M}}_{H})' \hat{\bm{\Omega}}^{\row}_{H} \right) \\
    =& \tr\left(\sum_{i \in H}((\bm{X}_i -\hat{\bm{\M}}_{H})\hat{\bm{\Omega}}^{\col}_{H} (\bm{X}_i -\hat{\bm{\M}}_{H})')\hat{\bm{\Omega}}^{\row}_{H}\right) \\
    =& \tr\left(qh \hat{\bm{\Sigma}}^{\row}_{H} \hat{\bm{\Omega}}^{\row}_{H} \right) 
    = hpq .
\end{align*}

Thus, the terms in the second row of Equation~\eqref{eq:MMCD_likelihood} are all constant, and it is sufficient to maximize only the term in the first row, which contains the (negative) determinant of Equation~\eqref{eq:MMCD_cov_determinant}.
\end{proof}

\subsection{Properties of MMCD estimators}
\label{supplement:mmcd_properties_proofs}


\begin{proof}[Proof of Lemma~\ref{lemma:kronecker_affine_transformation}]
      
    Ad (a): We show that the MMCD estimators are matrix affine equivariant. Let us consider the objective of the MMCD for the transformed samples, which is to minimize 
    \begin{align*}
        \det(\hat{\bm{\Sigma}}^{\col}_{\bm{\mathfrak{Z}}_{H}} \otimes \hat{\bm{\Sigma}}^{\row}_{\bm{\mathfrak{Z}}_{H}}) &= \det\Bigl((\bm{B}'\hat{\bm{\Sigma}}^{\col}_{\bm{\mathfrak{X}}_{H}}\bm{B}) \otimes (\bm{A}\hat{\bm{\Sigma}}^{\row}_{\bm{\mathfrak{X}}_{H}}\bm{A}')\Bigr)\\
    &= \Bigl[\det(\bm{B}'\hat{\bm{\Sigma}}^{\col}_{\bm{\mathfrak{X}}_{H}}\bm{B})\Bigr]^p \Bigl[\det (\bm{A}\hat{\bm{\Sigma}}^{\row}_{\bm{\mathfrak{X}}_{H}}\bm{A}')\Bigr]^q\\
    &= \Bigl[\det(\bm{B}')\det(\hat{\bm{\Sigma}}^{\col}_{\bm{\mathfrak{X}}_{H}})\det(\bm{B})\Bigr]^p \Bigl[\det(\bm{A})\det(\hat{\bm{\Sigma}}^{\row}_{\bm{\mathfrak{X}}_{H}})\det(\bm{A}')\Bigr]^q\\
    &= 4\det(\bm{B})^p \det(\bm{A})^q \det(\hat{\bm{\Sigma}}^{\col}_{\bm{\mathfrak{X}}_{H}})^p \det(\hat{\bm{\Sigma}}^{\row}_{\bm{\mathfrak{X}}_{H}})^q.
    \end{align*}
    Since $4\det(\bm{B})^p \det(\bm{A})^q$ is constant, the objective does not change, and we obtain the same $h$-subset. Since the MMCD estimators correspond to the trimmed MLEs and the objective is not affected by the transformation, the matrix affine equivariance of the MMCD estimators follows from the matrix affine equivariance of the MLEs. 

    Ad (b): 
    Suppose that $(\hat{\bm{\M}}_{\bm{\mathfrak{Z}}}, \hat{\bm{\Sigma}}^{\row}_{\bm{\mathfrak{Z}}}, \hat{\bm{\Sigma}}^{\col}_{\bm{\mathfrak{Z}}})$ are matrix affine equivariant estimators of location and covariance of the transformed sample $\bm{\mathfrak{Z}}$, then 
    \begin{align*}
    &\mmd^2(\bm{Z}_i;\hat{\bm{\M}}_{\bm{\mathfrak{Z}}}, \hat{\bm{\Sigma}}^{\row}_{\bm{\mathfrak{Z}}}, \hat{\bm{\Sigma}}^{\col}_{\bm{\mathfrak{Z}}})\\
    =& \tr(\hat{\bm{\Omega}}^{\col}_{\bm{\mathfrak{Z}}}(\bm{Z}_i-\hat{\bm{\M}}_{\bm{\mathfrak{Z}}})'\hat{\bm{\Omega}}^{\row}_{\bm{\mathfrak{Z}}}(\bm{Z}_i-\hat{\bm{\M}}_{\bm{\mathfrak{Z}}})) \\
    =& \tr\Big(\big(\bm{B}^{-1}\hat{\bm{\Omega}}^{\col}_{\bm{\mathfrak{X}}}(\bm{B}')^{-1}
    (\bm{A}\bm{X}_i\bm{B} + \bm{C}-(\bm{A}\hat{\bm{\M}}_{\bm{\mathfrak{X}}}\bm{B}+ \bm{C}))'\big) \\
    &\qquad\big((\bm{A}')^{-1}\hat{\bm{\Omega}}^{\row}_{\bm{\mathfrak{X}}}\bm{A}^{-1} (\bm{A}\bm{X}_i\bm{B} + \bm{C} -(\bm{A}\hat{\bm{\M}}_{\bm{\mathfrak{X}}}\bm{B} + \bm{C}))\big)\Big) \\
    =& \tr(\bm{B}^{-1}\hat{\bm{\Omega}}^{\col}_{\bm{\mathfrak{X}}}(\bm{B}')^{-1} \bm{B}'(\bm{X}_i-\hat{\bm{\M}}_{\bm{\mathfrak{X}}})'\bm{A}' (\bm{A}')^{-1}\hat{\bm{\Omega}}^{\row}_{\bm{\mathfrak{X}}}\bm{A}^{-1} \bm{A}(\bm{X}_i-\hat{\bm{\M}}_{\bm{\mathfrak{X}}})\bm{B})\\
    =& \tr(\hat{\bm{\Omega}}^{\col}_{\bm{\mathfrak{X}}}(\bm{X}_i-\hat{\bm{\M}}_{\bm{\mathfrak{X}}})'\hat{\bm{\Omega}}^{\row}_{\bm{\mathfrak{X}}}(\bm{X}_i-\hat{\bm{\M}}_{\bm{\mathfrak{X}}})) = \mmd^2(\bm{X}_i;\hat{\bm{\M}}_{\bm{\mathfrak{X}}}, \hat{\bm{\Sigma}}^{\row}_{\bm{\mathfrak{X}}}, \hat{\bm{\Sigma}}^{\col}_{\bm{\mathfrak{X}}}).
\end{align*}
\end{proof}


The proofs of Theorems~\ref{theorem:MMCD_breakdown_point} and \ref{theorem:MMCD_reweighted_breakdown_point} require some definitions and properties related to the vector space of matrices, which are introduced before the proofs of the theorems.
Since all matrices of a fixed size form a vector space, objects such as ellipsoids or a simplex that are defined on the more common vector spaces are also defined here. Let 
\begin{align}
    E(\bm{T}, \bm{U}, \bm{V}) = \{\bm{X}: \tr(\bm{V}^{-1}(\bm{X} - \bm{T})'\bm{U}^{-1}(\bm{X} - \bm{T})) \leq 1\} \label{eq:matrix_ellipsoid}
\end{align} 
be the ellipsoid containing the matrices $\bm{X}  \in \R^{p \times q}$ with $\mmd^2(\bm{X};\bm{T}, \bm{U}, \bm{V}) \leq 1$, where $\bm{T} \in \R^{p \times q}$, $\bm{U} \in \pds(p)$ and $\bm{V} \in \pds(q)$. 
The volume of this ellipsoid is given by 
\begin{align}\label{eq:det(E)}
    \vol(E(\bm{T}, \bm{U}, \bm{V})) = \underbrace{\frac{\pi^{\nicefrac{pq}{2}}}{\Gamma(\nicefrac{pq}{2}+1)}}_{=:\beta_{pq}}\prod_{i = 1}^p\prod_{j = 1}^q\sqrt{\lambda_i(\bm{U})\lambda_j(\bm{V})}  = \beta_{pq} \underbrace{\det(\bm{U})^{\nicefrac{q}{2}}\det(\bm{V})^{\nicefrac{p}{2}}}_{=:\det(E(\bm{T}, \bm{U}, \bm{V}))},
\end{align}
where $\Gamma$ is the gamma function, $0 < \lambda_p(\bm{U}) \leq \ldots \leq \lambda_1(\bm{U})$ and $0 < \lambda_q(\bm{V}) \leq \ldots \leq \lambda_1(\bm{V})$ are the eigenvalues of $\bm{U}$ and $\bm{V}$, respectively. 
Moreover, the axes have lengths $\sqrt{\lambda_i(\bm{U})\lambda_j(\bm{V})}$.

Let $\bm{A}$ be a symmetric nonnegative definite $p \times p$ matrix, then 
\begin{align} \label{eq:eigenvalue_inf_sup}
    \lambda_1(\bm{A}) = \sup_{\bm{z} \in \R^{p}}\frac{\bm{z}'\bm{A}\bm{z}}{\bm{z}'\bm{z}} \quad \text{and} \quad \lambda_n(\bm{A}) = \inf_{\bm{z} \in \R^{p}}\frac{\bm{z}'\bm{A}\bm{z}}{\bm{z}'\bm{z}}.
\end{align}
Consider another symmetric nonnegative definite $p \times p$ matrix $\bm{B}$, then using Equation~\eqref{eq:eigenvalue_inf_sup} we get that 
\begin{align} \label{eq:sum_eigenvalue_inequality}
    \lambda_1(\bm{A} + \bm{B}) \leq \lambda_1(\bm{A}) + \lambda_1(\bm{B}) \quad \text{and} \quad \lambda_p(\bm{A} + \bm{B}) \geq \lambda_p(\bm{A}) + \lambda_p(\bm{B}).
\end{align}

If $\bm{A} \in \pds(p)$ with eigenvalues $0 < \lambda_p(\bm{A}) \leq \ldots \leq \lambda_1(\bm{A})$ then the eigenvalues of $\bm{A}^{-1}$ are the reciprocals of the eigenvalues of $\bm{A}$, i.e. $\lambda_i(\bm{A}^{-1}) = \lambda^{-1}_i(\bm{A})$. Hence, we have that 
\begin{align*}
    \frac{1}{\lambda_1(\bm{A})} = \inf_{\bm{z} \in \R^{p}}\frac{\bm{z}'\bm{A}^{-1}\bm{z}}{\bm{z}'\bm{z}},
\end{align*}
which implies that for any $\bm{x} \in \R^p$ 
\begin{align}
    \frac{1}{\lambda_1(\bm{A})} \leq \frac{\bm{x}'\bm{A}^{-1}\bm{x}} {\bm{x}'\bm{x}} \Leftrightarrow \bm{x}'\bm{x} \leq \bm{x}'\bm{A}^{-1}\bm{x} \lambda_1(\bm{A}). \label{eq:eigenvalue_inequality}
\end{align}

Suppose $\bm{A} \in \pds(p)$, $\bm{B} \in \pds(q)$ and let $\lambda(\bm{A})$ be an eigenvalue of $\bm{A}$ with corresponding eigenvector $\bm{v}(\bm{A})$, and $\lambda(\bm{B})$ an eigenvalue of $\bm{B}$ with corresponding eigenvector $\bm{v}(\bm{B})$. Then $\lambda(\bm{A})\lambda(\bm{B})$ is an eigenvalue of $\bm{B} \otimes \bm{A}$ with corresponding eigenvector $\bm{v}(\bm{B}) \otimes \bm{v}(\bm{A})$.
We denote the sequence of eigenvalues of $\bm{A}$ and $\bm{B}$ as $0 < \lambda_p(\bm{A}) \leq \ldots \leq \lambda_1(\bm{A})$ and $0 < \lambda_q(\bm{b}) \leq \ldots \leq \lambda_1(\bm{b})$, respectively. It follows that the smallest eigenvalue $\lambda_{pq}(\bm{A},\bm{B}) = \lambda_p(\bm{A})\lambda_q(\bm{B})$ and the largest eigenvalue $\lambda_{1}(\bm{A},\bm{B}) = \lambda_1(\bm{A})\lambda_1(\bm{B})$. Moreover note that for $\bm{Z}\in \R^{p \times q}$
\begin{align*}
    \vect(\bm{Z})'(\bm{B} \otimes \bm{A})\vect(\bm{Z}) = \tr(\bm{B}\bm{Z}'\bm{A}\bm{Z}),
\end{align*}
as in Equation~\eqref{eq:MD=MMD}, which implies that 
\begin{align*}
    \lambda_{pq}(\bm{A},\bm{B}) = \inf_{\bm{Z}\in \R^{p \times q}} \frac{\tr(\bm{B}\bm{Z}'\bm{A}\bm{Z})}{\tr(\bm{Z}'\bm{Z})} \quad \text{and} \quad \lambda_{1}(\bm{A},\bm{B}) = \sup_{\bm{Z}\in \R^{p \times q}} \frac{\tr(\bm{B}\bm{Z}'\bm{A}\bm{Z})}{\tr(\bm{Z}'\bm{Z})}.
\end{align*}
This leads us to the matrix-variate version of Equation~\eqref{eq:eigenvalue_inequality}, where for any matrix $\bm{X} \in \R^{p \times q}$
\begin{align}
    \norm{\bm{X}}_F^2 = \tr(\bm{X}'\bm{X}) \leq \tr(\bm{B}^{-1}\bm{X}'\bm{A}^{-1}\bm{X})\lambda_{1}(\bm{A},\bm{B}) = \tr(\bm{B}^{-1}\bm{X}'\bm{A}^{-1}\bm{X})\lambda_{1}(\bm{A}) \lambda_{1}(\bm{B})\label{eq:matrix_eigenvalue_inequality}
\end{align}

\begin{lemma} \label{lemma:simplex}
    Take $p,q \in \N$, $d = \lfloor \nicefrac{p}{q} + \nicefrac{q}{p} \rfloor$, $d + 2 \leq s \leq pq$, and matrices $\bm{X}_1,\ldots,\bm{X}_{s} \in \R^{p \times q}$ that are in general position, i.e., no subset of $r, 2 \leq r \leq s$ samples lies on an $r-2$ dimensional subspace. For an ellipsoid $E(\bm{T}, \bm{U}, \bm{V})$ as in Equation~\eqref{eq:matrix_ellipsoid}, containing the matrices $\bm{X}_1,\ldots,\bm{X}_{s}$, it holds that for every $C > 0$ there exists a constant $\alpha:= \alpha(\bm{X}_1,\ldots,\bm{X}_{s}) > 0$ only depending on $\bm{X}_1,\ldots,\bm{X}_{s}$ such that $\norm{\bm{T}}_F = \sqrt{\tr(\bm{T}'\bm{T})} > \alpha$ implies $\det(E(\bm{T}, \bm{U}, \bm{V})) > C$, i.e.,
    \begin{align*}
        \forall C > 0 ~ \exists \alpha > 0 : \norm{\bm{T}}_F > \alpha \implies \det(E(\bm{T}, \bm{U}, \bm{V})) > C.
    \end{align*}
\end{lemma}
\begin{proof}
    The samples $\bm{X}_1,\ldots,\bm{X}_{s}$ are in general position, which implies that they span a nonempty $s-1$ simplex. Since $E(\bm{T}, \bm{U}, \bm{V})$ contains those samples, it also contains the simplex spanned by those matrices. This implies that there exists a constant $a > 0$ only depending on $\bm{X}_1,\ldots,\bm{X}_{s}$, such that the length of $k$, $s-1 \leq k \leq pq$, of the $pq$ axes of the ellipsoid $E(\bm{T}, \bm{U}, \bm{V})$ is at least $a$, i.e., there are $k$ out of $pq$ indices $(i,j), 1 \leq i \leq p, 1 \leq j \leq q$ such that  
    \begin{align}
        \sqrt{\lambda_i(\bm{U})\lambda_j(\bm{V})} > a \label{eq:lemma_eigenvalues}.
    \end{align}
    In Equation~\eqref{eq:lemma_eigenvalues}, $\lambda_i(\bm{U}), i \in \{1,\dots,p\}$, are the eigenvalues of $\bm{U}$, and $\lambda_j(\bm{V}), j\in \{1,\dots,q\}$,  are the eigenvalues of $\bm{V}$. 
    For any matrix $\bm{X}$ contained in $E(\bm{T}, \bm{U}, \bm{V})$, Equations~\eqref{eq:matrix_ellipsoid} and \eqref{eq:matrix_eigenvalue_inequality} imply that
    \begin{align}
    \begin{split}
        \norm{\bm{X} - \bm{T}}_F^2
        &= \tr((\bm{X} - \bm{T})'(\bm{X} - \bm{T})) \\
        &\leq \tr(\bm{V}^{-1}(\bm{X} - \bm{T})'\bm{U}^{-1}(\bm{X} - \bm{T})) \lambda_1(\bm{U})\lambda_1(\bm{V})\\
        &\leq \lambda_1(\bm{U})\lambda_1(\bm{V}). \label{eq:lemma_norm}
    \end{split}        
    \end{align}
    Without loss of generality, we assume that the matrix of all zeros $\bm{0} \in \R^{p \times q}$ is contained in the ellipsoid $E(\bm{T}, \bm{U}, \bm{V})$, then Equation~\eqref{eq:lemma_norm} implies that $\norm{\bm{T}}_F^2 \leq \lambda_1(\bm{U})\lambda_1(\bm{V})$. Take $\alpha = C/(a^{pq - 1})$, then we have that 
    \begin{align*}
        \frac{C}{a^{pq - 1}} < \norm{\bm{T}}_F \leq \sqrt{\lambda_1(\bm{U})\lambda_1(\bm{V})} \Leftrightarrow C < \sqrt{\lambda_1(\bm{U})\lambda_1(\bm{V})}a^{pq - 1}
    \end{align*}
    and from Equaiton~\eqref{eq:lemma_eigenvalues} it follows that 
    \begin{align*}
    \begin{split}
        \det(E(\bm{T}, \bm{U}, \bm{V})) :&= \det(\bm{U})^{\nicefrac{q}{2}}\det(\bm{V})^{\nicefrac{p}{2}} \\
        &= \prod_{i = 1}^p\prod_{j = 1}^q\sqrt{\lambda_i(\bm{U})\lambda_j(\bm{V})} \\
        &> \sqrt{\lambda_1(\bm{U})\lambda_1(\bm{V})}a^{pq-1} > C.
    \end{split}
    \end{align*}
\end{proof}

\begin{proof}[\hypertarget{proof:MMCD_breakdown_point}{\textbf{Proof of Theorem~\ref{theorem:MMCD_breakdown_point}}}]
    We show that the breakdown points of the MMCD estimators of location and covariance defined in Equations~\eqref{eq:breakdown_point_location} and \eqref{eq:breakdown_point_covariance}, respectively, are both $\nicefrac{m}{n}$, with $m = \lfloor\min(n - h + 1,h - (d + 1))\rfloor$, $d = \lfloor \nicefrac{p}{q} + \nicefrac{q}{p} \rfloor$.    
    First, we prove that $\varepsilon^{\ast}(\hat{\bm{\M}}, \bm{\mathfrak{X}}) = \varepsilon^{\ast}(\hat{\bm{\Sigma}}^{\row}, \hat{\bm{\Sigma}}^{\col}, \bm{\mathfrak{X}}) \geq \nicefrac{m}{n}$. Let $\bm{\mathfrak{Y}}$ be the sample obtained by replacing at most $m -1$ matrices of $\bm{\mathfrak{X}}$ by arbitrary $p \times q$ matrices. 
    Since $n - (m-1) \geq h$, $\bm{\mathfrak{Y}}$ contains at least $h$ matrices of the orginial sample $\bm{\mathfrak{X}}$ and because $m - 1 \leq h - (d + 1)- 1$, every subset of size $h$ of $\bm{\mathfrak{Y}}$ includes at least $d + 2$ matrices of the original sample $\bm{\mathfrak{X}}$. Hence, the MMCD estimators can almost surely be computed for any $h$-subset of $\bm{\mathfrak{Y}}$. 
    Let us consider three ellipsoids:
    \begin{itemize}
        \item Let $E_{\max} = E(\bm{0},c_{\max} \bm{I}, \bm{I})$ denote the smallest sphere that contains all samples in $\bm{\mathfrak{X}}$, where $c_{\max}$ is chosen accordingly.
        \item Let $E_{h} = E(\bm{0},c_{h} \bm{I}, \bm{I})$ denote the smallest sphere that contains the $h$ samples of $\bm{\mathfrak{X}}$ that are also in $\bm{\mathfrak{Y}}$, where $c_{h}$ is chosen accordingly.
        \item Let $E_{\mmcd} = E(\hat{\bm{\M}}_{\bm{\mathfrak{Y}}}, \hat{\bm{\Sigma}}^{\row}_{\bm{\mathfrak{Y}}}, \hat{\bm{\Sigma}}^{\col}_{\bm{\mathfrak{Y}}})$ denote the MMCD ellipsoid.
    \end{itemize}
    It follows that $\det(E_{\mmcd}) \leq \det(E_{h}) \leq \det(E_{\max}) =: \alpha$, where for an ellipsoid $E=E(\bm{T},\bm{U},\bm{V})$, $\det(E)$ is defined in \eqref{eq:det(E)}. Note that $\bm{\mathfrak{X}}$ is a collection of random samples from a continuous distribution and therefore it is in general position almost surely. 
    Further, $E_{\mmcd}$ covers at least $h$ samples, and those include at least $d + 2$ samples of $\bm{\mathfrak{X}}$, which span a nonempty $d+1$ simplex. Lemma~\ref{lemma:simplex} shows that there exists a constant $\alpha > 0$ that only depends on those $d + 2$ samples such that, if $\norm{\hat{\bm{\M}}_{\bm{\mathfrak{Y}}}}_{F} > C$ it would imply $\det(E_{\mmcd}) > \alpha$. As shown above, this is not possible, hence $\norm{\hat{\bm{\M}}_{\bm{\mathfrak{Y}}}}_{F} \leq C$. 
    
    Similarly, since $\bm{\mathfrak{Y}}$ contains at least $d+2$ matrices of the original sample $\bm{\mathfrak{X}}$, the MMCD estimators almost surely yield positive definite covariance estimates $\hat{\bm{\Sigma}}^{\row}_{\bm{\mathfrak{Y}}}$ and $\hat{\bm{\Sigma}}^{\col}_{\bm{\mathfrak{Y}}}$.
    More specifically, let $\bm{\mathfrak{X}}_T$, $T \subseteq H$ be the subset of the at least $d + 2$ matrices of the original sample that are in $\bm{\mathfrak{Y}}$. Since $\abs{T} \geq d + 2 = \lfloor p/q+q/p\rfloor+2$ the MLE estimators $(\hat{\bm{M}}_{\bm{\mathfrak{X}}_T}, \hat{\bm{\Sigma}}_{\bm{\mathfrak{X}}_T}^{\row}, \hat{\bm{\Sigma}}_{\bm{\mathfrak{X}}_T}^{\row})$ of this subsample are almost surely positive definite. Let $E_T = E(\hat{\bm{M}}_{\bm{\mathfrak{X}}_T}, \hat{\bm{\Sigma}}_{\bm{\mathfrak{X}}_T}^{\row}, \hat{\bm{\Sigma}}_{\bm{\mathfrak{X}}_T}^{\row})$ denote the corresponding ellipsoid which is the smallest ellipsoid, of the type $E=E(\textbf{T},\textbf{U},\textbf{V})$ as in Equation~(\ref{eq:matrix_ellipsoid}), containing the samples $\bm{\mathfrak{X}}_T$ as one can think of it as the MMCD ellipsoid for those $\abs{T} \geq d+2$ samples with $H = T$. 
        This further implies that the volume of the corresponding ellipsoid $E_T$ is bounded from below by a constant only depending on $\bm{\mathfrak{X}}$, i.e. 
    $\det(E_T)\geq v>0$. As $E_{\mmcd}$ is also an ellipsoid containing the samples  $\bm{\mathfrak{X}}_T$, $\det(E_{\mmcd})\geq \det(E_T)\geq v>0$. Moreover, it also means that there exists a constant $k$ depending only on $\bm{\mathfrak{X}}$, such that $E_T \subseteq kE_{\mmcd}$,
    implying that there exists a constant $\gamma>0$ depending only on $\bm{\mathfrak{X}}$, such that $\lambda_i(\hat{\bm{\Sigma}}^{\row}_{\bm{\mathfrak{Y}}}) \lambda_j(\hat{\bm{\Sigma}}^{\col}_{\bm{\mathfrak{Y}}}) > \gamma, 1 \leq i \leq p, 1 \leq j \leq q$. Especially, 
    $\lambda_p(\hat{\bm{\Sigma}}^{\row
    }_{\bm{\mathfrak{Y}}}) \lambda_q(\hat{\bm{\Sigma}}^{\col}_{\bm{\mathfrak{Y}}}) > \gamma$. 
    Since also $\det(E_{\mmcd}) \leq \alpha$ there exists a constant $\delta > 0$, depending only on $\bm{\mathfrak{X}}$ such that $\lambda_i(\hat{\bm{\Sigma}}^{\row}_{\bm{\mathfrak{Y}}})\lambda_j(\hat{\bm{\Sigma}}^{\col}_{\bm{\mathfrak{Y}}}) < \delta, 1 \leq i \leq p, 1 \leq j \leq q$.   
    
    Next we show that $\varepsilon^{\ast}(\hat{\bm{\M}}, \bm{\mathfrak{X}}) = \varepsilon^{\ast}(\hat{\bm{\Sigma}}^{\row}, \hat{\bm{\Sigma}}^{\col}, \bm{\mathfrak{X}}) \leq \nicefrac{m}{n}$. If $m = n - h + 1$, we replace $m = n - h + 1 $ matrices of $\bm{\mathfrak{X}}$ to obtain $\bm{\mathfrak{Y}}$, then $n - m = h - 1$, implying that every subset of $h$ samples of $\bm{\mathfrak{Y}}$ contains at least one conaminated sample. Hence, $E_{\mmcd} = E(\hat{\bm{\M}}_{\bm{\mathfrak{Y}}}, \hat{\bm{\Sigma}}^{\row}_{\bm{\mathfrak{Y}}}, \hat{\bm{\Sigma}}^{\col}_{\bm{\mathfrak{Y}}})$ also includes at least one contaminated sample. Let $\norm{\bm{X}}_F \rightarrow \infty$ for all contaminated samples $\bm{X}$, then at least one eigenvalue of $E_{\mmcd}$ explodes and the MMCD location and covariance estimators break down. 
    Finally, consider the case where $m = h - (d+1)$. To construct $\bm{\mathfrak{Y}}$, take any $d + 1$ samples of $\bm{\mathfrak{X}}$ and consider the $d$ dimensional hyperplane $L$  they determine. Replace $h - (d+1)$ samples that are not in $L$ and replace them with matrices on $L$. Then $L$ contains $h$ points of $\bm{\mathfrak{Y}}$ and the ellipsoid covering those points has volume zero and hence determinant zero. 
    Since $\bm{\mathfrak{X}}$ is in general position, we can construct $\bm{\mathfrak{Y}}$ such that no other lower dimensional hyperplane contains $h$ points of $\bm{\mathfrak{Y}}$. Hence, $\hat{\bm{\M}}_{\bm{\mathfrak{Y}}}$ lies on $L$ and $E_{\mmcd} = E(\hat{\bm{\M}}_{\bm{\mathfrak{Y}}}, \hat{\bm{\Sigma}}^{\row}_{\bm{\mathfrak{Y}}}, \hat{\bm{\Sigma}}^{\col}_{\bm{\mathfrak{Y}}})$ has zero determinant. This implies that at least one eigenvalue is zero, hence the MMCD location and covariance estimators break down. 
\end{proof}

\begin{proof}[\hypertarget{proof:consistency}{\textbf{Proof of Theorem~\ref{theorem:consistency}}}]
    Let $(\bm{X}_1,\dots,\bm{X}_n)$ be a sample of matrix-variate observations and $(\bm{x}_1,\dots,\bm{x}_n)$, $\bm{x}_i=\mathrm{vec}(\bm{X}_i)$, $i=1,\dots n$ its vectorized form. The MCD estimator can also be found as a solution to the following maximization problem: 
    \begin{align*}
        \max_{\bm{w}, \hat{\bm{\mu}}, \hat{\bm{\Sigma}}} l(\bm{w}, \hat{\bm{\mu}}, \hat{\bm{\Sigma}}|(\bm{x}_1,\dots,\bm{x}_n)) = 
        \max_{\bm{w}, \hat{\bm{\mu}}, \hat{\bm{\Sigma}}} 
        -\frac{1}{2} \sum_{i = 1}^n w_i (\ln(\det(\hat{\bm{\Sigma}})) + pq \ln(2\pi) + \md^2(\bm{x}_i,\hat{\bm{\mu}}, \hat{\bm{\Sigma}}))
    \end{align*}
    subject to $w_1,\dots,w_n \in \{0,1\}$, $\sum_{i = 1}^n w_i = h$, $\hat{\bm{\mu}} \in \R^{pq}$, $\hat{\bm{\Sigma}} \in \pds(pq)$; see \cite{Raymaekers2022} for more insight. 
    Similarly, the MMCD estimator is a solution to the following  maximization problem: 
    \begin{align*}
        &\max_{\bm{w}, \hat{\bm{\M}}, \hat{\bm{\Sigma}}^{\row}, \hat{\bm{\Sigma}}^{\col}} l(\bm{w},\hat{\bm{\M}}, \hat{\bm{\Sigma}}^{\row}, \hat{\bm{\Sigma}}^{\col}|(\bm{X}_1,\dots,\bm{X}_n)) \\
        =&\max_{\bm{w}, \hat{\bm{\M}}, \hat{\bm{\Sigma}}^{\row}, \hat{\bm{\Sigma}}^{\col}}
        -\frac{1}{2} \sum_{i=1}^n w_i \Bigl( 
        p\ln(\det(\hat{\bm{\Sigma}}^{\col})) + q\ln(\det(\hat{\bm{\Sigma}}^{\row})) 
        + \mmd^2(\bm{X}_i) + pq\ln(2\pi) \Bigr)
    \end{align*}
    subject to $w_1,\dots,w_n \in \{0,1\}$, $\sum_{i = 1}^n w_i = h$, $\hat{\bm{\M}} \in \R^{p \times q}$, $\hat{\bm{\Sigma}}^{\row} \in \pds(p)$, $\hat{\bm{\Sigma}}^{\col} \in \pds(q)$; see Proposition~\ref{proposition:MCD_likelihood}. 
    
    Denote further $(\bm{w}_{\mcd},\hat{\bm{\mu}}_{\mcd},\hat{\bm{\Sigma}}_{\mcd})$ and $(\bm{w}_{\mmcd},\hat{\bm{\mu}}_{\mmcd},\hat{\bm{\Sigma}}_{\mmcd}^{\col} \otimes \hat{\bm{\Sigma}}_{\mmcd}^{\row})$ weights, mean and covariance estimators for the vectorized sample $(\bm{x}_1,\dots,\bm{x}_n)$, based on MCD and MMCD, respectively. 
    As $\bm{X}_i\sim\mathcal{ME}(\bm{\M}, \bm{\Sigma}^{\row}, \bm{\Sigma}^{\col},g)$, then $\bm{x}_i\sim\mathcal{E}(\bm{\mu}, \bm{\Sigma}^{\col} \otimes \bm{\Sigma}^{\row},g)$, with $\mathbb{E}(\bm{x}_i)=\bm{\mu} = \vect(\bm{\M})$, $\cov{(\bm{x}_i)}=c_g\bm{\Sigma}^{\col} \otimes \bm{\Sigma}^{\row}$, where $c_g$ is a distribution-specific scaling parameter; for more details see Theorem 2.11 in \cite{gupta2012elliptically}.    Moreover, the mean estimator $\hat{\bm{\mu}}_{\mcd}$ and properly scaled covariance estimator $\hat{\bm{\Sigma}}_{\mcd}$ are strongly consistent for the population counterparts $\bm{\mu}$ and $\bm{\Sigma}^{\col} \otimes \bm{\Sigma}^{\row}$; see e.g. \cite{croux1999influence} and \cite{cator2012central}. Especially, this implies that for every $\delta>0$ there exists $n\in \mathbb{N}$ such that
    \begin{equation*}\label{eq:MCD_convergence}
        \norm{\hat{\bm{\mu}}_{\mcd}-\bm{\mu}} \overset{a.s.}{<} \delta,\quad \norm{\hat{\bm{\Sigma}}_{\mcd}-\bm{A}\otimes\bm{B}} \overset{a.s.}{<} \delta,
    \end{equation*}
    for some $\bm{A}\otimes\bm{B}\in \pds(p)\otimes\pds(q)$. 
    In the following, we will drop ${a.s.}$ superscript from (in)equality signs when it is clear from the context. For fixed weights $\bm{w}$, the log-likelihood function 
    $l_{\cdot,\bm{w}|(\bm{x}_1,\dots,\bm{x}_n)}:(\bm{\mu},\bm{\Sigma})\mapsto l(\bm{\mu},\bm{\Sigma}|\bm{w},(\bm{x}_1,\dots,\bm{x}_n))$ 
    is continuous in both $\bm{\mu}$, and $\bm{\Sigma}$, and its continuity implies
    \begin{equation*}
        \abs{l(\bm{w}_{\mcd},\hat{\bm{\mu}}_{\mcd},\hat{\bm{\Sigma}}_{\mcd}|(\bm{x}_1,\dots,\bm{x}_n))-l(\bm{w}_{\mcd},\bm{\mu},\bm{A}\otimes\bm{B}|(\bm{x}_1,\dots,\bm{x}_n))}<\varepsilon,
    \end{equation*}
    for $\varepsilon=\varepsilon(\delta)>0$. Moreover, the solution $(\bm{w}_{\mcd},\hat{\bm{\mu}}_{\mcd},\hat{\bm{\Sigma}}_{\mcd})$ is optimal for $l(\cdot|(\bm{x}_1,\dots,\bm{x}_n))$, implying that 
    \begin{equation}\label{eq:consistency_proof_mcd}
        0<l(\bm{w}_{\mcd},\hat{\bm{\mu}}_{\mcd},\hat{\bm{\Sigma}}_{\mcd}|(\bm{x}_1,\dots,\bm{x}_n))-l(\bm{w}_{\mcd},\bm{\mu},\bm{A}\otimes\bm{B}|(\bm{x}_1,\dots,\bm{x}_n))<\varepsilon.
    \end{equation}
    Similarly, $(\bm{w}_{\mmcd},\hat{\bm{\mu}}_{\mmcd},\hat{\bm{\Sigma}}_{\mmcd}^{\col} \otimes \hat{\bm{\Sigma}}_{\mmcd}^{\row})$ is a maximizer of $l(\cdot|(\bm{x}_1,\dots,\bm{x}_n))$ in the set of all feasible weights, means, and covariances with Kronecker product structure. As $(\bm{w}_{\mcd},\bm{\mu},\bm{A}\otimes\bm{B})$ belongs to the same set,   
    \begin{equation*}
        l(\bm{w}_{\mmcd}, \hat{\bm{\mu}}_{\mmcd}, \hat{\bm{\Sigma}}_{\mmcd}^{\col} \otimes \hat{\bm{\Sigma}}_{\mmcd}^{\row} | (\bm{x}_1,\dots,\bm{x}_n)) > l(\bm{w}_{\mcd},\bm{\mu},\bm{A}\otimes\bm{B}| (\bm{x}_1,\dots,\bm{x}_n)).
    \end{equation*}        
    Denote further $\hat{\bm{S}}_{\mmcd}=\frac{1}{h}\sum_{i=1}^nw_{{\mmcd},i}(\bm{x}_i-\hat{\bm{\mu}}_{\mathrm{MMCD}})(\bm{x}_i-\hat{\bm{\mu}}_{\mathrm{MMCD}})'$ to be the estimate of ${\bm{\Sigma}}^{\col} \otimes {\bm{\Sigma}}^{\row}$, based on the weights (subset) produced by the MMCD algorithm. As $\hat{\bm{S}}_{\mmcd}$ is optimal for $l$ given fixed weights $\bm{w}_\mathrm{MMCD}$,  
    \begin{align}\label{eq:consistency_proof_mmcd}            l(\bm{w}_{\mcd},\hat{\bm{\mu}}_{\mcd},\hat{\bm{\Sigma}}_{\mcd}|(\bm{x}_1,\dots,\bm{x}_n))
    &>l(\bm{w}_{\mmcd}, \hat{\bm{\mu}}_{\mmcd}, \hat{\bm{S}}_{\mmcd}|(\bm{x}_1,\dots,\bm{x}_n))\nonumber\\   
    &>l(\bm{w}_{\mmcd}, \hat{\bm{\mu}}_{\mmcd}, \hat{\bm{\Sigma}}_{\mmcd}^{\col} \otimes \hat{\bm{\Sigma}}_{\mmcd}^{\row} | (\bm{x}_1,\dots,\bm{x}_n))\nonumber\\    
    &>l(\bm{w}_{\mcd},\bm{\mu},\bm{A}\otimes\bm{B}|(\bm{x}_1,\dots,\bm{x}_n)).
    \end{align}  
    \eqref{eq:consistency_proof_mcd} and \eqref{eq:consistency_proof_mmcd} now give that
    \begin{equation*}
        0<l(\bm{w}_{\mcd},\hat{\bm{\mu}}_{\mcd},\hat{\bm{\Sigma}}_{\mcd}|(\bm{x}_1,\dots,\bm{x}_n))-l(\bm{w}_{\mmcd},\hat{\bm{\mu}}_{\mmcd},\hat{\bm{S}}_{\mmcd} |(\bm{x}_1,\dots,\bm{x}_n))<\varepsilon,
    \end{equation*}
    i.e., due to Proposition~\ref{proposition:MCD_likelihood},
    \begin{equation}\label{eq:consistency_proof_det}
        0<\det(\hat{\bm{S}}_{\mmcd})-\det(\hat{\bm{\Sigma}}_{\mcd})<\varepsilon,
    \end{equation}
    for $\varepsilon=\varepsilon(n)>0$, arbitrarily small $(\varepsilon(n)\to 0, \, n\to \infty)$. As both $\hat{\bm{\Sigma}}_{\mcd}$ and $\hat{\bm{S}}_{\mmcd}$ are weighted sample covariances for the random sample of vectorized observations calculated using the weights satisfying the same constraints, Corollary 4.1. in \cite{cator2012central} (taking $P_t$ to be the empirical measure based on the sample $(\bm{x}_1,\dots,\bm{x}_n)$) implies that 
    \begin{equation}\label{eq:44}
     \hat{\bm{\mu}}_{\mmcd} \xrightarrow{a.s.} \bm{\mu},\quad  \hat{\bm{S}}_{\mmcd} \xrightarrow{a.s.} c(\alpha)^{-1} \bm{\Sigma}^{\col} \otimes \bm{\Sigma}^{\row} \quad,  
    \end{equation} 
    where $c(\alpha)>0$ is a distribution-specific consistency factor of the MCD given in \cite{croux1999influence}.
    
    To complete the proof consider reparametrization of $l(\bm{w},\bm{a},\bm{A}|(\bm{x}_1,\dots,\bm{x}_n))$ for fixed weights $\bm{w} \in \R^{n}$, mean $\bm{a} \in \R^{pq}$, and covariance $\bm{A} \in \pds(pq)$ in terms of the precicion matrix $\bm{B}=\bm{A}^{-1}$. Denote this new parametrization as $g(\bm{B}|\bm{w},\bm{a},(\bm{x}_1,\dots,\bm{x}_n))=l(\bm{w},\bm{a},\bm{B}^{-1}|\bm{x}_1,\dots,\bm{x}_n)$, which is now concave in $\bm{B}$. 
    Especially, for $\bm{w}=\bm{w}_{\mmcd}$ and $\bm{a}=\hat{\bm{\mu}}_{\mmcd}$, the function $g(\bm{B}|\bm{w}_{\mmcd},\hat{\bm{\mu}}_{\mmcd},
    (\bm{x}_1,\dots,\bm{x}_n))$ is concave in $\bm{B}$ and achieves a unique global maximum at $\bm{B}=\hat{\bm{S}}_{\mmcd}$. 
    Equations~\eqref{eq:consistency_proof_mcd} and \eqref{eq:consistency_proof_mmcd} then give
    \begin{align*}
        0 < &l(\bm{w}_{\mmcd},\hat{\bm{\mu}}_{\mmcd},\hat{\bm{S}}_{\mmcd}|(\bm{x}_1,\dots,\bm{x}_n))\\
        &-l(\bm{w}_{\mmcd},\hat{\bm{\mu}}_{\mmcd}, \hat{\bm{\Sigma}}_{\mmcd}^{\col} \otimes \hat{\bm{\Sigma}}_{\mmcd}^{\row}|(\bm{x}_1,\dots,\bm{x}_n))<\varepsilon,
    \end{align*}
    further implying that 
    \begin{align*}
        0<&g(\hat{\bm{S}}_{\mmcd}^{-1}|\bm{w}_{\mmcd},\hat{\bm{\mu}}_{\mmcd},(\bm{x}_1,\dots,\bm{x}_n))\\
        &-g((\hat{\bm{\Sigma}}_{\mmcd}^{\col} \otimes \hat{\bm{\Sigma}}_{\mmcd}^{\row})^{-1}|\bm{w}_{\mmcd},\hat{\bm{\mu}}_{\mmcd},(\bm{x}_1,\dots,\bm{x}_n))<\varepsilon,
    \end{align*}
    as both $\hat{\bm{S}}_{\mmcd}$ and 
    $\hat{\bm{\Sigma}}_{\mmcd}^{\col} \otimes \hat{\bm{\Sigma}}_{\mmcd}^{\row}$ are a.s. positive definite for $n$ large enough. Concavity of $g$ and the fact that  $\hat{\bm{S}}_{\mmcd}^{-1}$ is its global maximum further imply that 
    \begin{equation*}
        \|\hat{\bm{S}}_{\mmcd}^{-1}-(\hat{\bm{\Sigma}}_{\mmcd}^{\col} \otimes \hat{\bm{\Sigma}}_{\mmcd}^{\row})^{-1}\|<\delta_1,
    \end{equation*}
    for $\delta_1=\delta_1(\varepsilon)\to 0$ as $n\to \infty$. Almost sure positive definiteness of $\hat{\bm{S}}_{\mmcd}$ and $\hat{\bm{\Sigma}}_{\mmcd}^{\col} \otimes \hat{\bm{\Sigma}}_{\mmcd}^{\row}$, and continuity of matrix inverse imply that 
    \begin{equation}\label{eq:45}
     \|\hat{\bm{S}}_{\mmcd} - \hat{\bm{\Sigma}}_{\mmcd}^{\col} \otimes \hat{\bm{\Sigma}}_{\mmcd}^{\row}\| < \delta,   
    \end{equation}
    for $\delta=\delta(\varepsilon)\to 0$ as $n\to \infty$. Equations~\eqref{eq:44} and \eqref{eq:45} now complete the proof. Observe that the proof indicates that the distribution-specific consistency factor is inherited from the MCD covariance estimator; see \cite{croux1999influence}. 
\end{proof}

\begin{proof}[\hypertarget{proof:MMCD_breakdown_point_reweighted}{\textbf{Proof of Theorem~\ref{theorem:MMCD_reweighted_breakdown_point}}}]
We show that the breakdown points of the reweighted MMCD estimators are at least as high as the breakdown points of the raw MMCD estimators. 
Let $\bm{\mathfrak{Y}}$ be the sample obtained by replacing at most $m -1$ matrices of $\bm{\mathfrak{X}}$ by arbitrary $p \times q$ matrices. Let $\hat{\bm{\M}}_{\bm{\mathfrak{Y}}}$, $\hat{\bm{\Sigma}}^{\row}_{\bm{\mathfrak{Y}}}$, and $\hat{\bm{\Sigma}}^{\col}_{\bm{\mathfrak{Y}}}$ 
    denote the \emph{raw} MMCD estimators and 
    $\tilde{\bm{\M}}_{\bm{\mathfrak{Y}}}$, $\tilde{\bm{\Sigma}}^{\row}_{\bm{\mathfrak{Y}}}$, and $\tilde{\bm{\Sigma}}^{\col}_{\bm{\mathfrak{Y}}}$ 
    denote the \emph{reweighted} MMCD estimators based on the corrupted sample $\bm{\mathfrak{Y}}$. 
    Further, $d(\bm{Y}_i) = \mmd(\bm{Y}_i;\hat{\bm{\M}}_{\bm{\mathfrak{Y}}},\hat{\bm{\Sigma}}^{\row}_{\bm{\mathfrak{Y}}},\hat{\bm{\Sigma}}^{\col}_{\bm{\mathfrak{Y}}})$, $i \in N = \{1,\dots,n\}$, denote the matrix Mahalanbois distances of the corrupted sample based on the \emph{raw} MMCD estimators. 
    Since $m \leq \varepsilon^{\ast}(\hat{\bm{\M}}_{\bm{\mathfrak{X}}}, \bm{\mathfrak{X}}) - 1 = \varepsilon^{\ast}(\hat{\bm{\Sigma}}^{\row}_{\bm{\mathfrak{X}}}, \hat{\bm{\Sigma}}^{\col}_{\bm{\mathfrak{X}}}, \bm{\mathfrak{X}}) -1$ it follows that there exist constants $k_0$, $k_1$, and $k_2$ that only depend on $\bm{\mathfrak{X}}$, such that 
    \begin{align} \label{eq:eigenvalue_bounds}
    \begin{split}
        &\norm{\hat{\bm{\M}}_{\bm{\mathfrak{Y}}}} \leq k_0 < \infty 
        \quad \text{and} \\
        &0 < k_1 < 
        \lambda_{p}(\hat{\bm{\Sigma}}^{\row}_{\bm{\mathfrak{Y}}}) \lambda_{q}(\hat{\bm{\Sigma}}^{\col}_{\bm{\mathfrak{Y}}}) \leq 
        \lambda_{1}(\hat{\bm{\Sigma}}^{\row}_{\bm{\mathfrak{Y}}}) \lambda_{1}(\hat{\bm{\Sigma}}^{\col}_{\bm{\mathfrak{Y}}}) \leq k_2 < \infty.
    \end{split}
    \end{align}
    Since at least $\lfloor\nicefrac{(n+d+2)}{2}\rfloor$ have a positive weight and at most $\lfloor \nicefrac{(n-d)}{2} \rfloor - 1$ observations are replaced, there are at least $d + 2$ observations of the orginal sample $\bm{\mathfrak{X}}$ contained in $\bm{\mathfrak{Y}}$ that have a positive weight. 
    Let $T \subseteq N$ denote the indices of those samples, then we have that
    \begin{align} \label{eq:denominator_bound}
        \sum_{i = 1}^n w(d(\bm{Y}_i)) = \sum_{i \in N \setminus T} w(d(\bm{Y}_i)) + \sum_{i \in T} w(d(\bm{X}_i)) \geq \sum_{i \in T} w(d(\bm{X}_i)) \geq (d+2) c_0 > 0,
    \end{align}
    with $c_0 := \min_{i \in T}w(d(\bm{X}_i)) > 0$. This implies that the denominators of $\tilde{\bm{\M}}_{\bm{\mathfrak{Y}}}$, $\tilde{\bm{\Sigma}}^{\row}_{\bm{\mathfrak{Y}}}$, and $\tilde{\bm{\Sigma}}^{\col}_{\bm{\mathfrak{Y}}}$ are always positive. 
    
    Let us now show that there exists a constant $\alpha_0 < \infty$ only dependent on $\bm{\mathfrak{X}}$ such that $\norm{\tilde{\bm{\M}}_{\bm{\mathfrak{Y}}}}_{F} < \alpha_0$. 
    From Equation~\eqref{eq:matrix_eigenvalue_inequality} we have that
    \begin{align*}
        \norm{\bm{Y}_i - \hat{\bm{\M}}_{\bm{\mathfrak{Y}}}}_F^2 &
        \leq \tr(\hat{\bm{\Omega}}^{\col}_{\bm{\mathfrak{Y}}} (\bm{Y}_i - \hat{\bm{\M}}_{\bm{\mathfrak{Y}}})' \hat{\bm{\Omega}}^{\row}_{\bm{\mathfrak{Y}}} (\bm{Y}_i - \hat{\bm{\M}}_{\bm{\mathfrak{Y}}})) \lambda_1(\hat{\bm{\Sigma}}^{\row}_{\bm{\mathfrak{Y}}}) \lambda_1(\hat{\bm{\Sigma}}^{\col}_{\bm{\mathfrak{Y}}})\\
        &= d(\bm{Y}_i) \lambda_1(\hat{\bm{\Sigma}}^{\row}_{\bm{\mathfrak{Y}}}) \lambda_1(\hat{\bm{\Sigma}}^{\col}_{\bm{\mathfrak{Y}}}).
    \end{align*}
    When computing $\tilde{\bm{\M}}_{\bm{\mathfrak{Y}}}$ we have that $w(d(\bm{Y}_i)) = 0$ if $d(\bm{Y}_i) > c_1$ and for all $\bm{Y}_i \in \bm{\mathfrak{Y}}$ that are assigned positive weights, Equation~\eqref{eq:eigenvalue_bounds} yields
    \begin{align}\label{eq:mmcd_reweighted_obs_bound}
        \norm{\bm{Y}_i}_F^2 \leq \norm{\bm{Y}_i - \hat{\bm{\M}}_{\bm{\mathfrak{Y}}}}_F^2 + \norm{\hat{\bm{\M}}_{\bm{\mathfrak{Y}}}}_F^2 \leq c_1 k_2 + k_0^2.
    \end{align}
    Since the denominator of is $\tilde{\bm{\M}}_{\bm{\mathfrak{Y}}}$ bounded according to Equation~\eqref{eq:denominator_bound}, $w$ is non-increasing and bounded, and $k_0$ and $k_2$ are only dependent on $\bm{\mathfrak{X}}$, there exsits a constant $\alpha_0$ only dependent on $\bm{\mathfrak{X}}$ such that 
    \begin{align}\label{eq:eq:mmcd_reweighted_mean_bound}
        \norm{\tilde{\bm{\M}}_{\bm{\mathfrak{Y}}}}_F \leq \alpha_0 < \infty.
    \end{align}
    To show that the covariance does not break down, we first consider the case of the weight function $w(d_i)=\mathds{1}(d_i\leq c_1)$, for $c_1>0$. 
    Let $S \subseteq \{1,\dots,N\}$ denote the subset of indices of the $s = \abs{S}$ samples of $\bm{\mathfrak{Y}}=\{\bm{Y}_1,\dots,\bm{Y}_n\}$ for which $d_i \leq c_1, i \in S$.    
    Observe that the $h$ samples $\bm{Y}_i$, $i\in H$ are those with the smallest MD, hence $T \subseteq H \subseteq S$. 
    Let $\hat{\bm{M}}_{\bm{\mathfrak{Y}}_T}, \hat{\bm{\Sigma}}_{\bm{\mathfrak{Y}}_T}, \hat{\bm{\Omega}}_{\bm{\mathfrak{Y}}_T}$ and $\hat{\bm{M}}_{\bm{\mathfrak{Y}}_S}, \hat{\bm{\Sigma}}_{\bm{\mathfrak{Y}}_S}, \hat{\bm{\Omega}}_{\bm{\mathfrak{Y}}_S}$ denote the MLE estimators of $\bm{\mathfrak{Y}}_T = (\bm{Y}_i)_{i \in T}$ and $\bm{\mathfrak{Y}}_S = (\bm{Y}_i)_{i \in S}$, respectively. 
    Consider the following three ellipsoids:
    \begin{itemize}
        \item Let $E_T = E(\hat{\bm{M}}_{\bm{\mathfrak{Y}}_T}, \hat{\bm{\Sigma}}_{\bm{\mathfrak{Y}}_T}, \hat{\bm{\Omega}}_{\bm{\mathfrak{Y}}_T})$ denote the ellipsoid corresponding to the MLEs of $\bm{\mathfrak{Y}}_T$, i.e., the smallest ellipsoid containing those at least $d+2$ samples.
        \item Let $E_S = E(\hat{\bm{M}}_{\bm{\mathfrak{Y}}_S}, \hat{\bm{\Sigma}}_{\bm{\mathfrak{Y}}_S}, \hat{\bm{\Omega}}_{\bm{\mathfrak{Y}}_S})$ denote the ellipsoid corresponding to the MLEs of $\bm{\mathfrak{Y}}_S$.
        \item Let $E_0=E(\bm{0},k\bm{I}_p,\bm{I}_q)$ denote the smallest sphere containing the samples $\bm{\mathfrak{Y}}_S$, where $k=c_1k_2+k_0^2$ is as in \eqref{eq:mmcd_reweighted_obs_bound}.
    \end{itemize}    
    Observe first that as $E_T$ is the smallest ellipsoid containing the samples $\bm{\mathfrak{Y}}_T = \bm{\mathfrak{X}}_T$ that are also in $E_S$, there exists a constant $a_1$ depending only on $\bm{\mathfrak{X}}_T$, such that $E_T\subseteq a_1 E_S: = E(\hat{\bm{M}}_{\bm{\mathfrak{Y}}_S}, a_1\hat{\bm{\Sigma}}_{\bm{\mathfrak{Y}}_S}, \hat{\bm{\Omega}}_{\bm{\mathfrak{Y}}_S})$. On the other hand, $E_S$ is the smallest ellipsoid containing $\bm{\mathfrak{Y}}_S$. As these points are also in $E_0$, then there exist $\alpha=\alpha(c_1)$ such that $\det(E_S)\leq \det(E_0)\leq \alpha$. Equivalent argumentation as in the proof of Theorem \ref{theorem:MMCD_breakdown_point} completes the first part of the proof. 

    Let now $w=w(d_i)$ be an arbitrarily, nondecreasing, bounded weight function, such that $w(d_i)=0$ if $d_i> c_1$, $i=1,\dots,n$. The weighted log-likelihood function for the sample $\bm{\mathfrak{Y}}$, with the weights satisfying $\sum_{i=1}^n w_i=s$ is given by
    \begin{align*}
        l(\bm{w},\bm{\M}, \bm{\Sigma}^{\row}, \bm{\Sigma}^{\col}|\bm{\mathfrak{Y}}) =& -\frac{1}{2} \sum_{i=1}^m w_i \Bigl(p\ln(\det(\bm{\Sigma}^{\col})) + q\ln(\det(\bm{\Sigma}^{\row}))\\
        & + \tr(\bm{\Omega}^{\col}(\bm{Y}_i -\bm{\M})' \bm{\Omega}^{\row} (\bm{Y}_i -\bm{\M})) + pq\ln(2\pi) \Bigr)\\  
        =&-\frac{1}{2} \Bigl(s\bigl(p\ln(\det(\bm{\Sigma}^{\col})) + q\ln(\det(\bm{\Sigma}^{\row}))\bigr)\\
        & +\sum_{i=1}^s\tr(\bm{\Omega}^{\col}(\bm{Z}_i -\bm{\M})' \bm{\Omega}^{\row} (\bm{Z}_i -\bm{\M})) + pq\ln(2\pi) \Bigr)\\ 
        =&\, l(\tilde{\bm{w}},\bm{M},\bm{\Sigma}^{\row}, \bm{\Sigma}^{\col}|\bm{\mathfrak{Z}}),
    \end{align*}
    where $\tilde{\bm{w}} = (\tilde{w}(d_1),\dots,\tilde{w}(d_n))$, the new weight function satisfies $\tilde{w}(d_i)=\mathds{1}(d_1\leq c_1)$, $\bm{\mathfrak{Z}}=\{\bm{Z}_1,\dots,\bm{Z}_n\}$, and $\bm{Z}_i=\sqrt{w_i}\bm{Y}_i$, $i=1,\dots,n$. To complete the proof it is sufficient to observe the following: $\bm{Z}_1,\dots,\bm{Z}_{h}$ contains at least $d+2$ points of the form $\sqrt{w_i}\bm{X}_i$ and are in a general position, as $w_i\geq a_2>0$, for some constant depending only on $\bm{\mathfrak{X}}$. Moreover, $\|\bm{Z}_i\|_F^2=w_i\|\bm{Y}_i\|_F^2\leq w_i(c_1k_2+k_0^2)\leq w(0)(c_1k_2+k_0^2)$, $i=1,\dots , s$. The statement now follows from the first part of the proof, observing that assumption $\sum_{i=1}^mw_i=s$ without loss of generality, since $0 < w(0) \leq \sum_{i=1}^n w_i\leq s w(0)<\infty$.  
\end{proof}

\newpage
\section{MMCD algorithm}
\label{supplement:mmcd_algorithm}
\begin{algorithm}
    \caption{Iterative C-step procedure for the MMCD estimators}
    \label{algorithm:MMCD_C_step}
    \begin{algorithmic}[1] 
        \Procedure{CSTEP}{$(\bm{X}_1,\ldots,\bm{X}_n), H_{\old}, \varepsilon > 0$} 
            \State $(\hat{\bm{\M}}_{H_{\new}}, \hat{\bm{\Sigma}}^{\row}_{H_{\new}}, \hat{\bm{\Sigma}}^{\col}_{H_{\new}}) = \text{MLE}((\bm{X}_i)_{i \in H_{\old}})$
            \State $h = \abs{H_{\old}}$
            \Repeat
                \State $(\hat{\bm{\M}}_{H_{\old}}, \hat{\bm{\Sigma}}^{\row}_{H_{\old}}, \hat{\bm{\Sigma}}^{\col}_{H_{\old}}) = (\hat{\bm{\M}}_{H_{\new}}, \hat{\bm{\Sigma}}^{\row}_{H_{\new}}, \hat{\bm{\Sigma}}^{\col}_{H_{\new}})$
                \State $\bm{d} = (\mmd^2(\bm{X}_1; \hat{\bm{\M}}_{H_{\old}}, \hat{\bm{\Sigma}}^{\row}_{H_{\old}}, \hat{\bm{\Sigma}}^{\col}_{H_{\old}}),\ldots,\mmd^2(\bm{X}_n; \hat{\bm{\M}}_{H_{\old}}, \hat{\bm{\Sigma}}^{\row}_{H_{\old}}, \hat{\bm{\Sigma}}^{\col}_{H_{\old}}))$
                \State $\pi_{1}(i) = \{\{1,\ldots,n\} \rightarrow \{1,\ldots,n\} : i \mapsto j: d_{\pi(1)} \leq \ldots \leq d_{\pi(n)}\}$ 
                \State $H_{\new} = \{\pi(1),\pi(2),\ldots,\pi(h)\}$
                \State $(\hat{\bm{\M}}_{H_{\new}}, \hat{\bm{\Sigma}}^{\row}_{H_{\new}}, \hat{\bm{\Sigma}}^{\col}_{H_{\new}}) = \text{MLE}((\bm{X}_i)_{i \in H_{\new}})$
            \Until{$\abs{p(\ln(\det(\hat{\bm{\Sigma}}^{\col}_{H_{\old}})) - \ln(\det(\hat{\bm{\Sigma}}^{\col}_{H_{\new}}))) + q(\ln(\det(\hat{\bm{\Sigma}}^{\row}_{H_{\old}}))-\ln(\det(\hat{\bm{\Sigma}}^{\row}_{H_{\new}})))}~<~\varepsilon$}
            \State \textbf{return} $\hat{\bm{\M}}_{H_{\new}}, \hat{\bm{\Sigma}}^{\row}_{H_{\new}}, \hat{\bm{\Sigma}}^{\col}_{H_{\new}}, \bm{d}, H_{\new}$
        \EndProcedure
    \end{algorithmic}
\end{algorithm}

\begin{algorithm}
    \caption{Fast \emph{reweighted} MMCD procedure}
    \label{algorithm:MMCD}
    \begin{algorithmic}[1] 
        \Procedure{MMCD}{$\bm{\mathfrak{X}} = (\bm{X}_1,\ldots,\bm{X}_n)$} 
            \State $h = \lfloor\nicefrac{(n+d+2)}{2}\rfloor$
            \State $\alpha = \nicefrac{h}{n}$
            \State $N = \{1,\dots,n\}$
            \For{$k = 1$ to $500$}
                \State $H_k$ = sample($N$, size = $d+2$)
                \State $(\hat{\bm{\M}}_k, \hat{\bm{\Sigma}}^{\row}_k, \hat{\bm{\Sigma}}^{\col}_k, \bm{d}_k, H_k)$ = CSTEP$_2$($\bm{\mathfrak{X}}, H_k$) 
                \Comment{\textcolor{cyan}{2 MLE and C-step iterations}}
                \State $\delta_k = p\ln(\det(\hat{\bm{\Sigma}}^{\col}_k)) + q\ln(\det(\hat{\bm{\Sigma}}^{\row}_k))$
            \EndFor  
            \State $\pi_{\delta}(i) = \{\{1,\ldots,500\} \rightarrow \{1,\ldots,500\} : i \mapsto j: \delta_{\pi_{\delta}(1)} \leq \ldots \leq \delta_{\pi_{\delta}(500)}\}$ 
            \For{$l \in \{\pi_{\delta}(1), \pi_{\delta}(2), \dots, \pi_{\delta}(10)\}$}            
                \State $(\hat{\bm{\M}}_l, \hat{\bm{\Sigma}}^{\row}_l, \hat{\bm{\Sigma}}^{\col}_l, \bm{d}_l, H_l)$ = CSTEP($\bm{\mathfrak{X}}, H_l$) 
                \Comment{\textcolor{cyan}{Iterating C-steps until convergence}}
                \State $\delta_l = p\ln(\det(\hat{\bm{\Sigma}}^{\col}_l)) + q\ln(\det(\hat{\bm{\Sigma}}^{\row}_l))$
            \EndFor 
            \State $j = \argmin_{k \in N}(\delta_k)$
            \State $(\hat{\bm{\M}}, \hat{\bm{\Sigma}}^{\row}, \hat{\bm{\Sigma}}^{\col}) = (\hat{\bm{\M}}_j, c(\alpha)\hat{\bm{\Sigma}}^{\row}_j, \hat{\bm{\Sigma}}^{\col}_j)$ 
            \Comment{\textcolor{cyan}{Consistency scaling for raw MMCD}}
            \State $\bm{d} = (\mmd^2(\bm{X}_1; \hat{\bm{\M}}, \hat{\bm{\Sigma}}^{\row}, \hat{\bm{\Sigma}}^{\col}),\ldots,\mmd^2(\bm{X}_n; \hat{\bm{\M}}, \hat{\bm{\Sigma}}^{\row}, \hat{\bm{\Sigma}}^{\col}))$
            \State $H = H_j \cup \{i \in N | d_i < \chi^2_{0.975; pq}\}$
            \State $(\hat{\bm{\M}}, \hat{\bm{\Sigma}}^{\row}, \hat{\bm{\Sigma}}^{\col}) = \mle(\bm{X}_{i \in H})$  
            \Comment{\textcolor{cyan}{Computation of reweighted MMCD}}
            \State $\tilde{\alpha} = \nicefrac{\abs{H}}{n}$
            \State $(\hat{\bm{\M}}_{\ast}, \hat{\bm{\Sigma}}^{\row}_{\ast}, \hat{\bm{\Sigma}}^{\col}_{\ast}) = (\hat{\bm{\M}}, c(\tilde{\alpha})\hat{\bm{\Sigma}}^{\row}, \hat{\bm{\Sigma}}^{\col})$ 
            \Comment{\textcolor{cyan}{Consistency scaling for reweighted MMCD}}
            \State \textbf{return} $\hat{\bm{\M}}_{\ast}, \hat{\bm{\Sigma}}^{\row}_{\ast}, \hat{\bm{\Sigma}}^{\col}_{\ast}$
        \EndProcedure
    \end{algorithmic}
\end{algorithm}

\newpage

\subsection{Elemental subsets}
\label{supplement:mmcd_subsets}
For large $n$, the probability of obtaining at least one clean subset with $d+2$ observations among $m$ random subsets tends to
\begin{equation*}
    1 - (1 - (1-\varepsilon)^{d +2})^m,
\end{equation*}
with $\varepsilon$ denoting the percentage of outliers, see also \cite{Rousseeuw1999}. Hence, the number of subsets we must investigate to obtain at least one clean subset with a probability of $\beta$ is 
\begin{equation} \label{eq:subset_prob}
    \lceil \log(1-\beta)/log(1-(1-\varepsilon)^{d +2}) \rceil.
\end{equation}
In Figure~\ref{fig:cstep_obervation_contamination}, we plot the number of necessary subsets according to Equation~\eqref{eq:subset_prob} for $\beta = 0.99$ for $d$ between $1$ and $50$ and $\varepsilon$ between $0$ and $0.5$. The different green-shaded areas starting from the bottom right indicate settings where up to $m = 500$ initial subsets of size $d+2$ are sufficient to obtain at least one clean subset with a probability of $\beta = 0.99$ and the various shades of orange indicate settings where we need more elemental subsets. 

\begin{figure}[!ht]
    \centering
    \includegraphics[width = 0.67\linewidth]{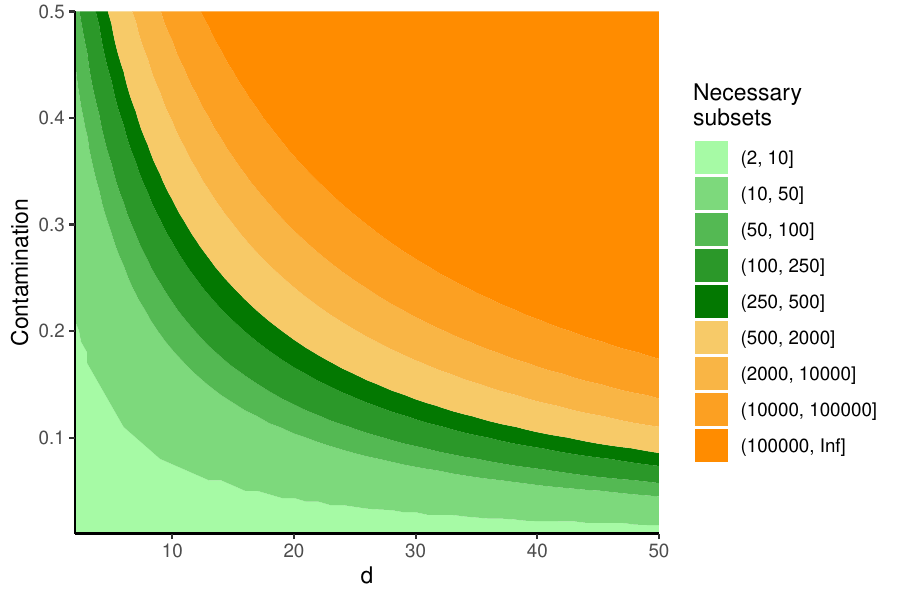}
    \caption{Number of subsets of size $d + 2$ we have to investigate for various levels of contamination, to obtain at least one clean subset with a probability of 99\%}
    \label{fig:cstep_obervation_contamination}
 \end{figure}

We assess the influence of using only 2 C-step and MLE iterations on the MMCD estimators' objective, the determinant of $\hat{\bm{\Sigma}}^{\col} \otimes \hat{\bm{\Sigma}}^{\row}$. We consider a setting with $n = 200$ observations with  $p = 2$ rows and $q = 8$ columns. The clean observations are generated by a centered matrix normal distribution with $\bm{\Sigma}^{\row} = \bm{\Sigma}^{\fix}(0.7)$ and $\bm{\Sigma}^{\col} = \bm{\Sigma}^{\mix}(0.7)$, with diagonal entries $\sigma_{jj}^{\fix} = \sigma_{jj}^{\mix} = 1$ and off-diagonal entries $\sigma_{jk}^{\fix}(0.7) = 0.7$ and $\sigma_{jk}^{\mix}(0.7) = 0.7^{\abs{j-k}}$, respectively. The outliers have a mean of $5$ and the same covariance as the regular observations. We use 100 random subsets and plot $\det(\hat{\bm{\Sigma}}^{\col} \otimes \hat{\bm{\Sigma}}^{\row})$ for subsequent C-step iterations with 40\% of contamination. We compare the setting when we limit the number of MLE iterations to 2 or iterate until convergence and/or use elemental subsets with $d + 2 = 6$ instead of $h$-subsets of size $\nicefrac{n}{2} = 100$. 
\begin{figure}[!ht]
    \centering
    \includegraphics[width = \linewidth]{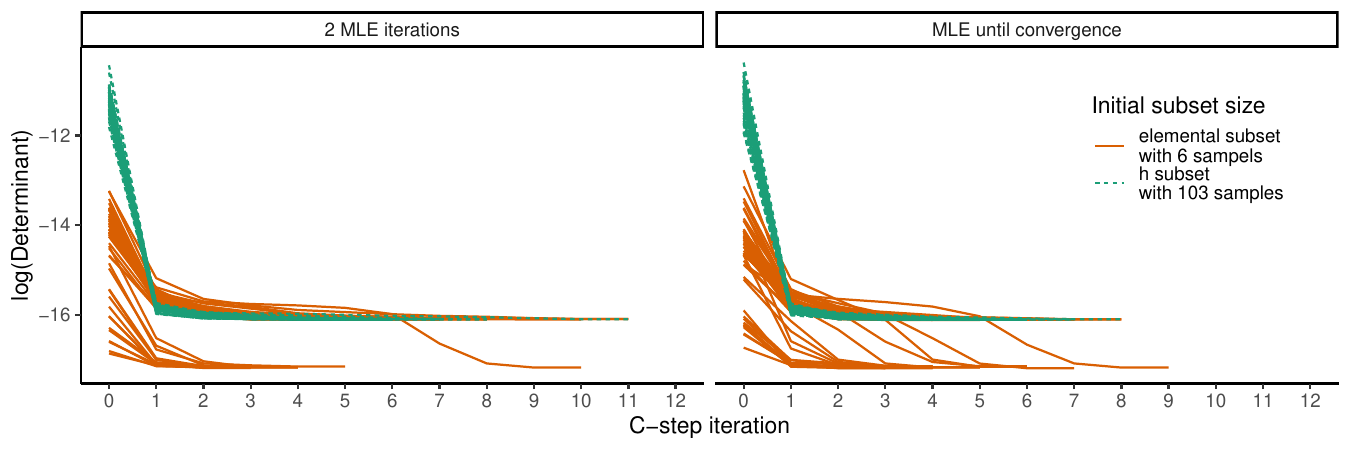}
    \caption{Logarithm of determinant for successive C-step iterations to analyze the effects of initial subset size and the number of ML iterations.}
    \label{fig:cstep_overview}
\end{figure}
Comparing the top and bottom row of Figure~\ref{fig:cstep_overview}, we see that there is virtually no difference in the objective function while limiting the ML iterations increases the computation speed. For the subset size, we see that several of the elemental subsets yield robust solutions with a lower covariance determinant than the larger $h$-subsets and that most of them are identified after 1 or 2 iterations. While 40\% contamination is not often encountered in practice, it shows that the algorithm can deal with settings with such a high level of contamination. We also analyzed settings with lower contamination, and using elemental subsets and fewer ML iterations had no negative effects in those settings, however, the larger $h$-subsets also led to robust solutions more frequently.

\begin{remark}
    Instead of using the consistency factor $c(\alpha)$ given in Equation~\eqref{eq:consistency_factor}, we could also scale the estimators to align the MMDs with a quantile of the chi-square distribution as in \cite{Rousseeuw1999}. Across the simulations and the examples considered in this paper, we have only seen very slight changes in the resulting estimators for both the raw and reweighted MMCD.
\end{remark}

\newpage
\section{Shapley proofs}
\label{supplement:shapley_proofs}

\begin{proof}[\hypertarget{proof:shapley_properties}{\textbf{Proof of Proposition~\ref{proposition:shapley_properties}}}]
    To show that cellwise Shapley values are not matrix affine equivariant, we consider a rowwise addition matrix $\bm{A}$ that adds the $w$-th row to the $v$-th row. For simplicity, let $\bm{B}$ be the identity matrix. Then Equation~\eqref{eq:shapley_lhs} yields
    \begin{align*}
        ((\bm{A}\bm{X}) \circ (\bm{C}\bm{Y}))_{jk}
        =& \begin{cases}
            (x_{jk} + x_{wk})(y_{jk} - y_{wk}) & {j = v}\\
            x_{jk}y_{jk} & {j \neq v}
        \end{cases}
    \end{align*}
    while 
    \begin{align*}
        (\bm{A}(\bm{X} \circ \bm{Y}))_{jk}
        &= \begin{cases}
                x_{jk}y_{jk} + x_{wk}y_{wk} & j = v\\
                x_{jk}y_{jk} & j \neq v
            \end{cases}.
    \end{align*} 
    Hence, we do not get invariance nor equivariance for rowwise or columnwise addition matrices. This also implies that the cellwise Shapley values are not, in general, matrix affine equivariant.
    
    Shift invariance follows from
    \begin{align*}
        \bm{\Phi}(\bm{X} + \bm{C}) &= ((\bm{X} + \bm{C})-(\bm{\M} + \bm{C})) \circ \bm{\Omega}^{\row} ((\bm{X} + \bm{C})-(\bm{\M} + \bm{C})) \bm{\Omega}^{\col} = \bm{\Phi}(\bm{X}),
    \end{align*}
    which means that we can assume that $\bm{X}$ has zero mean without loss of generality. 

    Let $\bm{Y} := \bm{\Omega}^{\row}\bm{X}\bm{\Omega}^{\col}$, $\bm{C} := (\bm{A}')^{-1}$ and $\bm{D} := (\bm{B}')^{-1}$, then we can write the cellwise Shapley values as $\bm{\Phi}(\bm{A}\bm{X}\bm{B}) = (\bm{A}\bm{X}\bm{B}) \circ (\bm{C}\bm{Y}\bm{D})$. The $jk$-th entry of this matrix can be written as
    \begin{align}
    \label{eq:shapley_lhs}
        \begin{split}     
            \phi_{jk}(\bm{A}\bm{X}\bm{B}) 
            =& ((\bm{A}\bm{X}\bm{B}) \circ (\bm{C}\bm{Y}\bm{D}))_{jk}
            = (\bm{A}\bm{X}\bm{B})_{jk} (\bm{C}\bm{Y}\bm{D})_{jk}\\
            =& \sum_{i = 1}^p \sum_{l = 1}^q a_{ji} x_{il} b_{lk} \sum_{m = 1}^p \sum_{n = 1}^q c_{jm} y_{mn} d_{nk}\\
            =& \sum_{i = 1}^p \sum_{l = 1}^q \sum_{m = 1, m \neq i}^p \sum_{n = 1, n \neq l}^q a_{ji}c_{jm} x_{il} y_{mn} b_{lk}d_{nk} \\
            &+ \sum_{i = 1}^p \sum_{l = 1}^q \sum_{n = 1, n \neq l}^q a_{ji}c_{ji} x_{il} y_{in} b_{lk}d_{nk} \\
            &+ \sum_{i = 1}^p \sum_{l = 1}^q \sum_{m = 1, m \neq i}^p a_{ji}c_{jm} x_{il} y_{ml} b_{lk}d_{lk} \\
            &+ \sum_{i = 1}^p \sum_{l = 1}^q a_{ji}c_{ji} x_{il} y_{il} b_{lk}d_{lk}.
        \end{split}
    \end{align}     
    
    If $\bm{A}$ is a scaling matrix, i.e., a diagonal matrix with non-zero entries, we have that 
    $$
        a_{ji}c_{jm} 
        = \begin{cases}
                1 & j = i = m\\
                0 & \text{otherwise}
        \end{cases},
    $$ 
    and similarly for $\bm{B}$. This implies that 
    $$
        \phi_{jk}(\bm{A}\bm{X}\bm{B}) = x_{jk} y_{jk} = (\bm{X} \circ \bm{Y})_{jk} = \phi_{jk}(\bm{X}) ,
    $$ showing the scale invariance.
    
    If $\bm{A}$ is a permutation matrix, i.e., a matrix consisting of any permutation of the canonical basis vectors, we have that $(\bm{A}')^{-1} = \bm{A}$ and 
    $$
        a_{ji}c_{jm} = a_{ji}a_{jm} 
        = \begin{cases}
            a_{ji} & i = m\\
            0 & i \neq m
        \end{cases},
    $$ 
    and similarly for $\bm{B}$. Hence Equation~\eqref{eq:shapley_lhs} becomes
    $$
        ((\bm{A}\bm{X}\bm{B}) \circ (\bm{C}\bm{Y}\bm{D}))_{jk} 
        = \sum_{i = 1}^p \sum_{l = 1}^q a_{ji}c_{ji} x_{il} y_{il} b_{lk}d_{lk} =
        (\bm{A}(\bm{X} \circ \bm{Y})\bm{B})_{jk},
    $$
    verifying the permutation equivariance. 
\end{proof}             

\begin{proof}[\hypertarget{proof:matrix_shapley}{\textbf{Proof of Theorem~\ref{proposition:matrix_shapley}}}]
To show that the computation of the rowwise Shapley value can be simplified, we start by rewriting the rowwise marginal contributions to the matrix Mahalanobis distance.

    \begin{align*}
        \Delta_a\mmd(\hat{\bm{X}}^{S}) :=& \mmd(\hat{\bm{X}}^{S \cup \{a\}}) - \mmd(\hat{\bm{X}}^{S})\\
        =& \sum_{i = 1}^p\sum_{j = 1}^p\sum_{k = 1}^q\sum_{l = 1}^q (\hat{x}^{S \cup \{a\}}_{ik} - m_{ik}) (\hat{x}^{S \cup \{a\}}_{jl} - m_{jl}) \omega_{lk}^{\col} \omega_{ij}^{\row} \\
        &- \sum_{i = 1}^p\sum_{j = 1}^p\sum_{k = 1}^q\sum_{l = 1}^q (\hat{x}^{S}_{ik} - m_{ik}) (\hat{x}^{S}_{jl} - m_{jl}) \omega_{lk}^{\col} \omega_{ij}^{\row}\\
        =& \sum_{i \in S \cup \{a\}}\sum_{j \in S \cup \{a\}}\sum_{k = 1}^q\sum_{l = 1}^q (x_{ik} - m_{ik}) (x_{jl} - m_{jl}) \omega_{lk}^{\col} \omega_{ij}^{\row} \\
        &- \sum_{i  \in S}\sum_{j \in S}\sum_{k = 1}^q\sum_{l = 1}^q (x_{ik} - m_{ik}) (x_{jl} - m_{jl}) \omega_{lk}^{\col} \omega_{ij}^{\row}\\
        =& \sum_{i \in S \cup \{a\}}\sum_{j \in S}\sum_{k = 1}^q\sum_{l = 1}^q (x_{ik} - m_{ik}) (x_{jl} - m_{jl}) \omega_{lk}^{\col} \omega_{ij}^{\row} \\
        &+ \sum_{i \in S \cup \{a\}}\sum_{k = 1}^q\sum_{l = 1}^q (x_{ik} - m_{ik}) (x_{al}-m_{al}) \omega_{lk}^{\col} \omega_{ia}^{\row} \\
        &- \sum_{i  \in S}\sum_{j \in S}\sum_{k = 1}^q\sum_{l = 1}^q (x_{ik} - m_{ik}) (x_{jl} - m_{jl}) \omega_{lk}^{\col} \omega_{ij}^{\row}\\
        =& \sum_{i \in S}\sum_{j \in S}\sum_{k = 1}^q\sum_{l = 1}^q (x_{ik} - m_{ik}) (x_{jl} - m_{jl}) \omega_{lk}^{\col} \omega_{ij}^{\row} \\
        &- \sum_{i  \in S}\sum_{j \in S}\sum_{k = 1}^q\sum_{l = 1}^q (x_{ik} - m_{ik}) (x_{jl} - m_{jl}) \omega_{lk}^{\col} \omega_{ij}^{\row}\\
        &+ \sum_{j \in S}\sum_{k = 1}^q\sum_{l = 1}^q (x_{ak} - m_{ak}) (x_{jl} - m_{jl}) \omega_{lk}^{\col} \omega_{aj}^{\row} \\
        &+ \sum_{i \in S }\sum_{k = 1}^q\sum_{l = 1}^q (x_{ik} - m_{ik}) (x_{al} - m_{al}) \omega_{lk}^{\col} \omega_{ia}^{\row} \\
        &+ \sum_{k = 1}^q\sum_{l = 1}^q (x_{ak} - m_{ak}) (x_{al} - m_{al}) \omega_{lk}^{\col} \omega_{aa}^{\row}\\
        =& 2\sum_{i \in S}\sum_{k = 1}^q\sum_{l = 1}^q (x_{al} - m_{al}) (x_{ik} - m_{ik}) \omega_{lk}^{\col} \omega_{ia}^{\row} \\
        &+ \sum_{k = 1}^q\sum_{l = 1}^q (x_{ak} - m_{ak}) (x_{al} - m_{al}) \omega_{lk}^{\col} \omega_{aa}^{\row}.
    \end{align*}
    Now the coordinates $\phi_a(\bm{X})$ of the Shapley value $\bm{\phi}(\bm{X})$ are given by $(w(\abs{S}) = \frac{\abs{S}!(p-\abs{S}-1)!}{p!})$
    \begin{align*}
        \phi_a(\bm{X}) =& \sum_{S \subseteq P \setminus \{a\}} w(\abs{S})  \Delta_a\mmd(\hat{\bm{X}}^{S})\\
        =& 2\sum_{S \subseteq P \setminus \{a\}} w(\abs{S})  \sum_{i \in S}\sum_{k = 1}^q\sum_{l = 1}^q (x_{al}-m_{al}) (x_{ik}-m_{ik}) \omega_{lk}^{\col} \omega_{ia}^{\row}\\
        &+ \sum_{S \subseteq P \setminus \{a\}} w(\abs{S})\sum_{k = 1}^q\sum_{l = 1}^q (x_{ak} - m_{ak}) (x_{al} - m_{al}) \omega_{lk}^{\col} \omega_{aa}^{\row}
    \end{align*}
    and we can simplify the first term of the sum as
    \begin{align*}
        & 2\sum_{S \subseteq P \setminus \{a\}} w(\abs{S})  \sum_{i \in S}\sum_{k = 1}^q\sum_{l = 1}^q (x_{al}-m_{al}) (x_{ik}-m_{ik}) \omega_{lk}^{\col} \omega_{ia}^{\row}\\
        &= 2 \sum_{s = 1}^{p-1} w(\abs{S})  \sum_{S \subseteq P \setminus \{a\}, \abs{S} = s} \sum_{i \in S}\sum_{k = 1}^q\sum_{l = 1}^q (x_{al}-m_{al}) (x_{ik}-m_{ik}) \omega_{lk}^{\col} \omega_{ia}^{\row}\\
        &= 2 \sum_{s = 1}^{p-1} \sum_{k = 1}^q\sum_{l = 1}^tw(\abs{S})  \sum_{S \subseteq P \setminus \{a\}, \abs{S} = s} \sum_{i \in S} (x_{al}-m_{al}) (x_{ik}-m_{ik}) \omega_{lk}^{\col} \omega_{ia}^{\row}\\
        &= 2 \sum_{s = 1}^{p-1} \sum_{k = 1}^q\sum_{l = 1}^q\frac{\abs{S}!(p-\abs{S}-1)!}{p!} \binom{p-2}{s-1}\sum_{i \in P \setminus \{a\}} (x_{al}-m_{al}) (x_{ik}-m_{ik}) \omega_{lk}^{\col} \omega_{ia}^{\row}\\
        &= 2 \frac{1}{p(p-1)}\sum_{s = 1}^{p-1}  s \sum_{k = 1}^q\sum_{l = 1}^q\sum_{i \in P \setminus \{a\}} (x_{al}-m_{al}) (x_{ik}-m_{ik}) \omega_{lk}^{\col} \omega_{ia}^{\row}\\
        &= 2   \frac{1}{p(p-1)}\frac{p(p-1)}{2} \sum_{k = 1}^q\sum_{l = 1}^q\sum_{i \in P \setminus \{a\}} (x_{al}-m_{al}) (x_{ik}-m_{ik}) \omega_{lk}^{\col} \omega_{ia}^{\row} \\
        &= \sum_{k = 1}^q\sum_{l = 1}^q\sum_{i \in P \setminus \{a\}} (x_{al}-m_{al}) (x_{ik}-m_{ik}) \omega_{lk}^{\col} \omega_{ia}^{\row}.
    \end{align*}
    Since the second term is independent of the subset $S$ and $\sum_{S \subseteq P \setminus \{a\}} w(\abs{S}) = 1$, we obtain 
    \begin{align*}
        \phi_a(\bm{X}) =& \sum_{k = 1}^q\sum_{l = 1}^q\sum_{i \in P \setminus \{a\}} (x_{al}-m_{al}) (x_{ik}-m_{ik}) \omega_{lk}^{\col} \omega_{ia}^{\row} \\
        &+ \sum_{k = 1}^q\sum_{l = 1}^q (x_{ak} - m_{ak}) (x_{al} - m_{al}) \omega_{lk}^{\col} \omega_{aa}^{\row} \\
        =& \sum_{i = 1}^p\sum_{k = 1}^q\sum_{l = 1}^q (x_{al}-m_{al}) (x_{ik}-m_{ik}) \omega_{lk}^{\col} \omega_{ia}^{\row},
    \end{align*}
    which completes the proof.
\end{proof}

\newpage
\section{Further simulation results}
\label{supplement:simulations}
    In order to select a simulation setting, one has to consider that the ML estimators for the parameters of the matrix-variate normal distribution employ an iterative algorithm, which is commonly initialized by setting either the rowwise or columnwise covariance matrix equal to the identity matrix~\citep{Dutilleul1999}. Therefore, identity covariance matrices will not be used for data generation as this could lead to an undesirable advantage for the estimation.

    To assess the quality of covariance estimation, we consider two additional measures to the KL divergence: the relative Frobenius error given as
    $$\frac{\norm{\hat{\bm{\Sigma}}^{\col} \otimes \hat{\bm{\Sigma}}^{\row}-\bm{\Sigma}^{\col} \otimes \bm{\Sigma}^{\row}}_F}{\norm{\bm{\Sigma}^{\col} \otimes \bm{\Sigma}^{\row}}_F},$$
    and angle error between eigenvalues given as 
    $$1-\frac{\hat{\bm{a}}^\top\bm{a}}{\sqrt{\hat{\bm{a}}^\top\hat{\bm{a}}}\sqrt{\bm{a}^\top\bm{a}}},$$
    where $\hat{\bm{a}}$ and ${\bm{a}}$ are the vectors of sorted eigenvalues of $\hat{\bm{\Sigma}}^{\col} \otimes \hat{\bm{\Sigma}}^{\row}$ and $\bm{\Sigma}^{\col} \otimes \bm{\Sigma}^{\row}$, respectively. Large values of the KL divergence and the relative Frobenius error indicate difficulties in the estimation of the covariances. 
    The angle error between the eigenvalues is in the interval $[0,1]$, and a large value means that the shape of the covariance matrix is not appropriately estimated.
    To assess the efficacy of outlier detection, we include the F-score in addition to precision and recall. The F-score is defined as the harmonic mean of precision and recall, where precision denotes the proportion of correctly identified outliers among all detected samples, while recall represents the proportion of correctly identified outliers among all contaminated samples.
    The \texttt{R} code of the simulations and all simulation results are available in the online supplement. 

    \subsection{Effects of dimensionality and computation time}  

    We start by considering additional metrics for the simulations discussed in Section~\ref{section:simulations}. Figure~\ref{fig:simulation_pq_ratio_line2} shows the F-score in addition to precision and recall. The F-score shows that for $n = 100$ and increasing dimensionality the robust MMCD estimators and the MLEs yield similar results. This is due to an increasing recall of the MLEs and a decreasing precision of the MMCD estimators. For $n = 1000$, the F-score of the MMCD estimators is close to the benchmark and for the MCD we see the advantage of using the deterministic MCD approach over the Fast-MCD method with increasing sample size. In Figure~\ref{fig:simulation_pq_ratio_line3} we see that the MCD performs best in all settings across all evaluation measures. For the MCD we do not see a difference in the KL divergence when swapping to the deterministic procedure. However, the angle error between eigenvalues shows clear improvements, indicating that the estimation of the shape of the covariance matrix improves. Both in terms of the angle and Frobenius error, the MCD estimator attains better scores than the MLEs even for higher $pq$, while the MLEs have better KL divergence.  

    We also analyze the computation times of the estimators in this setting. Figure~\ref{fig:simulation_pq_ratio_time} clearly shows that computation time depends on the dimensionality of the matrix-variate samples and the number of samples for all approaches. The relative increases of computation time of the matrix MLEs and the MMCD estimators are similar for $n \in \{20,100,300\}$ but for $n = 1000$ the relative increase in computation time for the matrix MLEs is larger than for the MMCD estimators, highlighting the effectiveness of the subsampling approach with increasing sample size. For the MCD, we observe a decrease in computation time when $pq > 300$ since the deterministic MCD is used instead of the Fast-MCD procedure. However, computing the MCD still takes longer than the MMCD approach. Hence, the matrix-variate approach does yield higher robustness and more accurate covariance estimation with shorter computation times. Although parallel processing is available for the MMCD procedure, it was not utilized in the simulations to ensure better comparability for the algorithms. Depending on the number of available threads, parallel processing yields substantial improvements in computation time.

    \subsection{Cellwise and block contamination}

    We also consider the additional metrics for the simulations comparing the three different contamination types in Figures~\ref{fig:simulation_contamination_type2} and \ref{fig:simulation_contamination_type3}. The robustness of the MMCD estimators is again confirmed using all three metrics assessing the quality of the covariance estimation. The angle error reveals that the cell contamination has less effect on the shape of the covariance matrix than the other two scenarios and that all three estimators seemingly do a good job of estimating the covariance shape. For block contamination, the MCD yields better results than the MLEs with increasing sample size and even gets close to the MMCD in terms of angle error. 

    \begin{remark}
        Our cell contamination setting does not correspond to the setting of cellwise outliers \citep{alqallaf2009propagation}. We first select a subset of outlying observations and permute the cells for this selection while \cite{alqallaf2009propagation} select a fraction of all cells from all samples. In our setting, we can guarantee that only 10 percent of the samples are contaminated while the cellwise contamination scheme of \cite{alqallaf2009propagation} would likely lead to more than half of the samples being contaminated. 
    \end{remark}
        
    In further simulations, we considered different fractions of contaminated samples as well as multiple rowwise and columnwise covariance matrices for cellwise and block contamination. Additionally, we analyzed the effect of the fraction of permuted cells per observation for cell contamination, and for block contamination, we considered different mean matrices. Those simulation results are not discussed here but are available in the online supplement.    
    
    \subsection{Shift outliers}
    \label{subsection:simulation_shift_outliers}

    For shift outliers, we include an in-depth analysis of the effect of the various simulation parameters. 
    The simulations involve generating regular and outlying samples from a matrix normal distribution. A fraction, $\varepsilon$, of the clean data is replaced by outliers. The clean observations are drawn from a centered distribution, while the mean of the outliers shifts based on the parameter $\gamma$, i.e., the mean of the outliers is set to a matrix with all entries equal to $\gamma$.
    Three types of covariance matrices are considered:
    The covariance matrix $\bm{\Sigma}^{\rnd}$, as proposed by \cite{agostinelli2015robust}, is randomly generated with low correlations. The covariance matrix $\bm{\Sigma}^{\fix}(0.7)$ induces a relatively collinear setting, with entries defined as:
    \begin{align*}
            \sigma_{jk}^{\fix}(0.7) =  \begin{cases}
                1 & \text{ if } j = k \\
                0.7 & \text{ if } j \neq k
            \end{cases}.
        \end{align*}
    The covariance matrix $\bm{\Sigma}^{\mix}(0.7)$ exhibits both large and small correlations, featuring entries as follows:
     \begin{align*}
            \sigma_{jk}^{\mix}(0.7) =  \begin{cases}
                1 & \text{ if } j = k \\
                0.7^{\abs{j-k}} & \text{ if } j \neq k
            \end{cases}.
        \end{align*}   
    While maintaining the same covariance structure for both outliers and clean samples, we explore the impact of increasing the outlier covariance by scaling the covariance of clean observations by the parameter $s$. Each simulation setting is replicated 100 times.
    Unless specified otherwise, we set $\bm{\Sigma}^{\row} = \bm{\Sigma}^{\rnd}$, $\bm{\Sigma}^{\col} = \bm{\Sigma}^{\mix}(0.7)$, and $s = 1$ as detailed in Section~\ref{section:simulations}. An overview of all parameters for the simulations is provided in Table~\ref{tab:sim_settings}. For $(p,q) = (5,20)$, all listed parameter combinations are considered, while for $(p,q) \in {(50,20),(100,50)}$, we only consider $s = 1$. 
    
    \begin{table}[htp]
        \centering
        \begin{tabular}{ll}
            \toprule
            Parameter & Parameter values\\ 
            \midrule
            Sample size $n$ & $20, 50, 100, 200, 300, 400, 500, 750, 1000$\\
            Contamination $\varepsilon$ & $0.1,0.2,0.3,0.4$\\
            Rowwise covariance $\bm{\Sigma}^{\row}$ & $\bm{\Sigma}^{\fix}(0.7),\bm{\Sigma}^{\rnd}$\\
            Columnwise covariance $\bm{\Sigma}^{\col}$ & $\bm{\Sigma}^{\mix}(0.7),\bm{\Sigma}^{\rnd}$\\
            Mean shift $\gamma$ & $1,2,3,4,5$\\
            Covariance multiplier $s$ & $1,2,3,4$\\
            \bottomrule
        \end{tabular}
        \caption{Parameters considered for the simulations with $p,q = (5,20)$.}
        \label{tab:sim_settings}
    \end{table}

    We analyze the effect of the mean shift in a setting with contamination of $\varepsilon = 0.2$ and compare $\gamma = 1$ and $\gamma = 3$. In the upper row of Figure~\ref{fig:simulation_line_plot1.1}, the boxplots depict F-scores across various parameter configurations. Notably, the MMCD estimators exhibit improved performance as sample sizes increase across all settings, consistently outperforming ML estimators. However, for $(p,q) = (5,20)$, in a scenario involving a minor mean shift, the F-scores derived from MMCD exhibit some volatility with larger sample sizes. This situation arises due to the proximity of outliers to regular observations, posing challenges in their identification. Notably, a more pronounced mean shift significantly simplifies outlier detection. 
    Moreover, we see that the recall of the MMCD estimators is close to one across all settings, except for $(p,q) = (5,20)$ and a small mean shift. The MLE estimators only detect the most severe outliers due to the masking effect, leading to a median recall below $0.25$ across all settings. With an increasing sample size, the precision of the MMCD is improving and has very low variability. On the other hand, the MLE shows very unstable results. 

    Figure~\ref{fig:simulation_line_plot1.2} presents the scores depicting covariance estimation. 
    For the MMCD estimators the covariance estimation performance is improving with the sample size across all settings. On the other hand, the sample size has a negligible effect on the quality of the MLE estimators in the presence of outliers and a larger mean shift decreases performance. For small sample sizes, MLE and MMCD estimators are close in terms of KL divergence, but the angle error and Frobenius error indicate worse performance of the MLE estimators also for small sample sizes. The relative Frobenius error of MMCD estimators is smaller than one and thus only plotted on $[0,1]$. For the MLE estimators, it is often above one and those settings are not visible in plots.     

    Figure~\ref{fig:simulation_lineplot_eps_low_vs_high} shows the difference between a contamination of $\varepsilon = 0.1$ and $\varepsilon = 0.4$ with mean shift $\gamma = 1$. The KL divergence reveals that the MMCD estimator yields more accurate results across all settings. However, for $\varepsilon = 0.1$, the F-scores of the MLE are increasing with the dimensionality and perform better than the MMCD for small sample sizes. For $\varepsilon = 0.4$, only the MMCD yields reliable results.

    For the setting with $(p,q) = (5,20)$ and $\varepsilon = 0.2$, we also computed the MCD on the vectorized samples in addition to the matrix MLE and MMCD and considered the true mean and covariance used to generate the data as a benchmark. Figures~\ref{fig:simulation_lineplot_low_dim_cov} and \ref{fig:simulation_lineplot_low_dim_outlier} summarize the results and reveal that the MCD on the vectorized observations does not lead to robust estimators. This issue arises because the robustness of the MCD and MMCD depends on the dimensionality of the data. For the MCD it depends on $p \cdot q$ and for the MMCD it depends on $\nicefrac{p}{q} + \nicefrac{q}{p}$.    
    To achieve a $99\%$ probability of obtaining at least one clean initial subset with $(p,q) = (5,20)$ and a contamination $\varepsilon = 0.2$, MCD requires approximately $2.8\cdot10^{10}$ initial subsets, while MMCD only needs 16. 
    For the setting with the smallest mean shift, the comparison between MMCD and the actual parameters in Figure~\ref{fig:simulation_lineplot_low_dim_outlier} highlights the difficulty of this setting since even using the actual parameters; the recall shows a lot of variability. 

    In addition to shifting the mean of the outliers by $\gamma \in \{1,\dots,5\}$, we now consider the effect of scaling the covariance by $s \in \{1,\dots,4\}$. The difference between $s \in \{2,3,4\}$ was negligible and Figures~\ref{fig:simulation_lineplot_low_dim_radial_cov} and \ref{fig:simulation_lineplot_low_dim_radial_outlier} summarize the results for $s = 2$. While the MLE performs quite well for outlier detection, especially compared to the setting with $s = 1$ (see Figure~\ref{fig:simulation_lineplot_eps_low_vs_high}), the estimated covariance matrices are not accurate. The overall performance of the MCD computed on the vectorized samples improves with increasing sample size $n$ but even more samples would be necessary to obtain similar results to the MMCD. 

    Finally, we compare the 4 different combinations of row- and columnwise covariance matrices with $\varepsilon = 0.2$. In Figure~\ref{fig:simulation_lineplot_compare_cov_small_shift} we use $\gamma = 1$ and in Figure~\ref{fig:simulation_lineplot_compare_cov_large_shift} we increase the mean shift to $\gamma = 5$. The F-score based on the true parameters is included as a reference. When $\bm{\Sigma}^{\row} = \bm{\Sigma}^{\fix}(0.7)$, $\bm{\Sigma}^{\col} = \bm{\Sigma}^{\mix}(0.7)$, and $\gamma = 1$, the mean shift is too small and the outliers cannot be separated from the regular observations. Increasing the mean shift to $\gamma = 5$, the separation becomes clearer and the MMCD yields robust results. 
    If $\gamma = 1$, we still see a lot of variability in the F-score if only the rowwise or columnwise covariance matrix is generated randomly. However, if both are generated randomly the distinction between outliers and regular observations is easier.

    \subsubsection{Effects of fine-grained mean shifts}    
    To get a more in-depth view of the effect of the mean shift we consider a finer grid for the parameter $\gamma \in \{0.1,0.2,\dots,2\}$ for $n \in \{20,100,1000\}$, $(p,q) = (5,20)$, $\varepsilon = 0.1$. 
    In Figure~\ref{fig:simulation_low_gamma_line2}, we see that for $n = 20$, the MMCD has a low precision but an even higher recall than we can achieve using the actual parameters used to generate the data to compute the Mahalanobis distances for outlier detection. For larger sample sizes, the precision of the MMCD increases while the recall remains high, resulting in an F-score close to the one achieved by the actual parameters. For $n = 1000$, we also computed the MCD on the vectorized observations, it attains a higher recall than the matrix MLEs but lower precision and performs worse than the MMCD in all settings.
    Likewise, to the results for outlier detection, Figure~\ref{fig:simulation_low_gamma_line3} shows similar results for covariance estimation. While the MMCD performs best in most settings it shows potential for improvement for small $n$ and $\gamma$. The simulations also show that at a level of 10 percent contamination, even a small shift $\gamma$ of the outliers negatively impacts covariance estimation and, consequently, outlier detection due to the masking effect.

    \subsection{Beyond normality, contaminated t-distribution}
    
    To analyze the effect of deviations from the matrix normal distribution we consider samples from a matrix t-distribution. Similar to the matrix normal distribution, the matrix t-distribution is parameterized by a mean matrix, rowwise and columnwise covariance matrices, and degrees of freedom as an additional parameter, see \cite{gupta1999} for more details. We also consider the ML estimators for the matrix t-distribution proposed by \cite{thompson2020classification}, which are implemented in the \texttt{R} package \texttt{MixMatrix}. We consider samples from a $p \times q = 5 \times 20$ centered matrix t-distribution with $\nu \in \{1,\dots,30\}$ degrees of freedom with $\bm{\Sigma}^{\row} = \bm{\Sigma}^{\rnd} \in \pds(p)$ and $\bm{\Sigma}^{\col} = \bm{\Sigma}^{\mix}(0.7) \in \pds(q)$ for $n \in \{20,100,1000\}$, $(p,q) = (5,20)$, $\varepsilon \in \{0.1,0.2\}$. The outliers are generated from a shifted distribution with a mean matrix of all ones with the same covariance structure and the same degrees of freedom. In Figure~\ref{fig:simulation_t_distribuiton}, we analyze the influence of the degrees of freedom on precision, recall, angle error between eigenvalues, and the logarithm of the relative Frobenius error for various estimators, number of samples, and levels of contamination. 
    The angle and Frobenius error clearly show the advantage of the MMCD estimators for covariance estimation. 
    If the distribution of the samples is known, the consistency correction outlined in Theorem~\ref{theorem:consistency} allows us to obtain consistency for any matrix elliptical distribution. Since we do not know the underlying distribution in practice, we use the consistency factor for the normal model given in Equation~\eqref{eq:consistency_factor} which does affect the scale of the covariance but not the shape. This is also reflected in the difference between the angle and Frobenius error of the MMCD estimators and MLEs for the matrix t-distribution since the scale of the covariance has a more profound impact on the Frobenius error. In terms of angle error, the MMCD estimators perform better than the MLEs for the matrix t-distribution for all degrees of freedom $\nu$ while the MMCD shows high Frobenius errors for $\nu \leq 4$. 

    While the MMCD estimators and MLEs for the matrix t-distribution have a recall close to one in all settings, we see a difference in precision depending on the fraction of contaminated samples and the number of samples. For $\varepsilon = 0.1$, the MLEs for the matrix t-distribution show a steep increase in precision with rising degrees of freedom for all the sample sizes. On the other hand, for $\varepsilon = 0.2$, the precision is constant and low for all $n$. Similarly to the simulations based on the normal model, the precision of the MMCD estimators is low for $n = 20$ and remains low for increasing degrees of freedom. While the precision increases alongside the degrees of freedom for larger sample sizes it is still low. However, this is what we would expect since the matrix t-distribution has heavier tails than the matrix normal distribution and the mean shift is rather small, such that we do not see the full potential of the MMCD estimators even under the normal model, see Section~\ref{subsection:simulation_shift_outliers}. 

    Both the normal MLEs and the MCD estimators computed on the vectorized samples perform poorly for covariance estimation and outlier detection when the samples are generated from a matrix t-distribution.

    \begin{figure}
        \centering
        \includegraphics[width = 1\linewidth]{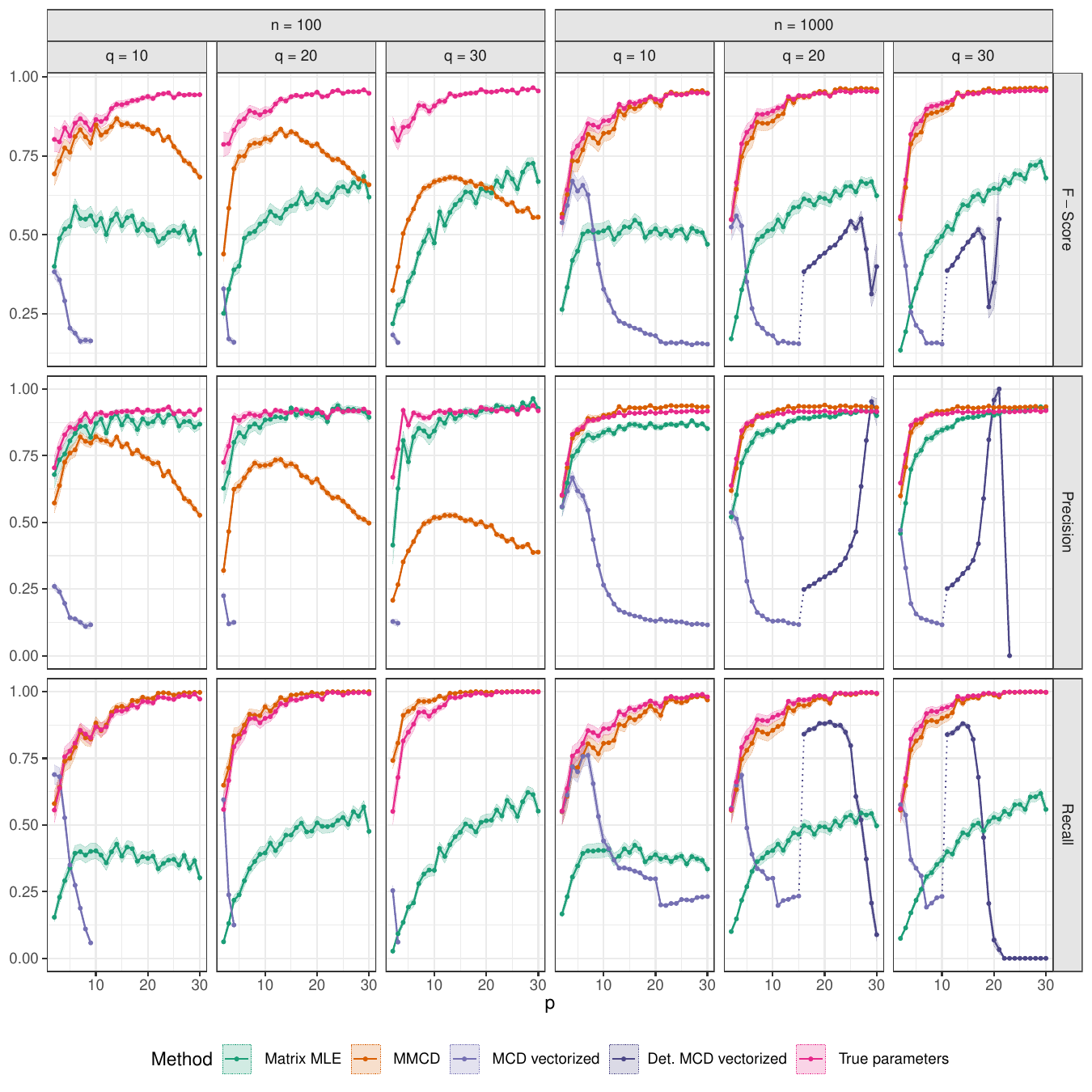}
        \caption{Outlier detection capabilities comparing multiple matrix sizes $p \in \{2,\dots,30\}$ and $q \in \{10,20,30\}$ for $n \in \{100,1000\}$, $\gamma = 1$, $\varepsilon = 0.1$.}
        \label{fig:simulation_pq_ratio_line2}
    \end{figure}
    
    \begin{figure}
        \centering
        \includegraphics[width = 1\linewidth]{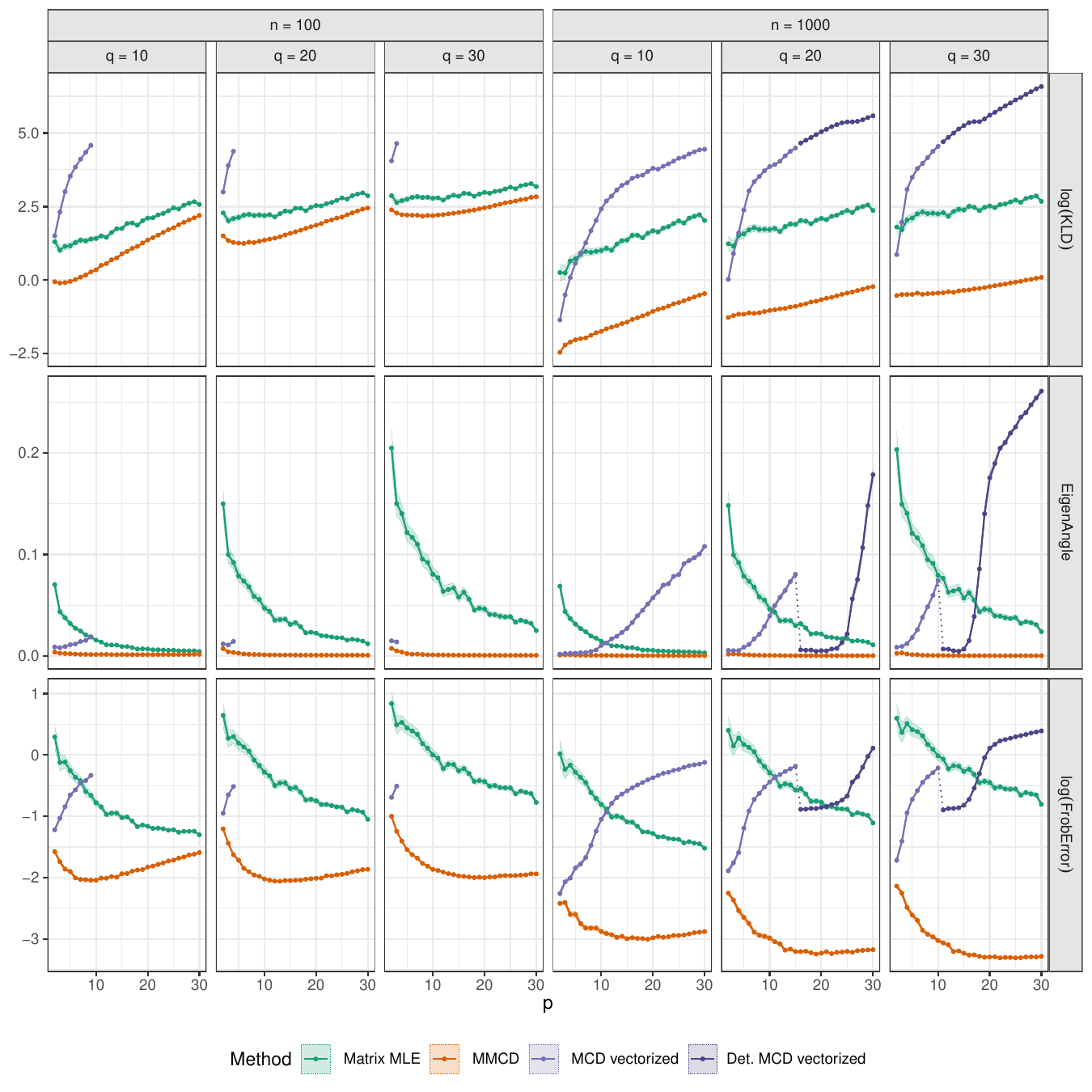}
        \caption{Quality of covariance estimation comparing multiple matrix sizes $p \in \{2,\dots,30\}$ and $q \in \{10,20,30\}$ for $n \in \{100,1000\}$, $\gamma = 1$, $\varepsilon = 0.1$.}
        \label{fig:simulation_pq_ratio_line3}
    \end{figure}

    \begin{figure}
        \centering
        \includegraphics[width = 1\linewidth]{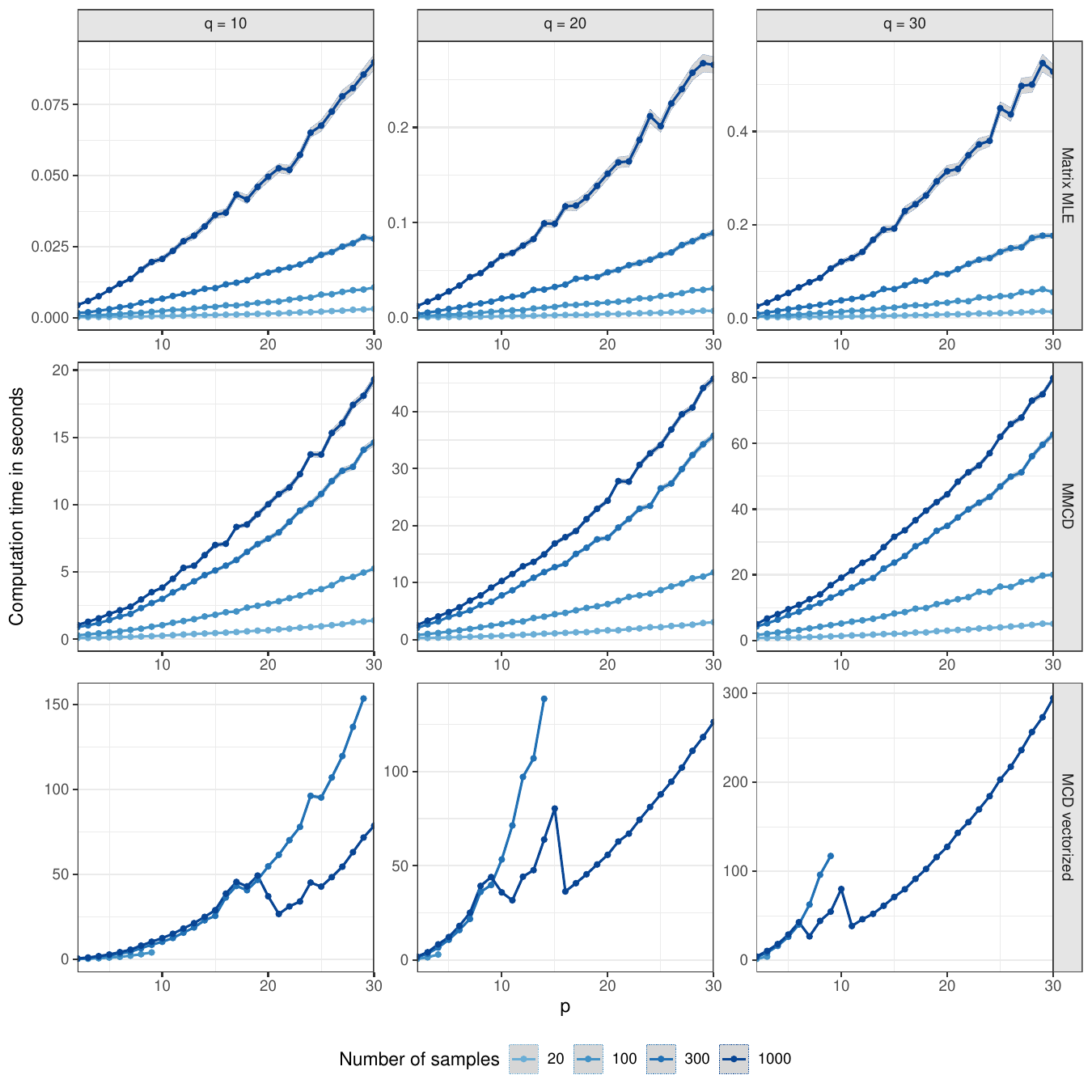}
        \caption{Comparison of computation time in seconds for multiple matrix sizes $p \in \{2,\dots,30\}$ and $q \in \{10,20,30\}$ for $n \in \{20,100,300,1000\}$.}
        \label{fig:simulation_pq_ratio_time}
    \end{figure}

    \begin{figure}
        \centering
        \includegraphics[width = 1\linewidth]{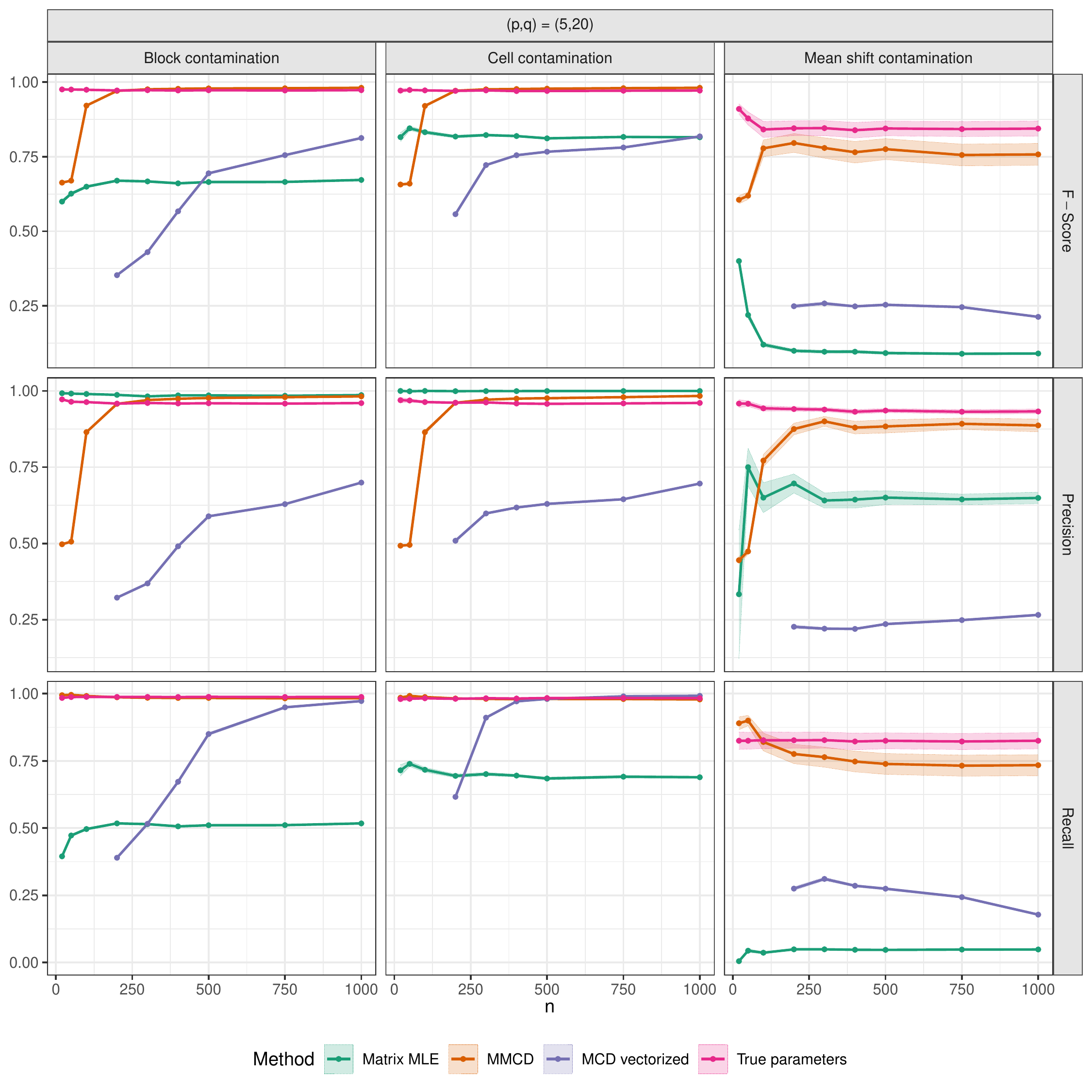}
        \caption{Quality of covariance estimation comparing block, cell, and sample contamination.}
        \label{fig:simulation_contamination_type2}
    \end{figure}

    \begin{figure}
        \centering
        \includegraphics[width = 1\linewidth]{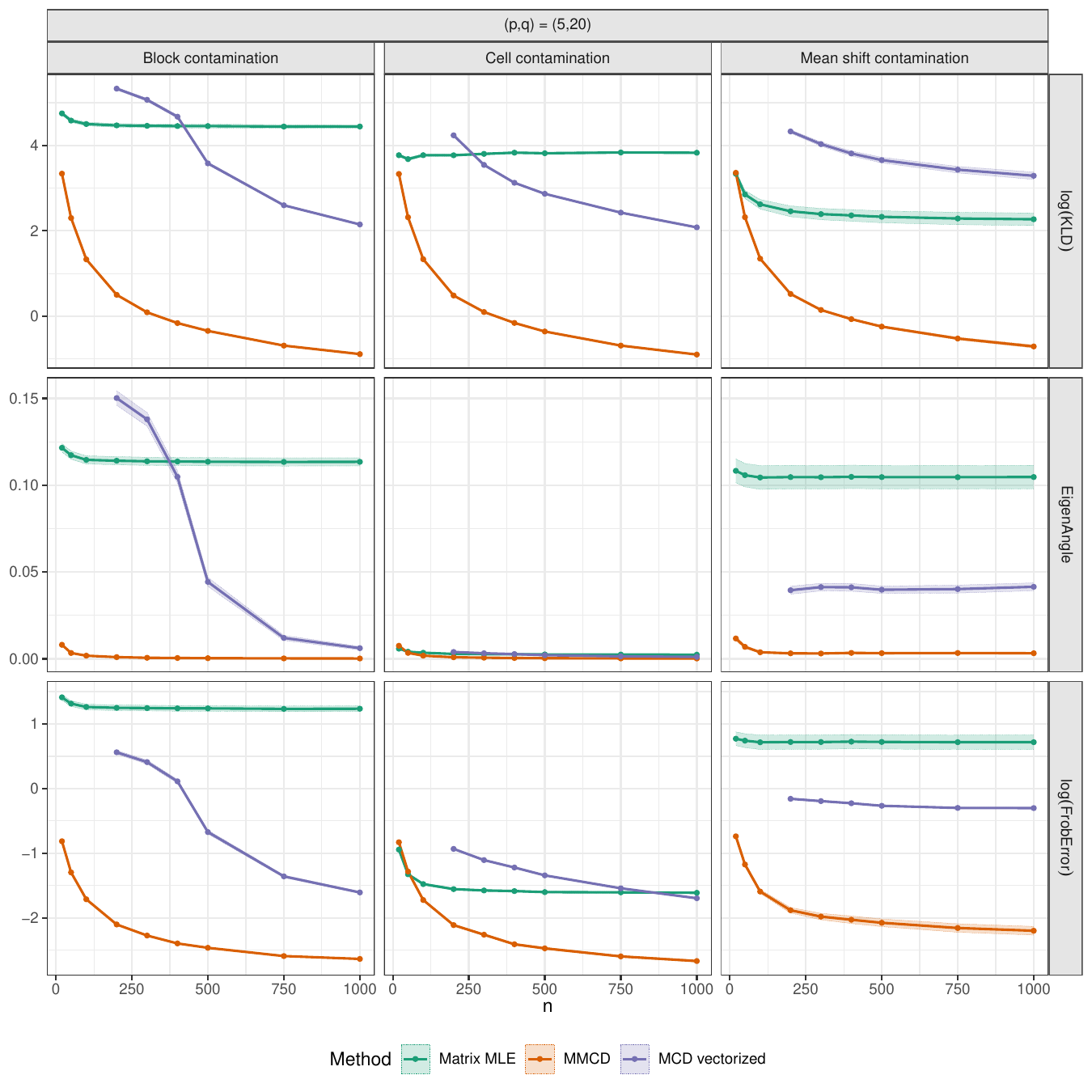}
        \caption{Outlier detection capabilities comparing block, cell, and sample contamination.}
        \label{fig:simulation_contamination_type3}
    \end{figure}
    
    \begin{figure}
        \centering
        \includegraphics[width = 1\linewidth]{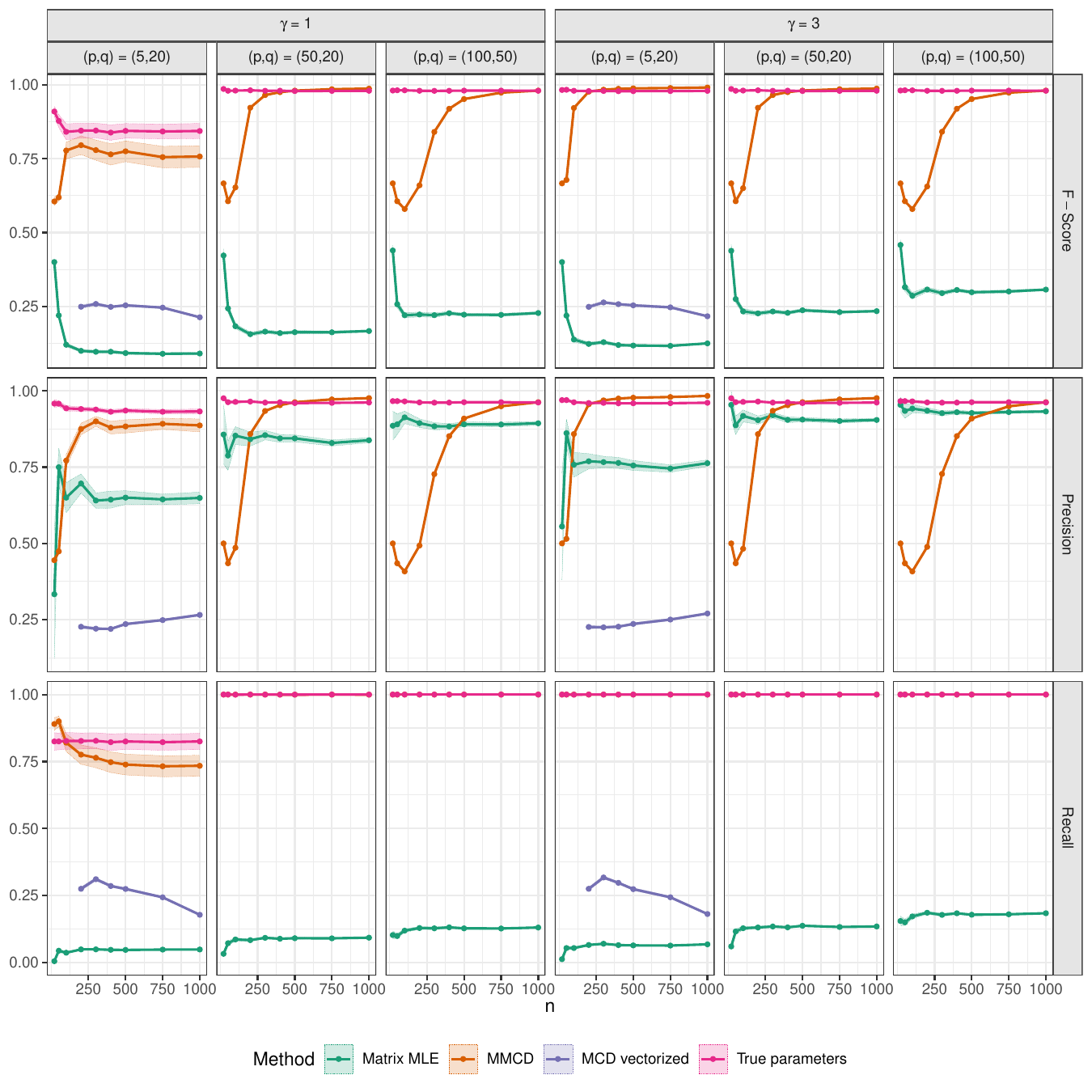}
        \caption{Overview of simulation results with a fraction $\varepsilon = 0.2$ of contaminated samples. The outlier detection capabilities are measured by F-score, precision, and recall.}
        \label{fig:simulation_line_plot1.1}
    \end{figure} 

    \begin{figure}
        \centering
        \includegraphics[width = 1\linewidth]{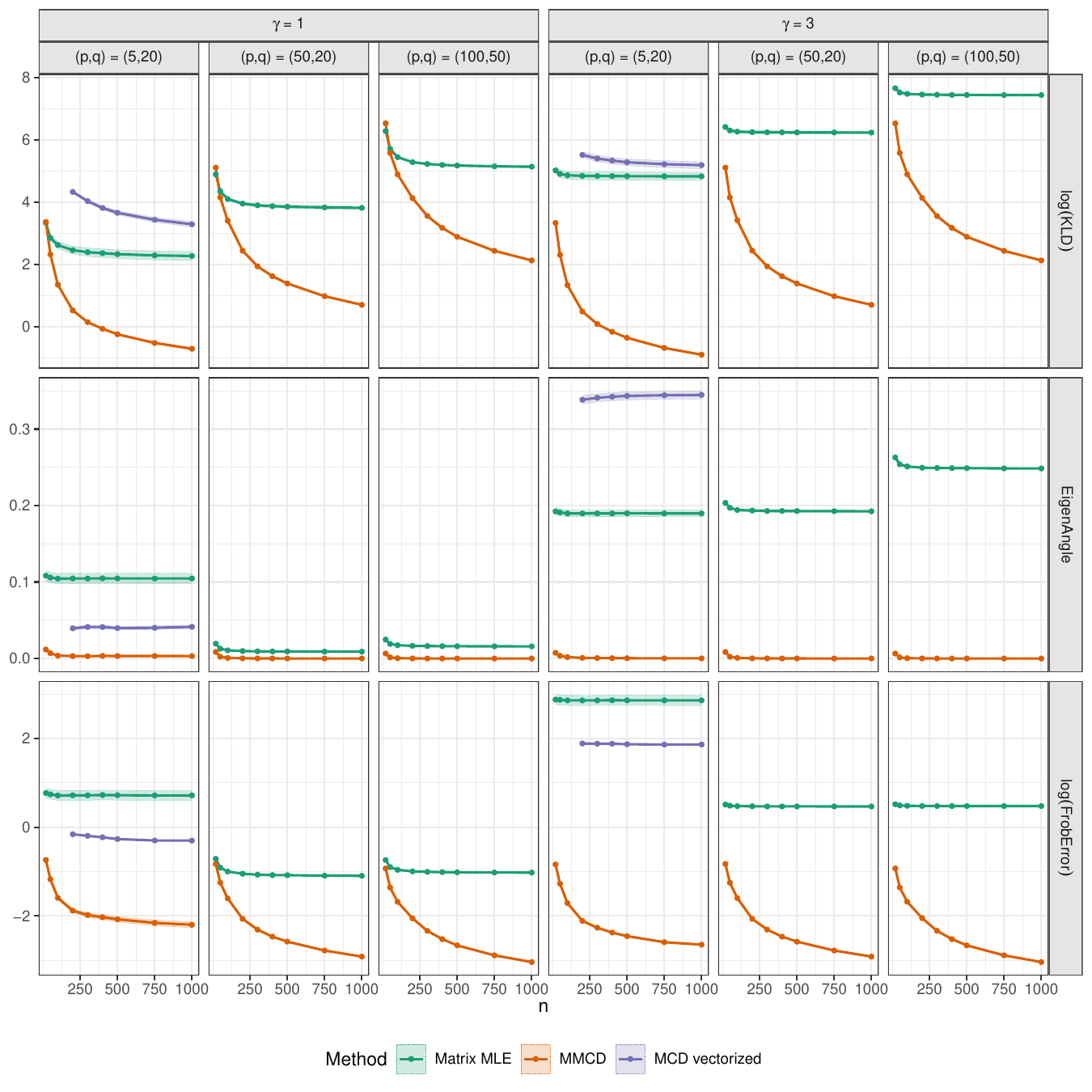}
        \caption{Overview of simulation results with a fraction $\varepsilon = 0.2$ of contaminated samples. The quality of covariance estimation is evaluated based on the logarithm of KL divergence, angle error between eigenvalues, and the logarithm of relative Frobenius error.}
        \label{fig:simulation_line_plot1.2}
    \end{figure}  
    
    \begin{figure}
        \centering
        \includegraphics[width = 1\linewidth]{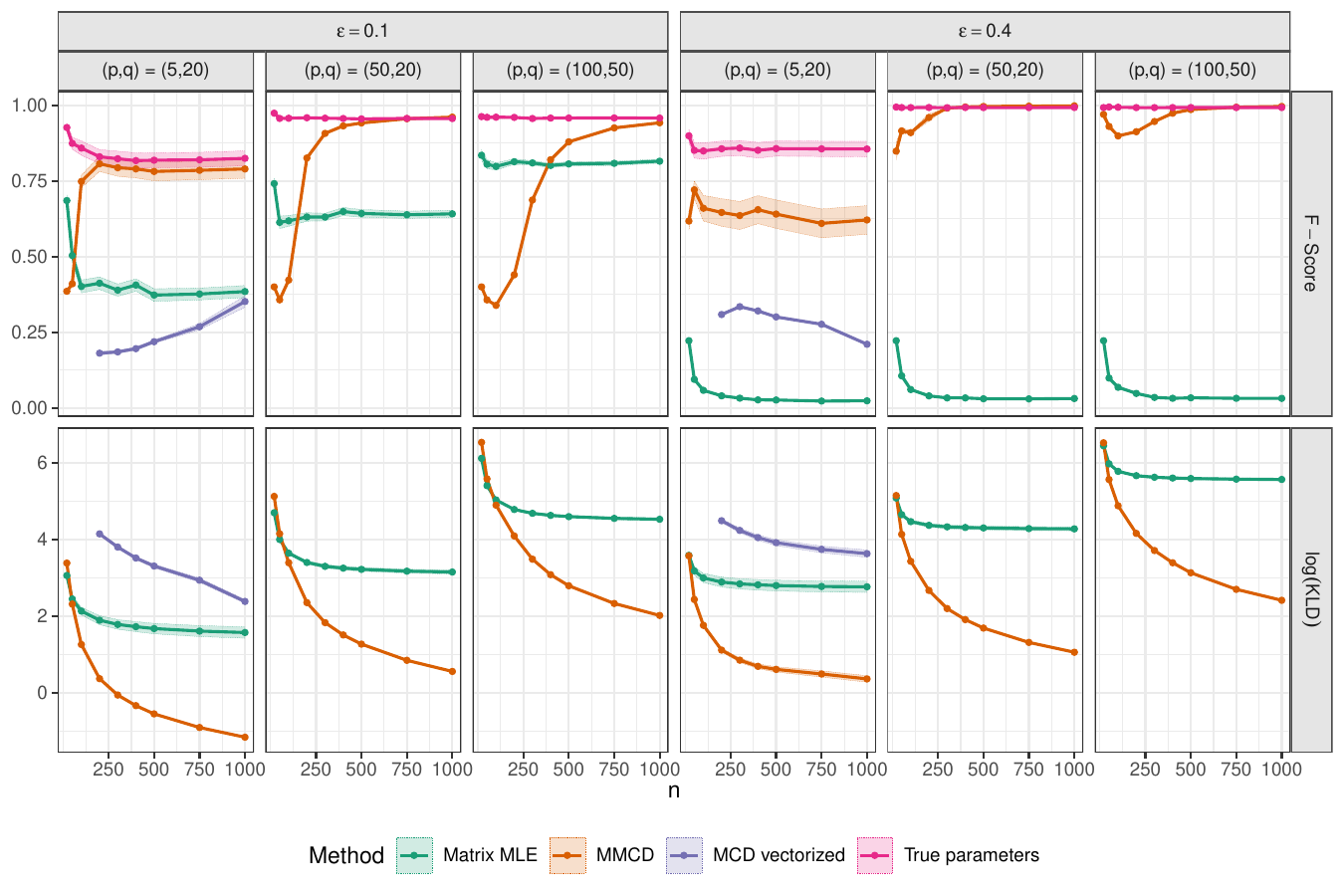}
        \caption{F-score and logarithm of KL divergence for simulations with mean shift $\gamma = 1$.}
        \label{fig:simulation_lineplot_eps_low_vs_high}
    \end{figure}   
    
    \begin{figure}
        \centering
        \includegraphics[width = 1\linewidth]{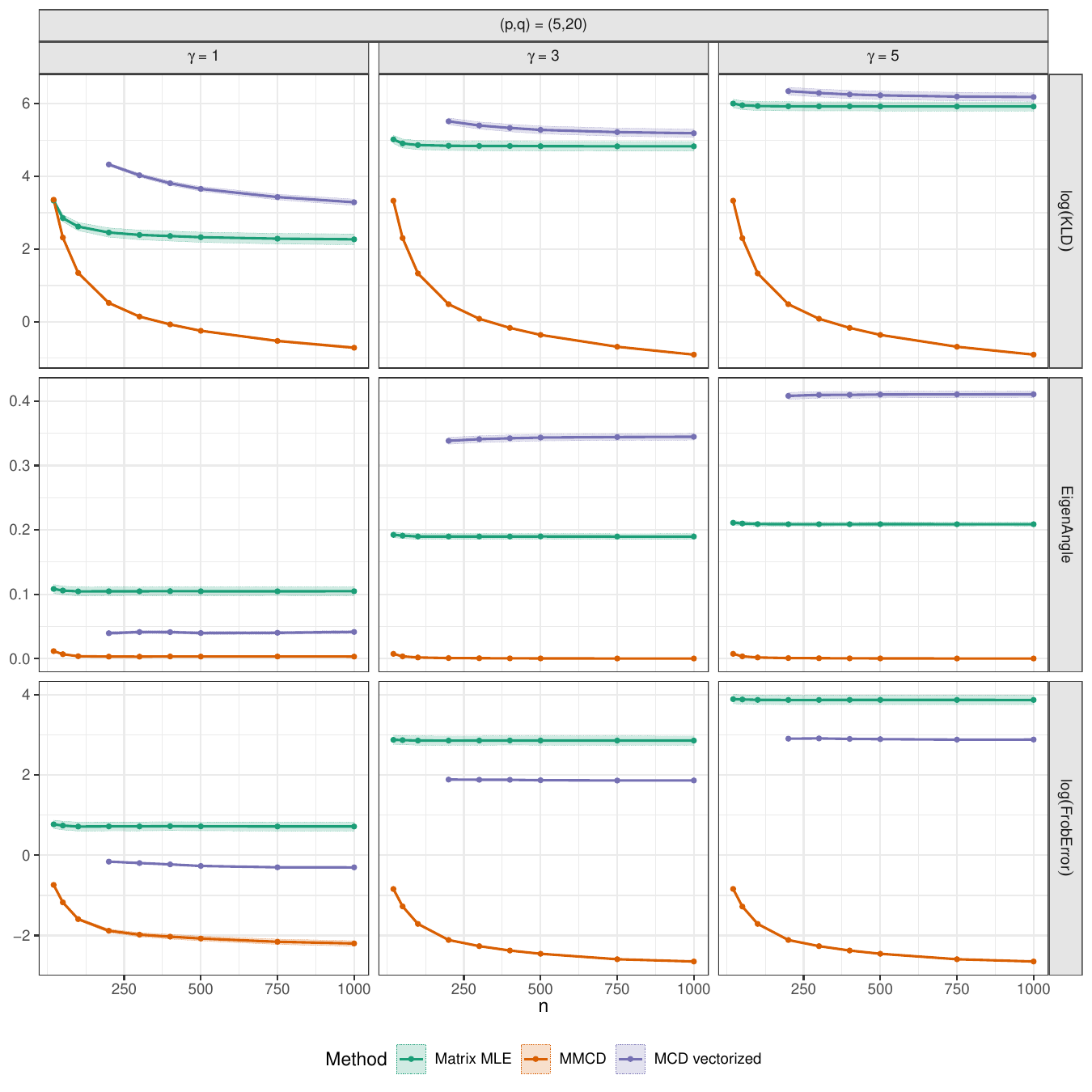}
        \caption{Quality of covariance estimation for simulations with $\varepsilon = 0.2$ and $(p,q) = (5,20)$.}
        \label{fig:simulation_lineplot_low_dim_cov}
    \end{figure}

    \begin{figure}
        \centering
        \includegraphics[width = 1\linewidth]{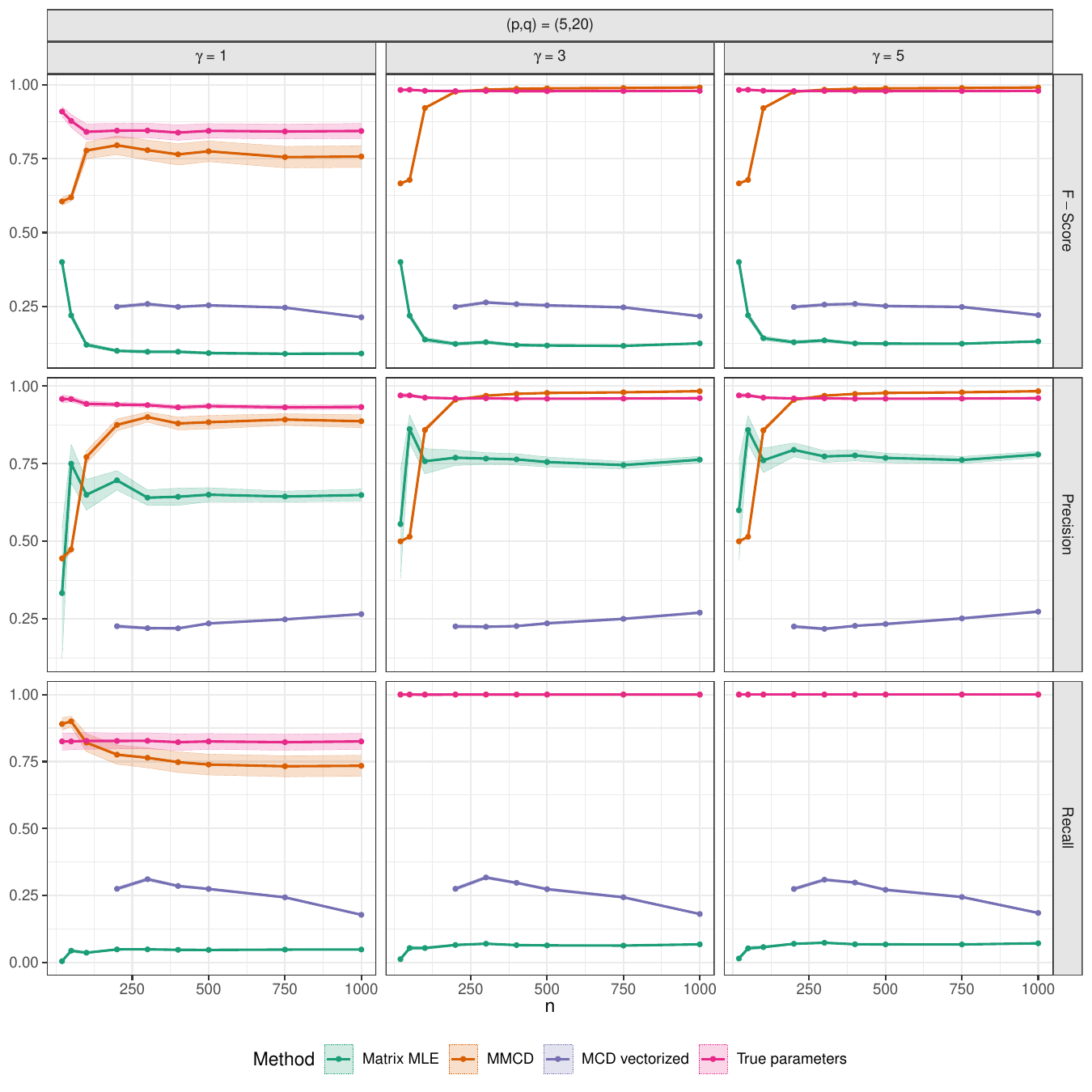}
        \caption{Outlier detection capabilities for simulations with $\varepsilon = 0.2$ and $(p,q) = (5,20)$.}
        \label{fig:simulation_lineplot_low_dim_outlier}
    \end{figure}    
    
    \begin{figure}
        \centering
        \includegraphics[width = 1\linewidth]{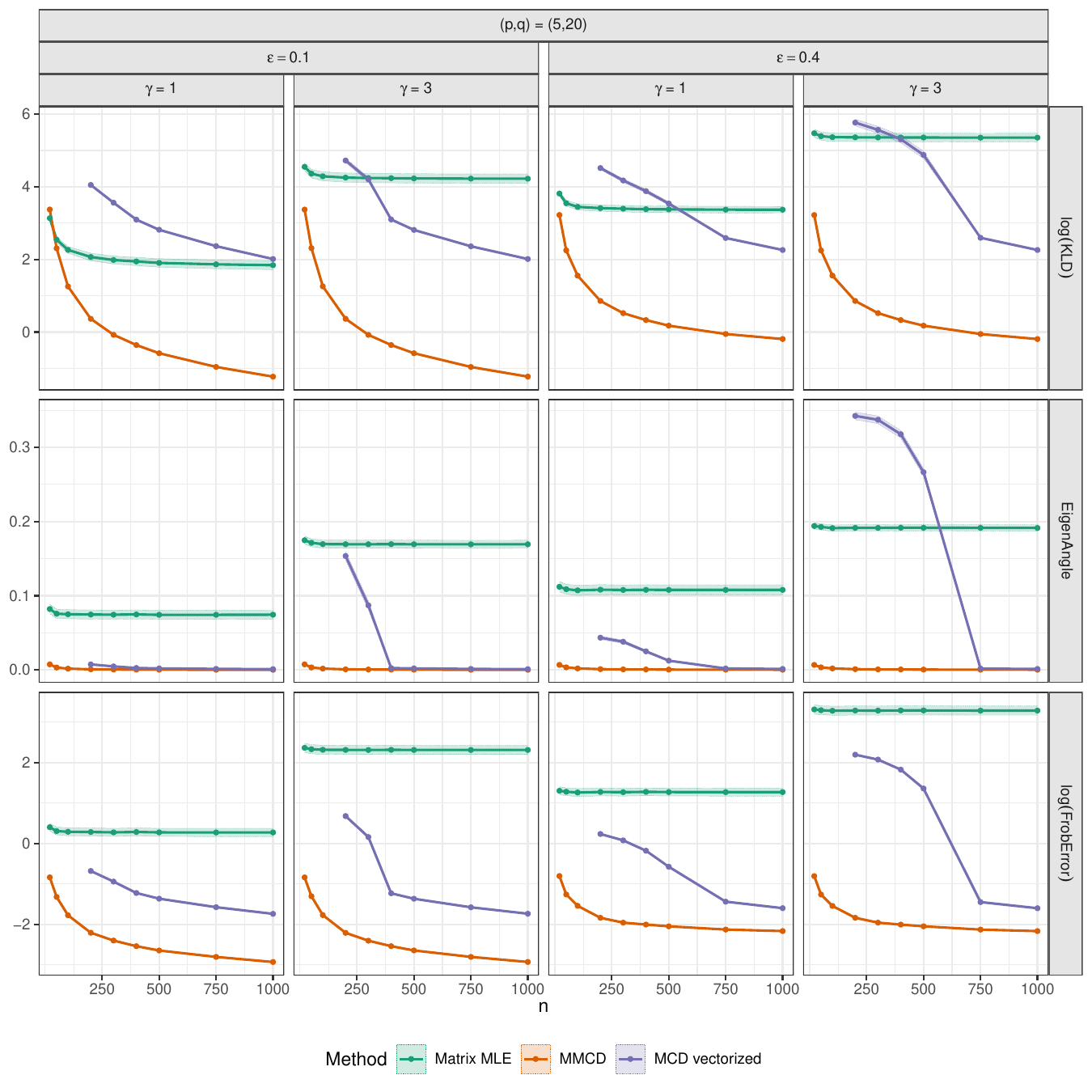}
        \caption{Quality of covariance estimation for simulations where the covariance of the outliers is scaled by $s = 2$.}
        \label{fig:simulation_lineplot_low_dim_radial_cov}
    \end{figure}

    \begin{figure}
        \centering
        \includegraphics[width = 1\linewidth]{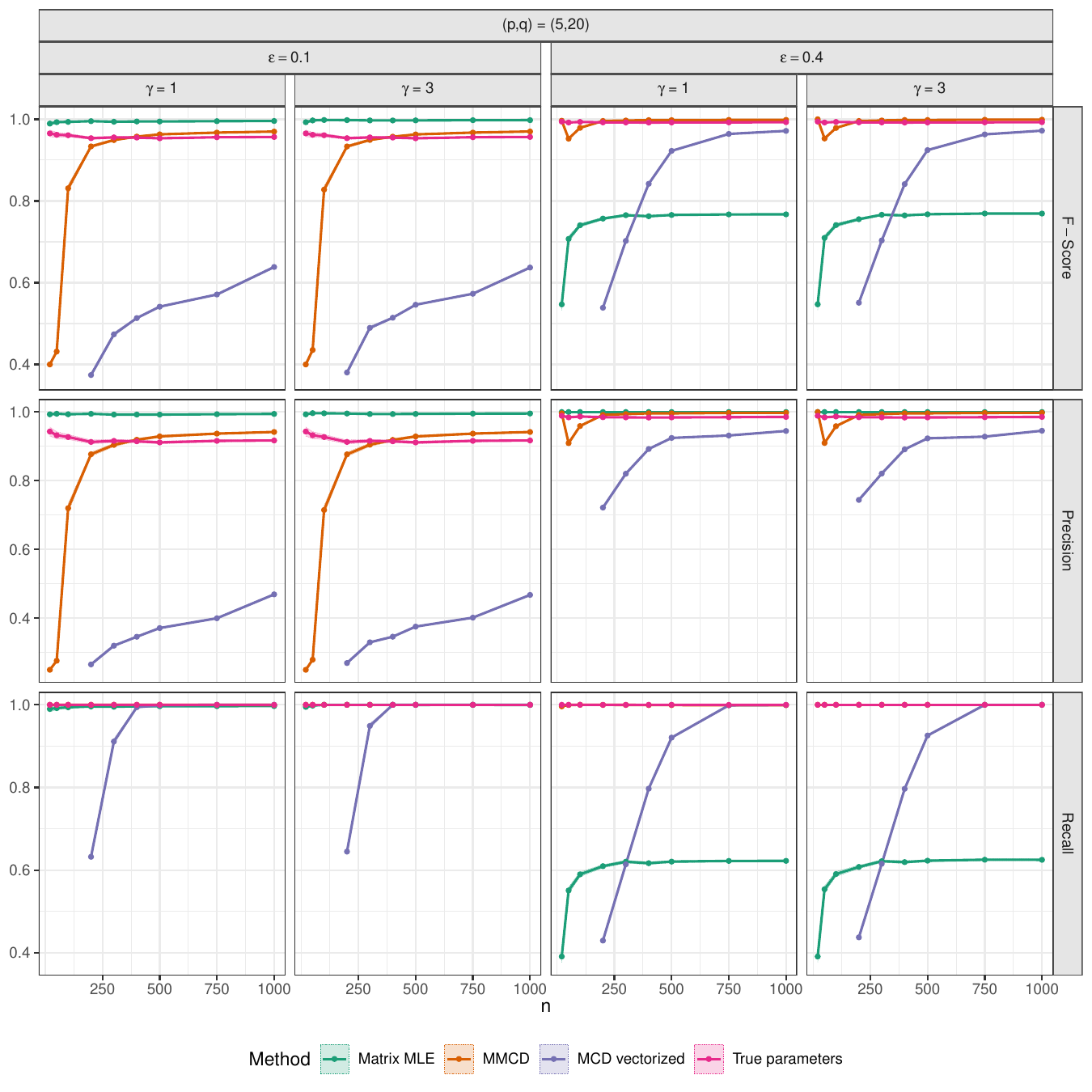}
        \caption{Outlier detection capabilities for simulations where the covariance of the outliers is scaled by $s = 2$.}
        \label{fig:simulation_lineplot_low_dim_radial_outlier}
    \end{figure}

    \begin{figure}
        \centering
        \includegraphics[width = 1\linewidth]{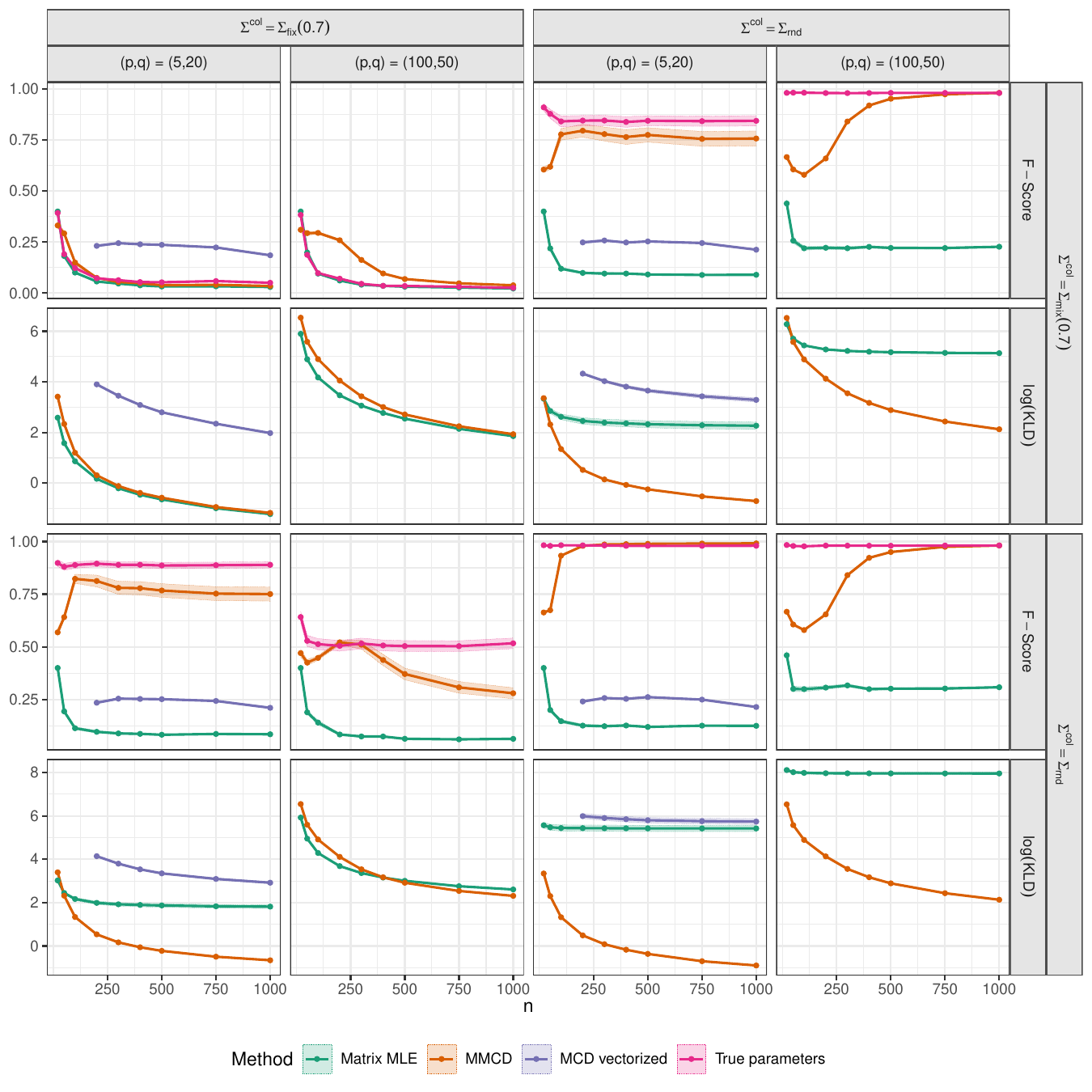}
        \caption{F-score and logarithm of KL divergence comparing 4 different combinations of row- and columnwise covariance matrices, $\gamma = 1$, and $\varepsilon = 0.2$.}
        \label{fig:simulation_lineplot_compare_cov_small_shift}
    \end{figure}

    \begin{figure}
        \centering
        \includegraphics[width = 1\linewidth]{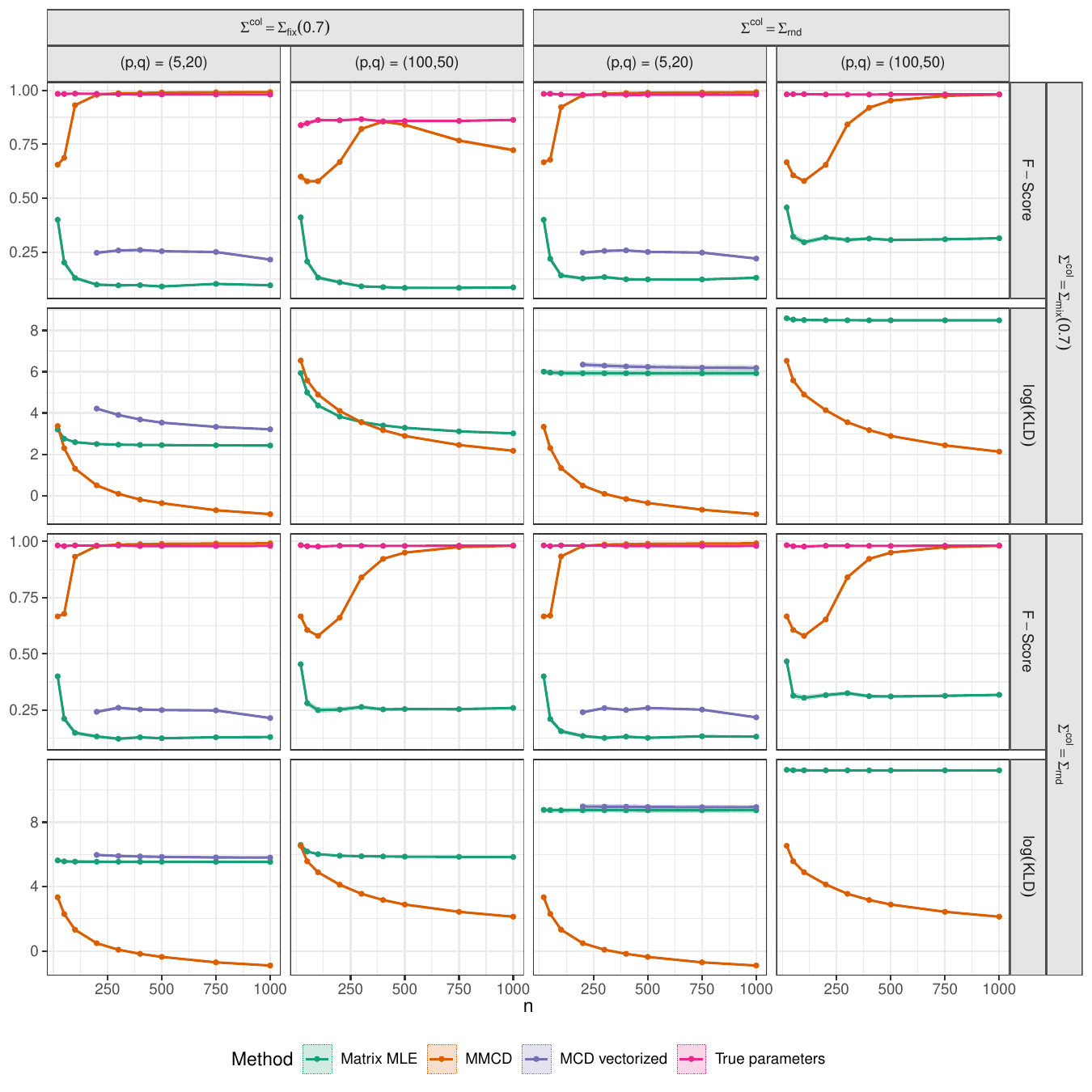}
        \caption{F-score and logarithm of KL divergence comparing 4 different combinations of row- and columnwise covariance matrices, $\gamma = 5$, and $\varepsilon = 0.2$.}
        \label{fig:simulation_lineplot_compare_cov_large_shift}
    \end{figure}

    \begin{figure}
        \centering
        \includegraphics[width = 1\linewidth]{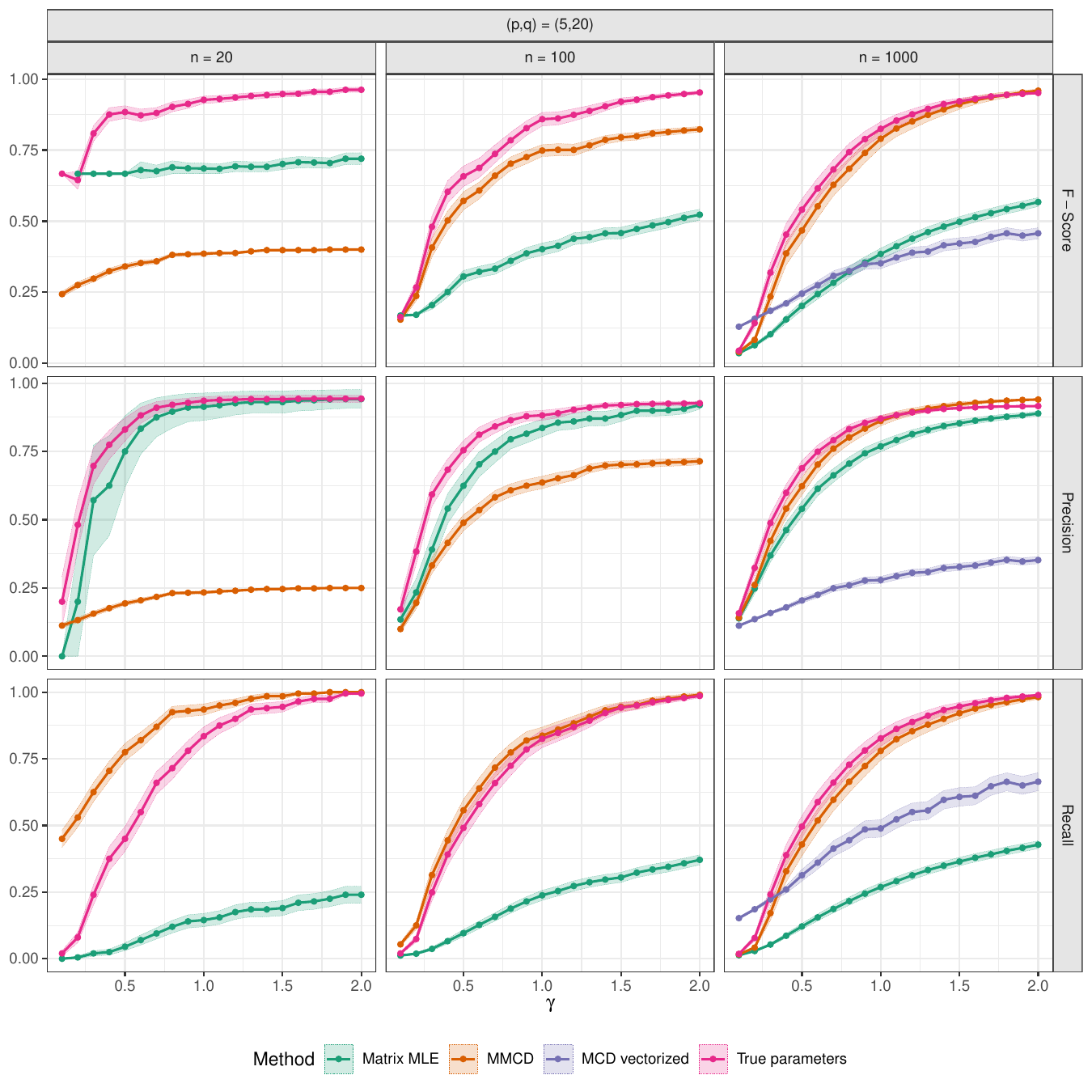}
        \caption{Outlier detection capabilities for simulations with mean shift $\gamma \in \{0.1,0.2,\dots, 2\}$ for $n \in \{20,100,1000\}$, $\varepsilon = 0.1$.}
        \label{fig:simulation_low_gamma_line2}
    \end{figure}

    \begin{figure}
        \centering
        \includegraphics[width = 1\linewidth]{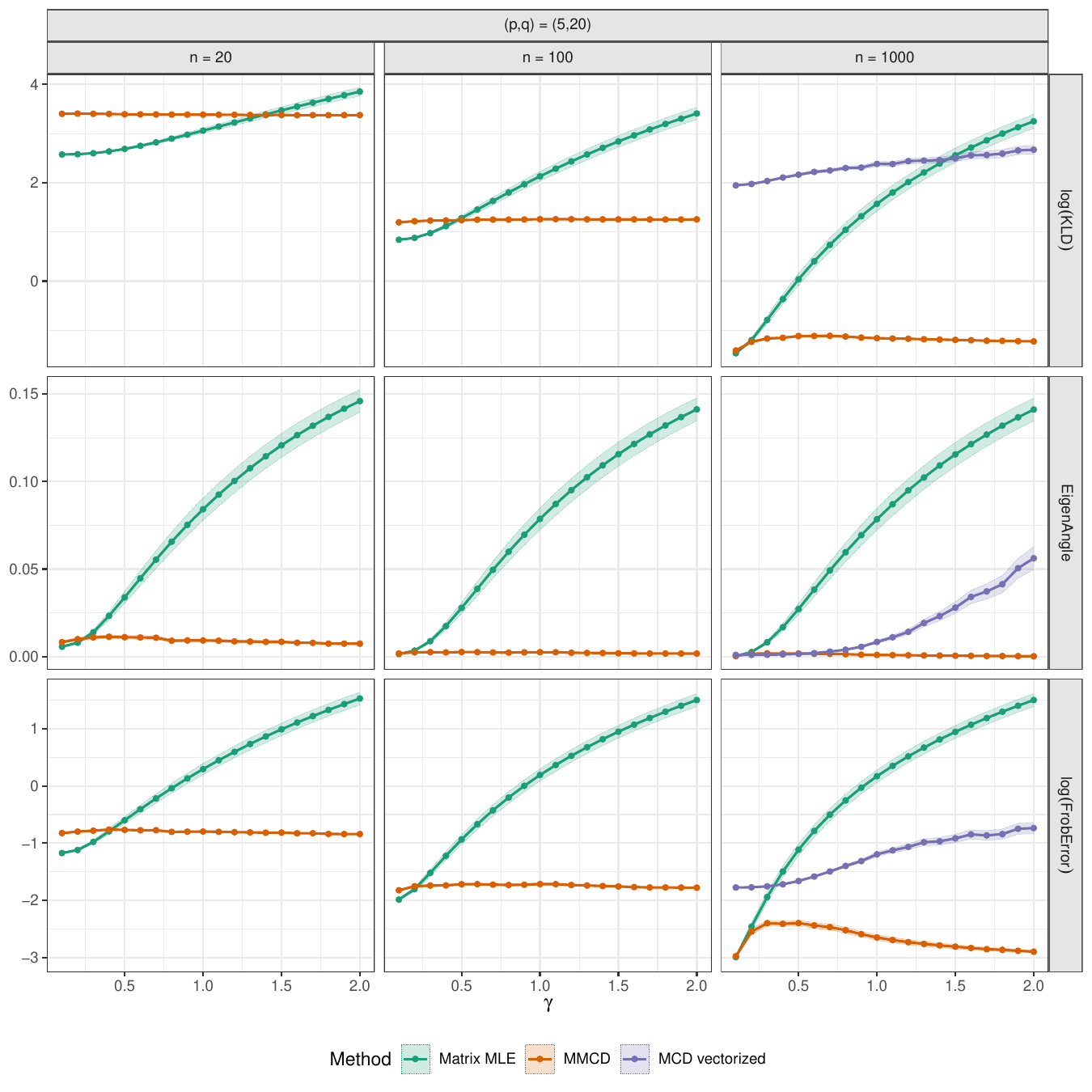}
        \caption{Quality of covariance estimation for simulations with mean shift $\gamma \in \{0.1,0.2,\dots, 2\}$ for $n \in \{20,100,1000\}$, $\varepsilon = 0.1$.}
        \label{fig:simulation_low_gamma_line3}
    \end{figure}

    \begin{figure}
        \centering
        \includegraphics[width = 1\linewidth]{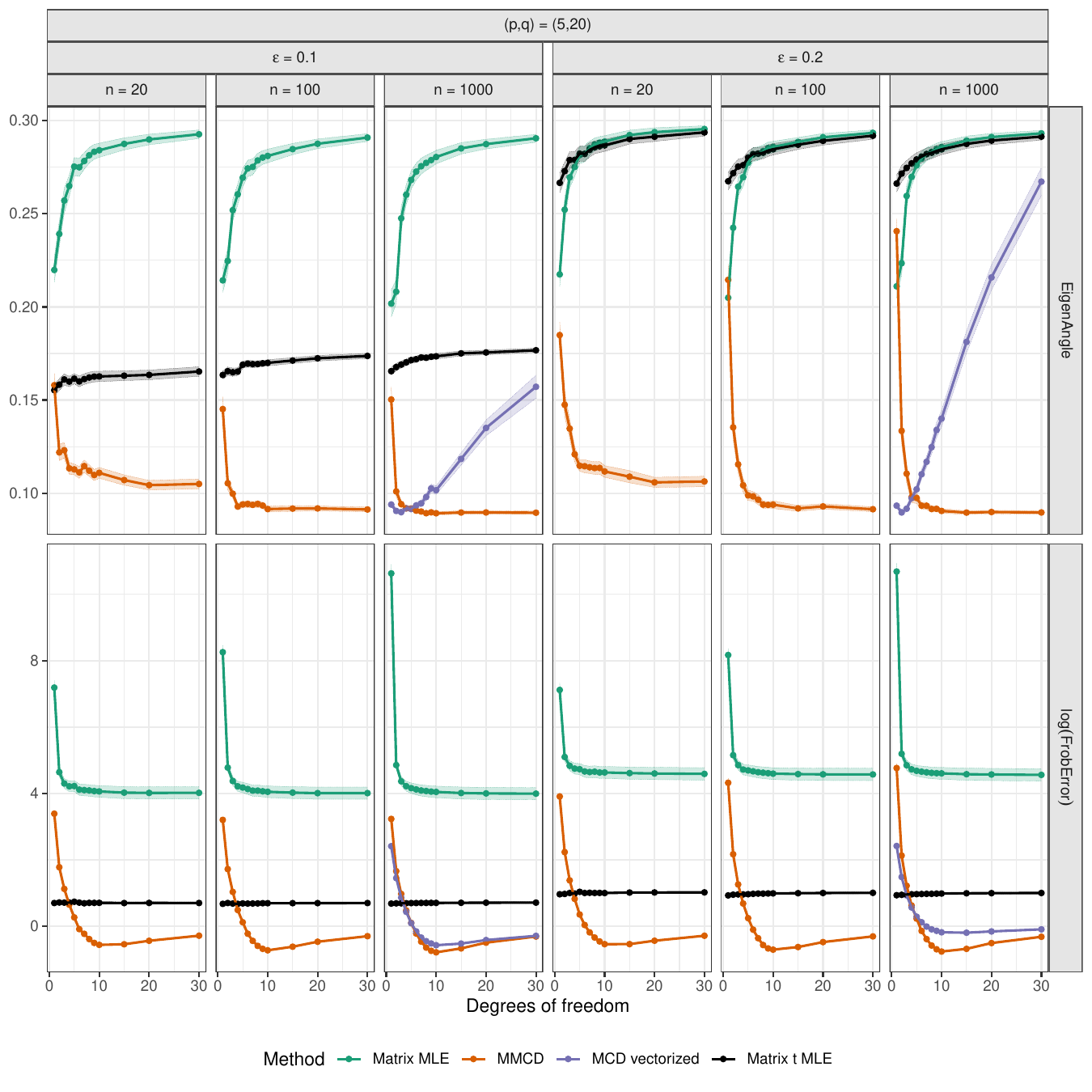}
        \caption{Precicion, recall, eigen angle and logarithm of relative Frobenius error of samples from a contaminated t-distribution with $\nu \in \{1,\dots,30\}$ degrees of freedom for $n \in \{20,100,1000\}$, $\gamma = 1$, $\varepsilon \in \{0.1,0.2\}$.}
        \label{fig:simulation_t_distribuiton}
    \end{figure}
\clearpage

\end{document}